\documentclass[twoside, 11pt]{article}
\usepackage{amsmath, amssymb, amsthm}
\usepackage{xcolor}
\usepackage{colortbl}
\usepackage{array}
\usepackage{pifont}
\usepackage{multirow}
\usepackage{graphicx}
\usepackage{tabularx} 
\usepackage{algorithm, algorithmic}
\usepackage{booktabs}
\usepackage{makecell} 
\usepackage{enumitem}
\usepackage{natbib}
\usepackage{mathtools}
\usepackage{subcaption}
\usepackage{tikz}
\usetikzlibrary{positioning, shapes.geometric, arrows.meta, backgrounds}
\usepackage{thmtools}
\newtheorem{assumption}{Assumption}
\usepackage{tikz}
\usetikzlibrary{arrows.meta, positioning}
\usepackage{xcolor}
\usepackage{microtype}
\usepackage[T1]{fontenc}
\usepackage{times}
\newtheorem{discussion}{Discussion}
\usepackage[normalem]{ulem}

\newcommand{\NH}[2][black]{\textcolor{#1}{#2}}

\usepackage{jmlr2e}

\numberwithin{theorem}{section}
\numberwithin{proposition}{section}
\numberwithin{corollary}{section}
\numberwithin{lemma}{section}
\numberwithin{assumption}{section}

\theoremstyle{plain}
\theoremstyle{definition}

\ShortHeadings{Slot-based Flows for Trans-Dimensional Inference}{Houba, Giarda and Speri}
\firstpageno{1}

\makeatletter
\newcommand{\plainfirstpage}{
  \def\@oddhead{} 
  \def\@evenhead{} 
  \let\@mkboth\@gobbletwo
}

\makeatother

\begin{document}

\title{SlotFlow: Amortized Trans-Dimensional Inference with Slot-Based Normalizing Flows}

\author{\name Niklas Houba \email nhouba@phys.ethz.ch \\
\addr Institute for Particle Physics and Astrophysics \\
ETH Zurich\\
Zurich, Switzerland
\AND
\name Giovanni Giarda \email ggiarda@phys.ethz.ch \\
\addr Institute for Particle Physics and Astrophysics \\
ETH Zurich\\
Zurich, Switzerland
\AND
\name Lorenzo Speri \email lorenzo.speri@esa.int \\
\addr European Space Agency (ESA) \\ European Space Research and Technology
Centre (ESTEC) \\
Noordwijk, Netherlands
}

\editor{N/A}

\maketitle
\plainfirstpage\thispagestyle{plain} 


\begin{abstract}
Inferring the number of distinct components contributing to an observation, while simultaneously estimating their parameters, remains a long-standing challenge across signal processing, astrophysics, and neuroscience.
Classical trans-dimensional Bayesian methods such as Reversible Jump Markov Chain Monte Carlo (RJMCMC) provide asymptotically exact inference but 
can be computationally expensive. 
Instead,
modern deep learning approaches 
provide a faster alternative to inference but
typically assume fixed component counts, sidestepping the core challenge of trans-dimensionality.
To address this, we introduce \emph{SlotFlow}, a 
deep learning architecture for trans-dimensional amortized inference.
The architecture processes time-series observations, which we represent jointly in the frequency and time domains through parallel encoders.
A classifier produces a distribution over component counts $K$, and its MAP estimate specifies the number of slots instantiated, resulting in $O(K)$ computational scaling.
Each slot is parameterized by a shared conditional normalizing flow trained via permutation-invariant Hungarian matching.
On sinusoidal decomposition with up to 10 overlapping components and Gaussian noise,
\emph{SlotFlow} achieves 99.85\% cardinality accuracy and well-calibrated parameter posteriors, with systematic biases well below one posterior standard deviation. Direct comparison with RJMCMC demonstrates agreement for amplitude and the phase parameters: Wasserstein distances $W_2 < 0.01$ for amplitude and $< 0.03$ for phase show that shared global context effectively captures inter-component structure despite factorized posterior approximation.
Frequency posteriors remain correctly centered but exhibit 2--3$\times$ broader credible intervals than RJMCMC, reflecting an encoder information bottleneck in preserving long-baseline phase coherence -- a fundamental challenge that we analyze and discuss for future architectural improvements.
The model achieves $10^6\times$ speedup over RJMCMC in our example, indicating that the approach may be applicable to time-critical workflows in gravitational-wave astronomy, neural spike sorting, and object-centric vision.
\end{abstract}

\vspace{5pt}

\section{Introduction}
\label{sec:intro}


Inferring both the number of components contributing to an observation and their parameters is a fundamental challenge across signal processing, computer vision, and the experimental sciences. In future spaceborne gravitational-wave astronomy, researchers must detect and characterize an unknown number of overlapping binary coalescence signals \citep{couvares2021gravitationalwavedataanalysis}. In neural spike sorting, voltage traces must be decomposed into spike trains from an indeterminate number of concurrently firing neurons \citep{REY2015106}. In spectral analysis, mixed signals require separation into constituent sources whose count remains uncertain \textit{a priori}. 

Formally, these problems demand inference over a \emph{trans-dimensional} joint posterior
\begin{equation}
p(K, \Theta \mid x),
\end{equation}
where $x$ is observed data, $K \in \mathbb{N}$ represents the unknown component count, and $\Theta = \{\theta_k\}_{k=1}^K$ denotes component-specific parameters inhabiting a union of spaces with varying dimensionality: $\mathcal{S} = \bigcup_{K=0}^\infty \{K\} \times \mathbb{R}^{D \cdot K}$. This structure introduces profound challenges: variable dimensionality complicates model specification and posterior exploration; permutation invariance creates $K!$ symmetric posterior modes that trap inference algorithms; and inter-component correlations produce high multimodality beyond simple label-switching.

\vspace{5pt}
\textbf{Classical methods provide rigor at prohibitive cost.} 
\NH{Classical Bayesian approaches to model uncertainty can be unified under the framework of \emph{product space sampling} \citep{product_space_sampling}, with many existing algorithms shown to arise as special cases \citep{Godsill2001}.} Reversible Jump MCMC (RJMCMC) \citep{green1995reversible} enables dimension-changing transitions through auxiliary variables and Jacobian corrections, yielding asymptotically exact inference. Yet effective proposals require deep problem-specific insight, and even with expert tuning, convergence often demands hours to days. The \emph{Gravitational-Wave International Committee} 3G \emph{Data Analysis Computing Challenges} report \citep{couvares2021gravitationalwavedataanalysis} notes that ``the fastest analyses take days, and the longest take weeks'' -- latencies incompatible with next-generation detector requirements. Sequential Monte Carlo (SMC) samplers \citep{doucet2001sequential} offer population-based alternatives but suffer particle degeneracy in high dimensions, requiring thousands of particles and extensive hyperparameter tuning. Dirichlet Process Mixtures (DPM) \citep{rasmussen2000infinite} elegantly sidestep cardinality specification through infinite priors but still depend on slow MCMC or mean-field variational approximations that systematically underestimate uncertainty \citep{ranganath2016hierarchical}.

\vspace{5pt}
\textbf{Modern amortized methods trade generality for efficiency.} Variational autoencoders \citep{kingma2014auto} enable millisecond-scale inference by training recognition networks that map observations to posterior parameters in a single forward pass. Normalizing flows \citep{rezende2015variational} provide expressive posterior families with exact density evaluation. Yet extending these approaches to trans-dimensional settings introduces architectural challenges: networks must produce variable-length outputs, enforce permutation invariance, and quantify uncertainty over $K$ itself. Early attempts imposed fixed maximum capacities with masking \citep{burgess2019monet,eslami2018neural}, conflating model-order uncertainty with architectural choices and wasting computation on unused slots. Permutation-equivariant architectures \citep{zaheer2017deep,lee2019set} addressed symmetry but typically assumed known cardinality. Recent work by Davies et al. \citep{davies2025amortizedvariationaltransdimensionalinference} introduced \emph{CoSMIC} flows, achieving trans-dimensional inference through dimension saturation -- padding all models to maximum dimension $d_{\max}$ with context-dependent masking. While providing formal convergence guarantees, this incurs $O(K_{\max})$ computational cost regardless of true cardinality $K$.


\subsection{Our Contribution: \emph{SlotFlow}}
We introduce \emph{SlotFlow}, a dual-stream model combining parallel frequency- and time-domain encoders for amortized trans-dimensional inference, achieving $O(K)$ computational scaling and millisecond-latency posterior estimation with calibrated uncertainties. \NH{Frequency- and time-domain representations provide complementary information: for signals that are compact in frequency (e.g., overlapping sinusoids), spectral features highlight global structure useful for cardinality estimation, while for signals localized in time (e.g., coalescing binaries), the temporal representation offers clearer cues for separating components through their differing coalescence times. In both regimes, time-domain features contribute local dynamics that support accurate parameter recovery.}

The architecture operates in three stages. First, the \emph{dual-stream encoder} processes frequency- and time-domain inputs through parallel convolutional-attention pipelines, with a shared latent bottleneck enforcing cross-domain consistency. A classifier leverages frequency features to predict $q_{\phi}(K \mid x)$, exploiting spectral peaks as evidence for component multiplicity. Second, \emph{dynamic slot allocation} instantiates exactly $K$ slot contexts by combining global embeddings with learned orthogonal identifiers, achieving efficiency without wasted capacity while maintaining differentiability. Third, \emph{shared conditional flows} -- specifically coupling-based architectures with rational-quadratic spline transforms -- parameterize per-slot marginal posteriors, trained via permutation-invariant Hungarian matching.

On synthetic sinusoidal decomposition with up to 10 overlapping components under severe spectral crowding, \emph{SlotFlow} achieves 99.85\% cardinality accuracy. Despite factorized posterior approximation, direct comparison with RJMCMC shows excellent agreement for amplitude ($W_2 < 0.01$) and good agreement for phase ($W_2 < 0.03$), demonstrating that shared global context effectively captures inter-component structure for these parameters. Frequency posteriors, while centered correctly, remain broader than RJMCMC due to encoder compression of long-baseline coherent phase information. The model exhibits well-calibrated uncertainties with absolute calibration bias below 3\%, gracefully degrading under increased noise. By trading some posterior correlation fidelity and frequency precision for orders-of-magnitude speedup, \emph{SlotFlow} offers a practical path toward scalable trans-dimensional Bayesian inference.

\paragraph{Contributions.}
\begin{itemize}[leftmargin=*,topsep=2pt,itemsep=1pt]
    \item A slot-based amortized architecture predicting $q_{\phi}(K \mid x)$ and parameterizing per-slot marginals via shared conditional flows
    \item Permutation-invariant training through Hungarian-matched negative log-likelihood with phase and frequency weighting
    \item Theoretical analysis of symmetry, factorization, and sample complexity properties
    \item Demonstration of competitive accuracy for amplitude and phase parameters with orders-of-magnitude speedup over RJMCMC, with analysis of frequency estimation limitations due to encoder information bottlenecks
    \item Open-source implementation at \url{https://github.com/nhouba/slotflow-inference}
\end{itemize}


\subsection{Related Work}

Trans-dimensional inference lies at the intersection of classical Monte Carlo methods and modern amortization. Table~\ref{tab:comparison} positions \emph{SlotFlow} relative to representative approaches.

\vspace{5pt}
\textbf{Classical methods.} RJMCMC \citep{green1995reversible} and SMC samplers \citep{doucet2001sequential} provide asymptotic exactness but face poor scaling: dimension-changing proposals require expert tuning and problem-specific designs, while particle degeneracy in high-dimensional spaces necessitates thousands of samples and careful resampling strategies. \NH{Other trans-dimensional approaches include diffusive nested sampling \citep{BrewerTrDimDiffusiveNestedSampling}, continuous-time MCMC \citep{Astorino:2025ccl}, and product space sampling frameworks \citep{product_space_sampling}, all of which offer elegant formulations of model uncertainty but remain computationally demanding in practice.}

\vspace{5pt}
\textbf{Amortized inference.} Deep Sets \citep{zaheer2017deep} and Set Transformers \citep{lee2019set} enable permutation-invariant processing through attention mechanisms but assume fixed cardinality, requiring \textit{a priori} specification of component count. Set Flow \citep{rasul2019setflowpermutationinvariant} combines equivariant coupling layers with normalizing flows for generative modeling of sets with known size, achieving permutation equivariance through architectural constraints. MONet \citep{burgess2019monet} fixes maximum capacity with spatial masking, introducing systematic distortions when true cardinality differs from allocated slots. \emph{CoSMIC} \citep{davies2025amortizedvariationaltransdimensionalinference} achieves trans-dimensional inference through dimension saturation -- padding all models to maximum dimension with context-dependent masking -- providing formal convergence guarantees but at $O(K_{\max})$ computational cost regardless of true $K$.

\vspace{5pt}
\textbf{Simulation-based inference.} Neural Posterior Estimation (NPE) methods \citep{10.5555/3294771.3294894, papamakarios2019sequentialneurallikelihoodfast, pmlr-v130-lueckmann21a} train neural density estimators directly on simulator outputs, enabling likelihood-free inference when forward models are intractable. \NH{Sequential variants (SNPE) further refine posterior accuracy by iteratively adapting the proposal distribution. However, NPE approaches typically assume fixed dimensionality, requiring separate mechanisms for model selection -- such as computing marginal likelihoods or evidences -- rather than jointly inferring both the component count $K$ and parameters $\Theta$ within a unified framework.}

\NH{Recent work extends amortized inference to decision-making: 
\citet{lyu2025dynamicsbiroundfreesequential} develop neural surrogates 
that directly output optimal decisions under posterior uncertainty, 
bypassing explicit posterior sampling. While orthogonal to trans-dimensional 
inference, this demonstrates the broader potential of amortization beyond 
parameter estimation -- including model selection, experimental design, 
and sequential decision-making under model uncertainty.}

\begin{table}[t]
\scriptsize
\setlength{\tabcolsep}{4pt}
\renewcommand{\arraystretch}{0.92}
\centering
\caption{Comparison of \emph{SlotFlow} with representative baselines. Trans-dimensional refers to posterior inference over component count $K$. Equivariant means predictions permute with inputs; invariant means predictions are unchanged. Inter-slot dependencies indicate whether the model captures correlations between components.}
\begin{tabular*}{\textwidth}{@{\extracolsep{\fill}}lcc
    >{\centering\arraybackslash}p{2.2cm}
    >{\centering\arraybackslash}p{2.0cm}
    c@{}}
\toprule
Method & Amortized & \makecell{Trans-\\Dimensional} & Equivariant & \makecell{Inter-Slot\\Dependencies} & Scalability \\
\midrule
RJMCMC \citep{green1995reversible} 
    & No 
    & Yes 
    & Inherits symmetry
    & Yes 
    & Poor \\
SMC \citep{doucet2001sequential} 
    & No 
    & Yes 
    & Inherits symmetry
    & Yes 
    & Poor \\
Set Flow \citep{rasul2019setflowpermutationinvariant} 
    & Yes 
    & No$^\dagger$
    & Yes 
    & Coupling-based 
    & $O(K)$ \\
MONet \citep{burgess2019monet} 
    & Yes 
    & No 
    & No 
    & Partial 
    & $O(K_{\max})$ \\
\emph{CoSMIC} \citep{davies2025amortizedvariationaltransdimensionalinference}
    & Yes
    & Yes
    & Via masking
    & Coupling-based
    & $O(K_{\max})$ \\
\textbf{\emph{SlotFlow}} 
    & \textbf{Yes} 
    & \textbf{Yes} 
    & \textbf{Yes} 
    & \textbf{Context-shared$^\ddagger$} 
    & $\mathbf{O(K)}$ \\
\bottomrule
\end{tabular*}
\vspace{-4pt}
\begin{flushleft}
{\scriptsize 
$\dagger$ Set Flow supports variable set sizes in generation but not trans-dimensional posterior inference. \\ 
$\ddagger$ Slots condition on shared global context via attention; posterior factorizes across slots. \emph{CoSMIC} maintains coupling through masked flows.
}
\end{flushleft}
\label{tab:comparison}
\end{table}


\paragraph{Paper Organization.} 
Section~\ref{sec:problem} formalizes the trans-dimensional inference problem and introduces sinusoidal mixtures as a canonical test case. Section~\ref{sec:architecture} presents the \emph{SlotFlow} architecture, including dual-stream encoding, dynamic slot allocation, and conditional flow parameterization. Section~\ref{sec:theory} develops theoretical properties of permutation invariance, sample complexity, and posterior factorization, with complete proofs in Appendix~\ref{app:proofs}. Section~\ref{sec:experiments} validates the approach through systematic experiments on cardinality accuracy, parameter calibration, posterior quality, and computational cost. Section~\ref{sec:conclusion} concludes with current limitations, perspectives on extensions to time-frequency representations and multi-class mixtures.


\section{Problem Formulation}
\label{sec:problem}

Trans-dimensional Bayesian inference addresses the fundamental challenge of reasoning about models where both the number of components $K$ and their parameters $\Theta = \{\theta_k\}_{k=1}^K$ are uncertain. This setting pervades scientific applications: gravitational-wave astronomy demands resolving overlapping 
sources like in observations from future detectors such as the Laser Interferometer Space Antenna (LISA) mission \citep{colpi2024lisadefinitionstudyreport}, the Einstein Telescope \citep{ET:2025xjr}, and Cosmic Explorer \citep{Reitze:2019iox}, 
neural spike sorting requires decomposing voltage traces into unknown numbers of neuron waveforms \citep{REY2015106}, and spectral analysis seeks separation of mixed signals into constituent sources.

\paragraph{Bayesian Formulation.}
The generative model factorizes as
\begin{equation}
p(x,\Theta,K,\sigma) = p(x \mid K,\Theta,\sigma)\,p(\Theta\mid K)\,p(K)\,p(\sigma),
\end{equation}
\NH{where $x \in \mathcal{X}$ denotes the observed data, with $\mathcal{X}$ the space of time-series observations}, $K \in \mathbb{N}$ the component count, $\Theta$ the parameter set, and $\sigma$ the noise scale. Assuming conditionally independent component priors given $K$, we have $p(\Theta \mid K) = \prod_{k=1}^K p(\theta_k)$ \NH{with $\theta_k \in \mathcal{X}_\theta$, where $\mathcal{X}_\theta$ denotes the component-parameter space.} The joint posterior
\begin{equation}
p(K,\Theta,\sigma \mid x) = \frac{p(x \mid K,\Theta,\sigma)\,p(\Theta\mid K)\,p(K)\,p(\sigma)}{p(x)},
\label{eq:joint_posterior}
\end{equation}
requires marginalizing over the trans-dimensional parameter space $\mathcal{S} = \bigcup_{K\geq 0} \{K\}\times \mathcal{X}_\theta^K$, where the evidence is given by
\begin{equation}
p(x) = \sum_{K=0}^\infty \int_{\mathcal{X}_\theta^K}\!\int_{\mathbb{R}_+} p(x \mid K,\Theta,\sigma)\,p(\Theta\mid K)\,p(K)\,p(\sigma)\,d\Theta\,d\sigma \, .
\end{equation}

\paragraph{Canonical Test Case: Sinusoidal Mixtures.}
We adopt sinusoidal decomposition as our primary benchmark, motivated by its prevalence in signal processing and gravitational-wave astronomy. A time series $x(t)$ on $t\in[0,T]$ is modeled as
\begin{equation}
x(t) = \sum_{k=1}^K A_k \cos(2\pi f_k t + \phi_k) + \epsilon(t), \quad \epsilon(t)\sim \mathcal{N}(0,\sigma^2),
\label{eq:sinusoid_model}
\end{equation}
where amplitudes $A_k > 0$, frequencies $f_k \in [f_{\min}, f_{\max}]$, and phases $\phi_k \in [0, 2\pi)$ parameterize components. This model exemplifies trans-dimensional challenges: (1) unknown $K$, (2) $K!$ permutation symmetry inducing degenerate posterior modes, and (3) strong inter-component correlations for nearby frequencies where spectral overlap prevents independent parameter recovery. The structure generalizes directly to Gaussian mixture models (with location/scale replacing frequency/phase), object detection (positions/scales/orientations), and gravitational-wave analysis (masses/spins/sky locations) -- all sharing additive superposition with unknown cardinality observed through noise.

\vspace{5pt}
\textbf{Gravitational-Wave Motivation.}
The LISA detector will observe tens of thousands of overlapping galactic binaries simultaneously \citep{PhysRevD.101.123021,Katz:2024oqg,Littenberg:2023xpl,Deng:2025wgk}. 
\NH{Existing global-fit pipelines return posterior samples for every source and noise parameter -- typically $\sim10^4$ samples for $\sim10$ parameters across $\sim10^4$ sources -- amounting to $\mathcal{O}(10^9)$ stored values, which becomes cumbersome for catalog-level inference. For example, \texttt{Erebor} \citep{Katz:2024oqg} required roughly seven days on four NVIDIA A100 GPUs and a dozen CPUs to analyze the LISA Data Challenge 2A \citep{LDC_Sangria_2020}.
Although the global fit itself is not latency-critical, fast amortized inference is valuable for \emph{post–global-fit analysis}: exploring large catalogs interactively, rapidly regenerating posterior samples under alternative priors or noise assumptions, and updating localized subsets of sources (e.g., candidate MBHBs) as new data arrive. These tasks benefit from rapid posterior evaluation even if the global fit does not. Amortized trans-dimensional inference thus complements the global-fit pipelines by providing compact posterior representations and enabling responsive, iterative scientific workflows.}

\paragraph{Core Inferential Challenges.}
Four interrelated difficulties distinguish trans-dimensional inference from standard fixed-dimensional problems:

\vspace{5pt}
\textit{Dimension matching.} The parameter space $\mathcal{S}$ is a disjoint union of manifolds with varying dimension. Standard MCMC cannot traverse between subspaces without specialized dimension-matching transformations (RJMCMC's auxiliary variables and Jacobian corrections \citep{green1995reversible}), which require expert tuning and often exhibit poor acceptance rates. \NH{In nested-model settings, the mapping between states may reduce to the linear addition or removal of parameters, in which case the Jacobian simplifies to the ratio of prior volumes \citep{geophys_rjmcmc}. However, this does not resolve the core challenge: designing effective, problem-specific proposal distributions for dimension-changing moves remains difficult and typically demands substantial manual tuning.}

\vspace{5pt}
\textit{Combinatorial multimodality.} Permutation symmetry creates $K!$ equivalent posterior modes. For moderate $K$, this combinatorial explosion produces complex landscapes that trap samplers or cause variational methods to collapse, exacerbating the label-switching problem where posterior summaries become meaningless without post-hoc alignment \citep{Jasra2005,Johnson:2025oyu}.

\vspace{5pt}
\textit{Inter-component correlations.} Components are statistically entangled through the likelihood. For sinusoids, nearby frequencies $f_i \approx f_j$ induce strong posterior dependencies: the model reallocates spectral energy between overlapping components, making factorized variational approximations systematically underestimate uncertainty \citep{Blei03042017}. Capturing these dependencies requires expressive posterior families that scale efficiently with $K$ while respecting permutation symmetry (e.g., exchangeable normalizing flows, set-based variational distributions, or symmetric energy-based models).

\vspace{5pt}
\textit{Information bottlenecks in amortized inference.} Neural encoders must compress high-dimensional observations into fixed-size latent representations for computational tractability. For sinusoidal decomposition, precise frequency estimation requires preserving long-baseline phase coherence: Fisher information for frequency scales as $T^2$, demanding access to global phase evolution $\partial\phi/\partial f$ across the full observation window. Convolutional architectures with stride-based downsampling can inadvertently discard fine-grained spectral structure during compression, creating an information bottleneck where the latent representation loses sub-bin phase curvature critical for Fisher-optimal frequency localization. This encoder-decoder trade-off between computational efficiency and information preservation represents a fundamental challenge distinct from flow expressiveness -- even with arbitrarily powerful decoders, posteriors cannot recover information discarded during encoding. Balancing these competing demands while maintaining $O(K)$ scalability motivates careful architectural choices in dual-stream processing and context representation.

\vspace{5pt}
These challenges motivate specialized architectures capable of navigating trans-dimensional posteriors efficiently while managing information preservation constraints. We now present \emph{SlotFlow}, which addresses these difficulties through dual-stream encoding, dynamic slot allocation, and permutation-invariant conditional flows, while analyzing the trade-offs between computational efficiency and frequency precision.

\section{\emph{SlotFlow} Architecture}
\label{sec:architecture}

\emph{SlotFlow} provides end-to-end amortized trans-dimensional inference through three core principles:  
(1) \emph{dual-stream processing} that extracts complementary frequency- and time-domain representations,  
(2) \emph{dynamic slot allocation} that instantiates a number of posterior components determined by the MAP estimate of $K$, yielding $O(K)$ computational scaling, and (3) \emph{factorized marginals with shared context} that capture global structure while maintaining permutation invariance.

Figure~\ref{fig:architecture} illustrates the full pipeline. Given an observation $x \in \mathbb{R}^T$, \emph{SlotFlow} proceeds in four stages:  
(i) \emph{dual-stream encoding} via parallel convolutional–attention pathways operating in frequency and time domains;  
(ii) \emph{cardinality estimation} inferring $q_\phi(K \mid x)$, leveraging spectral structure for cardinality cues;  
(iii) \emph{slot context generation} constructing $K$ conditioning vectors by combining global embeddings with orthogonal slot identifiers; and  
(iv) \emph{conditional flow inference} producing per-slot marginals $q_\phi(\theta_k \mid x, k)$ using a shared rational–quadratic spline flow trained with a Hungarian-matched permutation-invariant negative log-likelihood loss.

\begin{figure}[h!]
\centering
\begin{tikzpicture}[
    scale=1.14,
    node distance=0.9cm and 2.1cm,
    freqbox/.style={rectangle, draw=blue!55!black, fill=blue!7!white,
        line width=0.9pt, rounded corners=2pt,
        text width=2.2cm, align=center, minimum height=0.9cm,
        font=\footnotesize\sffamily},
    timebox/.style={rectangle, draw=orange!65!black, fill=orange!10!white,
        line width=0.9pt, rounded corners=2pt,
        text width=2.2cm, align=center, minimum height=0.9cm,
        font=\footnotesize\sffamily},
    fusionbox/.style={rectangle, draw=purple!60!black, fill=purple!8!white,
        line width=1pt, rounded corners=3pt,
        text width=2.8cm, align=center, minimum height=1cm,
        font=\footnotesize\sffamily},
    poolbox/.style={rectangle, draw=gray!55, fill=white,
        line width=0.8pt, rounded corners=2pt,
        text width=1.9cm, align=center, minimum height=0.75cm,
        font=\footnotesize\sffamily},
    inputbox/.style={rectangle, draw=gray!70!black, fill=gray!5,
        line width=1pt, rounded corners=3pt,
        text width=2.2cm, align=center, minimum height=0.9cm,
        font=\footnotesize\sffamily\bfseries},
    outputbox/.style={ellipse, draw=green!60!black, fill=green!6!white,
        line width=1pt, text width=2cm, align=center, minimum height=0.9cm,
        font=\footnotesize\sffamily\bfseries},
    arrow/.style={->, >=stealth, line width=0.8pt, draw=gray!70},
    control/.style={->, >=stealth, line width=1pt, draw=red!70!black, dashed},
    stage/.style={font=\footnotesize\bfseries\sffamily, text=gray!70}
]

\fill[gray!5] (-7.5,1.15) rectangle (3.9,-13.9);
\node[inputbox] (input) at (-0.9,0.5) {Input Signal\\$x \in \mathbb{R}^T$};

\node[freqbox] (fft) at (-3.8,-1.5) {FFT\\$\mathbf{x}_{\text{freq}}$};
\node[timebox] (raw) at (1.8,-1.5) {Raw Signal\\$x$};
\node[freqbox] (enc_freq) at (-3.8,-2.9) {Conv Encoder\\$\mathcal{E}_{\text{freq}}$};
\node[timebox] (enc_time) at (1.8,-2.9) {Conv Encoder\\$\mathcal{E}_{\text{time}}$};
\node[freqbox] (attn_freq) at (-3.8,-4.3) {PE + MHA\\$\tilde{\mathbf{h}}_{\text{freq}}$};
\node[timebox] (attn_time) at (1.8,-4.3) {PE + MHA\\$\tilde{\mathbf{h}}_{\text{time}}$};

\node[poolbox] (pool_cls) at (-5.3,-5.9) {Pool\\$\mathbf{z}_{\text{cls}}$};
\node[poolbox] (pool_freq) at (-2.3,-5.9) {Pool\\$\mathbf{z}_{\text{freq}}$};
\node[poolbox] (pool_time) at (1.8,-5.9) {Pool\\$\mathbf{z}_{\text{time}}$};
\node[freqbox] (cls_mha) at (-5.3,-7.4) {Global MHA\\$\mathbf{h}_{\text{cls}}$};
\node[freqbox] (classifier) at (-5.3,-8.9) {Classifier\\$q_\phi(K|x)$};

\node[fusionbox] (fusion) at (-0.2,-7.9) {Concat + MLP\\Global Context $\mathbf{g}$};
\node[fusionbox] (slots) at (-0.2,-9.7) {Slot Contexts\\$\mathbf{c}_k = [\mathbf{g}, \mathbf{s}_k]$\\$k=1,\ldots,\hat{K}$};

\node[fusionbox] (flow) at (-0.2,-11.4) {Conditional Flow\\$\theta_k = f_\phi^{-1}(z; \mathbf{c}_k)$};
\node[outputbox] (k_out) at (-5.3,-10.7) {$\hat{K}$\\Cardinality};
\node[outputbox] (theta_out) at (-0.2,-13.0) {$\{q_\phi(\theta_k|x,k)\}_{k=1}^{\hat{K}}$\\Posteriors};

\draw[arrow] (input) -- ++(0,-0.9) -| node[pos=0.2, above, font=\scriptsize\sffamily] {} (fft);
\draw[arrow] (input) -- ++(0,-0.9) -| node[pos=0.2, above, font=\scriptsize\sffamily] {} (raw);
\foreach \a/\b in {fft/enc_freq, enc_freq/attn_freq, raw/enc_time, enc_time/attn_time}
    \draw[arrow] (\a) -- (\b);

\draw[arrow] (attn_freq) -- ++(0,-0.75) -| (pool_cls);
\draw[arrow] (attn_freq) -- ++(0,-0.75) -| (pool_freq);
\draw[arrow] (attn_time) -- ++(0,-0.6) -| (pool_time);
\draw[arrow] (pool_cls) -- (cls_mha);
\draw[arrow] (cls_mha) -- (classifier);
\draw[arrow] (classifier) -- (k_out);

\draw[arrow] (pool_freq) |- (fusion);
\draw[arrow] (pool_time) |- (fusion);
\draw[arrow] (fusion) -- (slots);
\draw[arrow] (slots) -- (flow);
\draw[arrow] (flow) -- (theta_out);

\draw[control] (k_out.east) to[out=0, in=180] node[pos=0.5, above, font=\small\bfseries\sffamily, red!70!black] {$\hat{K}\,\,$} (slots.west);

\node[stage, rotate=90, anchor=south, align=center] at (-5.1, -2.9) {Stage I:\\Frequency Encoding};
\node[stage, rotate=90, anchor=south, align=center] at (3.5, -2.9) {Time Encoding};
\node[stage, rotate=90, anchor=south, align=center] at (-6.5, -8.1) {Stage II:\\Counting};
\node[stage, rotate=90, anchor=south, align=center] at (2.45, -9.7) {Stage III:\\Dynamic \\ Slot  Allocation};
\node[stage, rotate=90, anchor=south, align=center] at (2.3, -11.8) {Stage IV:\\Inference};

\end{tikzpicture}
\caption{\textbf{\emph{SlotFlow} Architecture.} The model processes input signals through four stages: (I) \emph{Dual-stream encoding} (FFT/time) via convolutional encoders and multi-head attention; (II) \emph{Cardinality estimation} via pooled global features; (III) \emph{Slot context generation} fusing global and slot-specific embeddings; and (IV) \emph{Conditional flow inference} producing per-component posteriors. The red control arrow denotes how the predicted $\hat{K}$ dynamically sets the number of slot contexts.}
\label{fig:architecture}
\end{figure}
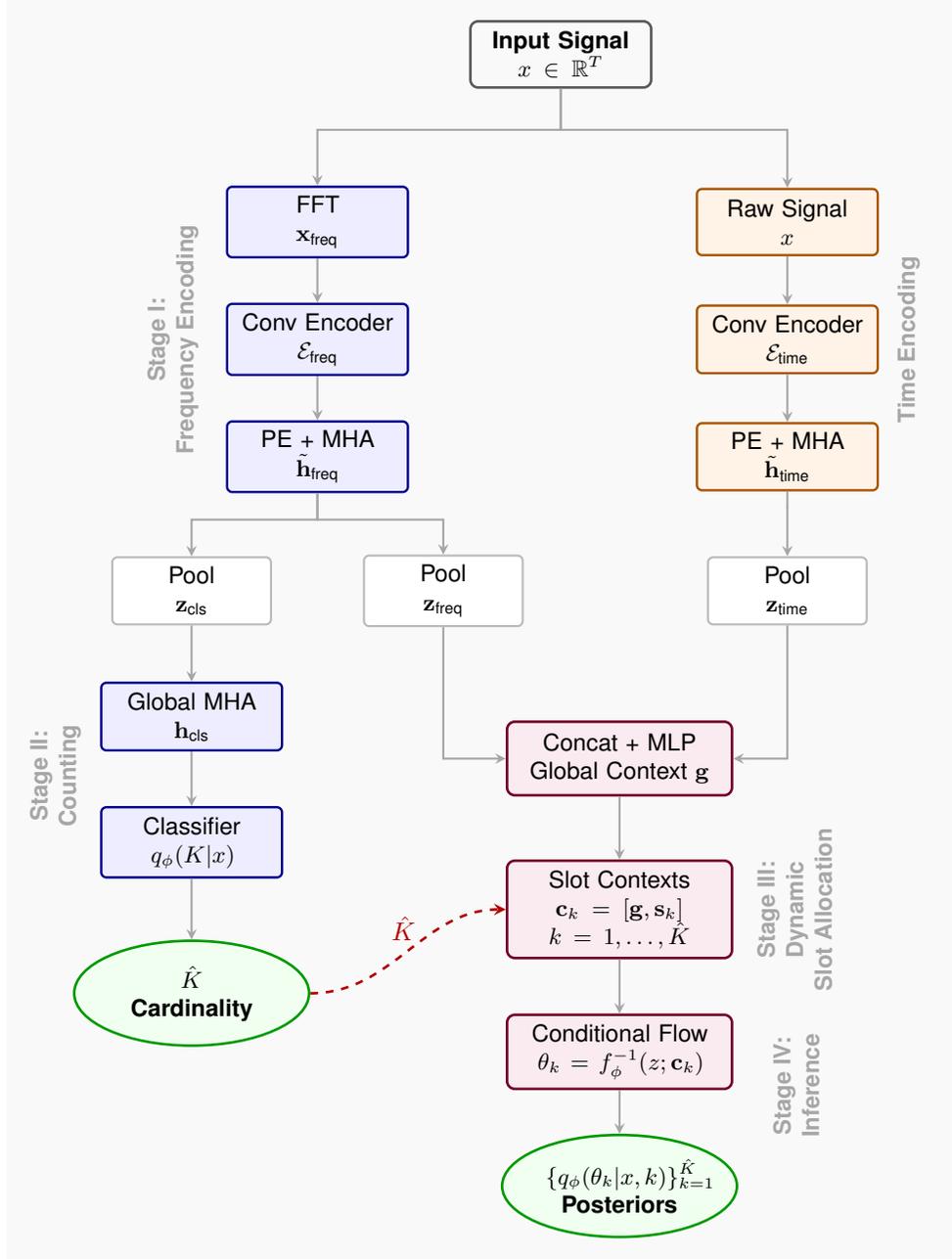

\subsection{Dual-Stream Feature Extraction}
\label{sec:dual_stream}

Separate frequency- and time-domain processing pathways capture complementary signal aspects: frequency analysis reveals harmonic structure and spectral peaks provide cardinality, while time-domain processing preserves transient phenomena and amplitude modulations obscured by spectral windowing.

\paragraph{Frequency Pathway.} 
For input $x \in \mathbb{R}^T$, we compute the discrete Fourier transform $X[f] = \text{FFT}(x) \in \mathbb{C}^{T/2+1}$ with orthonormal normalization, representing the complex spectrum as two-channel real-valued tensor $\mathbf{x}_{\text{freq}} = [\text{Re}(X), \text{Im}(X)] \in \mathbb{R}^{2 \times (T/2+1)}$. A three-layer convolutional encoder $\mathcal{E}_{\text{freq}}$ with GELU activations progressively expands channels (2→32→64→128). The first two layers use kernel size 8 with stride 2; the third layer uses kernel size 5 with stride 1, preserving approximately twice the spatial resolution for improved spectral detail, yielding $\mathbf{h}_{\text{freq}}$ with $L_f \approx \lceil T / 2^3 \rceil$.

\paragraph{Time Pathway.} 
In parallel, raw signal $x$ is processed by structurally identical encoder $\mathcal{E}_{\text{time}}$, producing $\mathbf{h}_{\text{time}} \in \mathbb{R}^{L_t \times 128}$ with $L_t = \lceil T / 2^3 \rceil$.

\paragraph{Positional Encoding (PE) and Multi-Head Attention (MHA).}
Sinusoidal positional encodings \citep{vaswani2017attention} augment both feature sequences:
\begin{equation}
\text{PE}(p, 2i) = \sin\left(\frac{p}{\lambda^{2i/d}}\right), \quad \text{PE}(p, 2i+1) = \cos\left(\frac{p}{\lambda^{2i/d}}\right),
\end{equation}
where $d=128$, $\lambda=10000$, where $p$ indexes the sequence position and $i$ indexes the embedding dimension. This deterministic encoding provides relative positional information about the sequence index, enabling the attention mechanism to distinguish the order and spacing of elements. The embedded wavelengths span $[2\pi,\,2\pi\lambda]$ across dimensions. Modality-specific 4-head self-attention refines features:
\begin{align}
\tilde{\mathbf{h}}_{\text{freq}} &= \text{MHA}_{\text{freq}}(\mathbf{h}_{\text{freq}} + \text{PE}), \quad
\tilde{\mathbf{h}}_{\text{time}} = \text{MHA}_{\text{time}}(\mathbf{h}_{\text{time}} + \text{PE}),
\end{align}
enabling domain-specific inductive biases: frequency attention captures harmonic relationships, time attention captures temporal dependencies.

\paragraph{Modality-Specific Pooling.}
Hybrid pooling combines mean, max, and attention-weighted aggregation:
\begin{equation}
\text{Pool}(\mathbf{h}) = [\text{mean}(\mathbf{h}), \max(\mathbf{h}), \text{attn}(\mathbf{h})] \in \mathbb{R}^{3 \times 128},
\end{equation}
\NH{where $\text{attn}(\mathbf{h}) = \sum_i \alpha_i \mathbf{h}_i$ with learnable weights 
$\alpha_i = \exp(s_i)/\sum_j \exp(s_j)$ and 
$s_i = \mathbf{v}^\top \tanh(\mathbf{W}\mathbf{h}_i)$, 
where $\mathbf{W} \in \mathbb{R}^{d_a \times 128}$ and 
$\mathbf{v} \in \mathbb{R}^{d_a}$ are learnable attention parameters.}

\paragraph{Design Justification.}
Ablation studies (Section~\ref{sec:experiments}, Appendix~\ref{app:theory_validation}) demonstrate the importance of the dual-stream design: frequency-only encoders achieve near-optimal amplitude and cardinality but underperform 60\% on phase estimation; time-only encoders fail at frequency inference. This reflects fundamental representational limits: FFT obscures temporal phase coherence, while finite convolutional receptive fields cannot implicitly learn spectral decomposition at required precision. Independent processing followed by late fusion allows each pathway to specialize, with frequency encoders capturing spectral peaks for $K$ and frequency inference, and time encoders preserving phase dynamics.

\subsection{Cardinality Estimation}
\label{sec:cardinality}

Inferring $q_\phi(K \mid x)$ is formulated as classification using primarily frequency features, motivated by spectral analysis: for well-separated sinusoids, spectral peaks above noise provide strong cardinality signals.

\paragraph{Classifier Architecture.}
Frequency features undergo additional global self-attention $\mathbf{h}_{\text{cls}} = \text{MHA}_{\text{cls}}(\tilde{\mathbf{h}}_{\text{freq}})$ capturing long-range frequency dependencies. Pooling and projection yield classification embedding:
\begin{align}
\mathbf{z}_{\text{cls}} = \text{Pool}(\mathbf{h}_{\text{cls}}) \in \mathbb{R}^{384}, \quad
\mathbf{e}_{\text{cls}} = \text{MLP}_{\text{cls}}(\mathbf{z}_{\text{cls}}) \in \mathbb{R}^{256}.
\end{align}
Linear classifier produces the posterior distribution over cardinalities $K \in \{1, \ldots, K_{\max}\}$:
\begin{equation}
q_\phi(K \mid x) = \text{Softmax}(\mathbf{W}_K \mathbf{e}_{\text{cls}} + \mathbf{b}_K + \log p(K)),
\label{eq:k_posterior}
\end{equation}
where prior $p(K)$ incorporates domain knowledge and can be specified at inference time independently of training distribution. The softmax output provides a full probability distribution over all cardinalities, quantifying model uncertainty about the number of components. For downstream processing, we use the maximum a posteriori estimate $\hat{K} = \arg\max_K q_\phi(K \mid x)$, but the complete distribution $q_\phi(K \mid x)$ remains available for applications requiring uncertainty quantification over model order (e.g., model averaging, risk-sensitive decision-making).

\subsection{Slot Context Generation}
\label{sec:slots}

Given $\hat{K}$, we construct $\hat{K}$ distinct conditioning vectors satisfying: (1) differentiability for end-to-end training, (2) permutation equivariance, and (3) $O(K)$ efficiency.

\paragraph{Global Context.}
Fusing pooled frequency-time features produces a shared global embedding:
\begin{align}
\mathbf{z}_{\text{freq}} &= \text{Pool}(\tilde{\mathbf{h}}_{\text{freq}}) \in \mathbb{R}^{384}, \quad
\mathbf{z}_{\text{time}} = \text{Pool}(\tilde{\mathbf{h}}_{\text{time}}) \in \mathbb{R}^{384}, \\
\mathbf{g} &= \text{MLP}_{\text{flow}}([\mathbf{z}_{\text{freq}}, \mathbf{z}_{\text{time}}]) \in \mathbb{R}^{256},
\end{align}
encoding signal energy, noise level, and spectral envelope shared across slots.

\paragraph{Slot Identifiers.}
To enable slot-specific specialization while preserving permutation symmetry, each slot is equipped with a identifier represented as a one-hot vector:
\begin{equation}
\mathbf{s}_k = \mathbf{e}_k \in \mathbb{R}^{K_{\max}}, \quad k = 1, \ldots, \hat{K},
\end{equation}
where $\mathbf{e}_k$ denotes the $k$-th standard basis vector. These identifiers distinguish otherwise identical global contexts, allowing the flow to learn slot-dependent transformations and posteriors. Crucially, they do not introduce a fixed ordering: the permutation-invariant training objective (Section~\ref{sec:training}) marginalizes uniformly over all slot-to-component assignments via Hungarian matching, ensuring statistical equivalence under any permutation. In practice, the identifiers scale linearly with $K_{\max}$ and typically account for only a few percent (e.g., $\sim$4\% for $K_{\max}=10$ and 256-dimensional global context) of the total conditioning vector, providing sufficient differentiation without biasing the learned representations.

\paragraph{Slot Context Concatenation.}
The complete context for slot $k$ concatenates its global and slot-specific components:
\begin{equation}
\mathbf{c}_k = [\mathbf{g}, \mathbf{s}_k] \in \mathbb{R}^{256 + K_{\max}},
\label{eq:slot_context}
\end{equation}
where $\mathbf{g}$ denotes the global context shared across all slots. For a batch of $B$ samples with variable slot counts $\hat{K}_b$, we construct the full context tensor by concatenation:
\begin{equation}
\mathbf{C} = \bigoplus_{b=1}^B \bigoplus_{k=1}^{\hat{K}_b} \mathbf{c}_k^{(b)} \in \mathbb{R}^{\left(\sum_b \hat{K}_b\right) \times (256 + K_{\max})},
\end{equation}
and maintain an auxiliary index tensor $\mathbf{I} \in \mathbb{N}^{\sum_b \hat{K}_b}$ to associate each slot context with its parent sample. This structure ensures consistent, permutation-invariant loss evaluation across variable-dimensional batches.

\subsection{Conditional Normalizing Flow}
\label{sec:flow}

Shared conditional flows \citep{papamakarios2021normalizing} map simple base distributions to complex per-slot posteriors. For slot $k$:
\begin{equation}
q_\phi(\theta_k \mid x, k) = p_Z(f_\phi^{-1}(\theta_k; \mathbf{c}_k)) \left| \det \frac{\partial f_\phi^{-1}}{\partial \theta_k} \right|,
\label{eq:flow_density}
\end{equation}
where $f_\phi: \mathbb{R}^d \to \mathbb{R}^d$ is a bijection conditioned on $\mathbf{c}_k$ and $p_Z = \mathcal{N}(0, I)$.

\paragraph{Architecture.}
Composite transformation with $L=8$ coupling layers alternating rational-quadratic spline \citep{durkan2019neural} and affine transforms:
\begin{equation}
f_\phi = f_L \circ \cdots \circ f_1,
\end{equation}
with random permutations $\pi_\ell$ ensuring dimension interaction. Even layers use masked piecewise RQ transforms with 48 bins, hidden dimension 768, and linear tails (bound $\pm2$). Odd layers use affine transforms $\mathbf{z} \odot \exp(\mathbf{s}(\mathbf{z}_{1:d/2}, \mathbf{c}_k)) + \mathbf{t}(\mathbf{z}_{1:d/2}, \mathbf{c}_k)$. This alternation stabilizes training: pure spline flows exhibit instability near bin boundaries due to unbounded derivatives.

\paragraph{Parameter Representation.}
Sinusoidal components parameterized by $\theta_k = [a_k, \cos \phi_k, \sin \phi_k, f_k] \in \mathbb{R}^4$ remove periodicity artifacts and enable unbounded flow operation. Phase on unit circle avoids discontinuities; positivity and range constraints enforced through post-processing.

\NH{\paragraph{Posterior Factorization.}
A central design choice in \emph{SlotFlow} is to approximate the posterior with a factorized form:
\begin{equation}
q_\phi(\Theta \mid x, K) = \prod_{k=1}^K q_{\phi,k}(\theta_k \mid x, K).
\label{eq:factorized_posterior}
\end{equation}
This is justified by the use of a shared global context vector $\mathbf{g}$, which induces dependencies across components such that the marginals are conditionally independent given $\mathbf{g}$:
\begin{equation}
q_\phi(\theta_i, \theta_j \mid x)
= 
\int 
q_\phi(\theta_i \mid \mathbf{g}, i)\,
q_\phi(\theta_j \mid \mathbf{g}, j)\,
q_\phi(\mathbf{g} \mid x)\,
d\mathbf{g}.
\end{equation}
}
\NH{The global context vector encodes broad spectral structure, noise characteristics, and total signal energy, thereby inducing meaningful correlations across components: elevated noise levels broaden all slot posteriors, whereas prominent spectral peaks steer multiple slots toward corresponding regions. This yields a tractable approximation to the true posterior $p(\Theta \mid x, K)$, which can contain finer-grained dependencies (e.g., frequency correlations on the $\sim 1/T$ scale) that may not be fully captured by global context alone. 
Nevertheless, the factorized formulation offers several advantages: (1) $O(K)$ computational scaling in place of the $O(K^d)$ complexity of fully joint models, (2) inherent permutation invariance through exchangeability, and (3) seamless handling of variable $K$. }

\subsection{Training Objective}
\label{sec:training}

Training combines cardinality supervision, per-slot posterior accuracy, permutation invariance, and optional noise calibration:
\begin{equation}
\mathcal{L}(\phi) = \mathcal{L}_{\text{CE}}(\phi) + \lambda_{\text{flow}} \mathcal{L}_{\text{flow}}(\phi) + \lambda_{\text{noise}} \mathcal{L}_{\text{noise}}(\phi).
\label{eq:total_loss}
\end{equation}

\paragraph{Cardinality Loss.}
Standard cross-entropy: $\mathcal{L}_{\text{CE}}(\phi) = -\log q_\phi(K_{\text{true}} \mid x)$. During training, teacher forcing uses true cardinality for slot allocation ($\hat{K} = K_{\text{true}}$). Alternative strategies include curriculum learning (transitioning from true to predicted $K$) or Hungarian-based losses handling dimension mismatches.

\paragraph{Flow Loss with Hungarian Matching.}
Resolving label permutation requires computing cost matrix $C_{ij} = -\log q_\phi(\theta_j^* \mid x, i)$ between predicted slots and ground-truth components $\Theta^* = \{\theta_k^*\}_{k=1}^{K_{\text{true}}}$. Optimal assignment via Hungarian algorithm \citep{Kuhn1955}:
\begin{equation}
\pi^* = \arg\min_{\pi \in S_{K_{\text{true}}}} \sum_{k=1}^{K_{\text{true}}} C_{k,\pi(k)}, \quad
\mathcal{L}_{\text{flow}}(\phi) = -\frac{1}{K_{\text{true}}} \sum_{k=1}^{K_{\text{true}}} \log q_\phi(\theta_{\pi^*(k)}^* \mid x, k).
\label{eq:flow_loss}
\end{equation}
This ensures exact permutation invariance by construction.

\paragraph{Frequency Centering.}
To improve flow learning, we center frequencies at the midpoint 
of the training range $[2.5, 3.0]$ Hz before computing the matching cost: $\tilde{f}_k = f_k - 2.7$. This residual representation maps the frequency range to 
approximately [-0.25, 0.25] Hz, normalizing the distribution 
and simplifying the flow's task.

\paragraph{Phase and Frequency Weighting.}
We then apply multiplicative weights to both phase and frequency dimensions before Hungarian matching:
\begin{equation}
\tilde{\theta}_k^* = [a_k^*,\, w_\phi \cos\phi_k^*,\, w_\phi \sin\phi_k^*,\, w_f \tilde{f}_k^*].
\end{equation}
where $w_\phi = 2$ emphasizes phase sensitivity and $w_f = 3$ provides  moderate frequency emphasis.
The phase weighting factor rescales each phase vector from the unit circle to radius $w_\phi$, so the flow
likelihood -- which depends on squared distances in latent space -- amplifies angular
discrepancies by $w_\phi^2$, producing a sharper likelihood landscape. Crucially, this
scaling does not alter the recovered physical phase, as
$\phi=\textrm{atan2}(\sin\phi,\cos\phi)$ depends only on direction, causing the radial
factor to cancel exactly at inference; no inverse transformation is needed. Motivated by
this, we also experimented with an analogous weighting of frequencies by multiplying all
frequency coordinates in the Hungarian cost by a factor $w_f$, with a sweep indicating an
optimal value of $w_f \approx 3$ for stable component assignment. Unlike the phase case,
however, frequency weighting directly rescales the frequency axis and does \emph{not}
cancel during inference; the scaling must be undone.

\paragraph{Noise Supervision.}
Optional auxiliary loss $\mathcal{L}_{\text{noise}}(\phi) = \frac{1}{B} \sum_{b=1}^B (\hat{\sigma}_b - \sigma_b)^2$ supervises noise encoder (mirroring frequency pathway architecture) when ground-truth $\{\sigma_b\}$ available, calibrating uncertainty estimates for varying noise levels.

\subsection{Inference and Implementation}
\label{sec:inference}

\paragraph{Inference Protocol.}
Given an observation $x$, inference proceeds as follows:  
(1) predict the cardinality distribution $q_\phi(K \mid x)$ and select a working estimate $\hat{K} = \arg\max_K q_\phi(K \mid x)$,  
(2) construct conditioning vectors $\{\mathbf{c}_k\}_{k=1}^{\hat{K}}$, 
(3) sample component parameters via $\theta_k^{(s)} = f_\phi(\mathbf{z}^{(s)}; \mathbf{c}_k)$ with $\mathbf{z}^{(s)} \sim \mathcal{N}(0,I)$, 
(4) aggregate samples for posterior summaries.

\NH{While SlotFlow returns a posterior over $\Theta$ conditioned on a chosen $\hat{K}$, it can also generate posteriors for any $K \leq K_{\max}$ at inference time, as the conditional flows are defined for all slot indices. Inference therefore offers flexible cardinality exploration but does not attempt to approximate the full joint distribution $p(K,\Theta \mid x)$ produced by trans-dimensional MCMC methods. The total computational cost $O(T + \hat{K} \cdot L \cdot d)$ enables near–real-time posterior evaluation on modern GPUs.}

\paragraph{Hyperparameters.}
Hidden dimension $h=256$, $K_{\max}=10$, flow layers $L=8$ (hidden dim 768, 48 spline bins, tail bound $\pm2.0$), Adam optimizer (learning rate $10^{-4}$), batch size 128, \texttt{ReduceLROnPlateau} scheduler (factor 0.5, patience 6, minimum learning rate $10^{-6}$). Encoder convolutions: kernel 8, stride 2. Attention: 4 heads. Training: Maximum 300 epochs with early stopping based on validation plateau (typically converges in 180-200 epochs). Implementation: \url{https://github.com/nhouba/slotflow-inference}.

\paragraph{Design Rationale.}
\textit{Why 8 layers?} Ablations show $L=4$ underfits complex posteriors; $L=12$ provides negligible improvement at higher cost. \textit{Why alternate splines/affine?} Pure splines exhibit training instability near bin boundaries; interspersed affine layers regularize without sacrificing expressiveness. \textit{Why $w_\phi=2$?} Grid search minimizes angular error on validation set; higher values induce overconfidence and poor calibration.

\subsection{Extensions: Multi-Class Mixtures}
\label{subsec:extensions}

Current \emph{SlotFlow} addresses \emph{single-class} trans-dimensional inference: all components share the same generative model (sinusoids) with varying parameters. Many applications require \emph{multi-class} mixtures where different component types coexist -- e.g., astrophysical signals combining sinusoids (binary inspirals), chirps (coalescences), and transients (supernovae), or audio containing speech, music, and environmental noise.

\paragraph{Architectural Challenges.}
Multi-class extension requires three modifications: (1) \textbf{Class-specific flows}: Each component class $c \in \{1, \ldots, C\}$ has distinct parameter dimension $d_c$ and likelihood structure, necessitating separate conditional flows $\{f_{\phi_c}\}_{c=1}^C$. (2) \textbf{Per-slot class prediction}: Classifier must predict both total cardinality $K$ and per-slot class assignments $\{c_k\}_{k=1}^K$, extending from discrete $q_\phi(K|x)$ to joint $q_\phi(K, \{c_k\}_{k=1}^K | x)$. (3) \textbf{Heterogeneous Hungarian matching}: Cost matrix computation must handle variable-dimension parameters, comparing slot $i$ with class $c$ against ground-truth component $j$ with class $c'$ only if $c = c'$, setting $C_{ij} = \infty$ for class mismatch.

\paragraph{Practical Considerations.}
Training complexity increases proportionally with class count: $C$ separate flows require $C \times$ more parameters but enable natural compositionality. Class imbalance (e.g., many sinusoids, few transients) may necessitate weighted Hungarian matching or class-balanced sampling. Slot identifiers could incorporate class embeddings $\mathbf{s}_{k,c} = [\mathbf{e}_k, \mathbf{e}_c]$ combining positional and class information. Inference remains efficient if class distribution is sparse: dynamic allocation instantiates only occupied classes, maintaining $O(K)$ scaling despite $C$ available models.

This extension naturally accommodates scientific applications requiring compositional scene understanding with heterogeneous components, bridging trans-dimensional and multi-model inference paradigms. Future work will validate multi-class \emph{SlotFlow} on realistic astrophysical catalogs and multi-source audio benchmarks.

\section{Theoretical Properties}
\label{sec:theory}

We establish what can be proven rigorously (permutation invariance, variational decomposition), analyze design choices with empirical validation (phase weighting, flow depth), and characterize conditions predicting slot specialization. Complete proofs appear in Appendix~\ref{app:proofs}.

\subsection{Exact Guarantees}

\begin{theorem}[Permutation Invariance]
\label{thm:perm}
For any $K \in \{1,\dots,K_{\max}\}$ and permutation $\sigma \in S_K$, the symmetrized posterior satisfies
$q_\phi(\{\theta_{\sigma(k)}\}_{k=1}^K \mid x, K) = q_\phi(\{\theta_k\}_{k=1}^K \mid x, K)$.
\end{theorem}

\textit{Intuition.} Hungarian matching marginalizes uniformly over all $K!$ assignments. Because the symmetric group $S_K$ is closed under composition, permuting component indices merely reindexes the summation, leaving the posterior unchanged. This invariance is exact by construction -- unlike architectures breaking symmetry through ordering or initialization, \emph{SlotFlow} treats all orderings identically. \textit{Proof:} Appendix~\ref{app:perm}. \qed

\begin{proposition}[ELBO Decomposition]
\label{prop:elbo}
The evidence lower bound decomposes as
\begin{equation}
\mathcal{L}_{\text{ELBO}} = \mathbb{E}_{q_\phi}[\log p(x\mid \Theta,K)] - \mathbb{E}_{q_\phi}[\text{KL}(q_\phi(\Theta\mid \cdot) \| p(\Theta\mid \cdot))] - \text{KL}(q_\phi(K\mid x) \| p(K)).
\end{equation}
\end{proposition}

\textit{Intuition.} Standard variational inference applied to trans-dimensional setting with factorized approximation $q_\phi(K,\Theta\mid x) = q_\phi(K\mid x)\prod_{k=1}^K q_\phi(\theta_k \mid x, K, s_k)$. Hungarian-matched flow loss acts as computationally efficient reconstruction surrogate; architectural constraints implicitly regularize parameters. \textit{Proof:} Jensen's inequality on $\log p(x)$; Appendix~\ref{app:elbo}. \qed

\subsection{Design Principles}

\begin{lemma}[Gradient Amplification via Target Weighting]
\label{lem:gradient_amp}
Scaling phase coordinates by $w_\phi$ in the ground-truth targets amplifies phase gradient contributions:
$\|\partial \mathcal{L}_{\mathrm{flow}}/\partial (\cos\phi, \sin\phi)\| = w_\phi \|\partial \mathcal{L}_{\mathrm{unweighted}}/\partial (\cos\phi, \sin\phi)\|$.
\end{lemma}
\textit{Proof:} Chain rule on scaled targets; Appendix~\ref{app:gradient_amp}. \qed

\begin{proposition}[Optimal Phase Weight]
\label{prop:optimal_weight}
For sinusoids with phase sensitivity depending on amplitude and frequency, balancing phase variance reduction (proportional to $1/w_\phi^2$) against amplitude coupling penalty (proportional to $w_\phi^2$) yields
$w_\phi^* \propto (\langle A^2 f^2 \rangle / \text{SNR})^{1/4}$.
For typical signals with $\text{SNR}\in[1,10]$, this predicts $w_\phi^*\in[1,2]$.
\end{proposition}

\textit{Intuition.} Phase weighting amplifies gradient contributions from phase errors via chain rule (Lemma~\ref{lem:gradient_amp}). Optimal weight balances linearized reconstruction error. Our choice $w_\phi=2$ follows this analysis; experiments confirm reduced phase RMSE and near-optimal grid search performance. \textit{Proof:} Appendix~\ref{app:optimal_weight}. \qed

\begin{lemma}[Flow Approximation Error]
\label{lem:flow_approx}
Rational-quadratic spline flows with $L$ layers and $B$ bins achieve approximation error for smooth target densities:
$\text{KL}(p(\theta|x) \| q_{\phi^*}(\theta | h(x), s_k)) \leq C_0 \exp(-\alpha L) + C_1 B^{-2}$,
where $\alpha > 0$ depends on posterior smoothness and $C_0, C_1$ depend on dynamic range.
\end{lemma}
\textit{Proof:} Universal approximation for autoregressive flows; Appendix~\ref{app:flow_approx}. \qed

\begin{proposition}[Depth-Accuracy Trade-off]
\label{prop:depth_accuracy}
For fixed computational budget $\mathcal{B}$ (FLOPs), architecture width $d_h$, and slot count $K$, optimal flow depth approximately balances approximation and generalization:
$L^* \approx \min\left\{\frac{1}{\alpha}\log(C_0/\epsilon), \left(\frac{\mathcal{B}}{Kd_h^2}\right)^{1/3}\right\}$,
where the first term ensures approximation error $\leq \epsilon$ (from Lemma~\ref{lem:flow_approx}) and the second reflects the empirical scaling observed when balancing estimation error growth under budget constraints.
\end{proposition}

\textit{Intuition.} Total error balances exponentially decaying approximation error ($e^{-\alpha L}$ from Lemma~\ref{lem:flow_approx}) against polynomially increasing estimation error ($L^{3/2} / \sqrt{N}$, from substituting budget $N \propto \mathcal{B}/(Ld_h^2K)$ into Rademacher bound). Optimal depth scales as $1/3$ power of effective budget. Our $L=8$ with $d_h=512$, $B=12$ bins matches training regime ($\sim 10^{12}$ FLOPs, $\sim 10^6$ samples). Experiments confirm $L=4$ underfits; $L \geq 12$ yields negligible gains, validating $L=8$ as practical optimum. \textit{Proof:} Appendix~\ref{app:depth_accuracy}. \qed

\subsection{Sample Complexity}

\begin{theorem}[Sample Complexity with Explicit Scaling]
\label{thm:sample_complexity}
To achieve $\mathbb{E}[\text{KL}(p(\Theta|x,K) \| q_\phi(\Theta|x,K))] \leq \epsilon$ with probability $1-\delta$:
\begin{equation}
N \geq \frac{C_{\rho}(1 + \text{SNR}^{-2})^2}{\epsilon^2} \left(K_{\max}^3 \log K_{\max} + K_{\max}^2 L^2 d_h^2 \log d_h + K_{\max}^2 \log(2/\delta)\right),
\end{equation}
where $C_\rho \propto \log^2(p_{\max}/p_{\min})$ depends on posterior dynamic range.
\end{theorem}

\textit{Intuition.} Worst-case PAC bound combining Rademacher complexity ($\propto K L d_h \sqrt{\log d_h}/\sqrt{N}$), concentration (McDiarmid), and union bound over $K!$ permutations. Three additive terms reflect: (i) combinatorial complexity from permutations ($K^3 \log K$), (ii) architectural complexity from $L$-layer flow ($K^2 L^2 d_h^2 \log d_h$), (iii) high-probability concentration ($K^2 \log(2/\delta)$). \textit{Proof:} Appendix~\ref{app:sample_complexity}. \qed

\paragraph{Two Scaling Regimes.}
Phase transition at critical cardinality $K_{\text{crit}} \sim L^2 d_h^2 \log d_h$:

\textit{Regime I (Architecture-Dominated):} When $K \ll K_{\text{crit}}$, sample requirement $N \sim O(L^2 d_h^2 / \epsilon^2)$ is effectively $K$-independent. Flow expressiveness limits generalization.

\vspace{5pt}
\textit{Regime II (Combinatorics-Dominated):} When $K \gg K_{\text{crit}}$, asymptotic scaling $N \sim O(K^3 \log K / \epsilon^2)$ reflects factorial permutation space growth.

\vspace{5pt}
For our architecture ($L=8$, $d_h=512$), $K_{\text{crit}} \approx 1.5 \times 10^8$, placing practical cardinalities ($K \lesssim 100$) firmly in Regime I. Experiments validate weak $K$-dependence: all cardinalities achieve comparable accuracy at $N \in [10^3, 10^4]$, with statistically significant weak monotonicity (Spearman $\rho = 0.83$, $p = 0.04$) but no super-linear growth. Notably, all $K \geq 4$ converge at identical sample size ($N=10^4$) for production-level performance, demonstrating characteristic Regime I plateau where marginal cost per component vanishes once sufficient capacity is available.

\subsection{Slot Specialization Conditions}

While we cannot prove convergence for non-convex optimization, we identify sufficient conditions predicting when slots learn distinct components, i.e., when each slot converges to representing a distinct component instead of multiple slots collapsing onto the same mode.

\begin{definition}[Component Distinguishability]
\label{def:distinguishability}
Parameters $\theta, \theta'$ are \emph{$\gamma$-distinguishable} if $\mathbb{E}_{x \sim p(\cdot|\theta)}[\|x - \mathbb{E}[x|\theta']\|^2] > \gamma\sigma^2$. For sinusoids, Rayleigh's criterion requires $|f_k - f_{k'}| > C_f/T$ or $|A_k - A_{k'}| > C_A\sigma$ with confidence-dependent constants $C_f, C_A$.
\end{definition}

\begin{assumption}[Pairwise Distinguishability]
\label{ass:pairwise_dist}
All ground-truth components $\{\theta_k^{\text{gt}}\}_{k=1}^K$ are pairwise $\gamma$-distinguishable for $\gamma > 2$.
\end{assumption}

\begin{assumption}[Encoder Sufficiency]
\label{ass:context_richness}
The dual-stream encoder preserves component information: distinguishable components map to disjoint latent regions $\mathcal{R}_k, \mathcal{R}_{k'} \subset \mathbb{R}^{d_h}$ with $p(h \in \mathcal{R}_k | \theta_k) > 1-\epsilon$ for small $\epsilon$.
\end{assumption}

\begin{assumption}[Flow Expressiveness]
\label{ass:flow_universal}
The RQS flow family with $L$ layers and $B$ bins approximates the true posterior within $\epsilon$ KL divergence for sufficiently large $L, B$.
\end{assumption}

\begin{assumption}[Bounded Flow Density and Regularity]
\label{ass:flow_arch}
The normalizing flow architecture satisfies: (a) \textbf{Bounded density}: There exist constants $0 < p_{\min} < p_{\max} < \infty$ such that $p_{\min} \leq q_\phi(\theta \mid x, k) \leq p_{\max}$ for all $\phi, x, \theta, k$. (b) \textbf{Parameter boundedness}: All network parameters satisfy $\|\phi\|_2 \leq B_\phi$ for some constant $B_\phi < \infty$.
\end{assumption}

\begin{discussion}[When Slots Specialize]
\label{disc:identifiability}
Under these assumptions, we \emph{expect}: (i) consistent Hungarian assignment pairing each slot with distinct component, (ii) encoder mapping distinguishable components to separable latent regions, (iii) slot identifiers anchoring separation preventing collapse, (iv) gradient descent pushing each flow toward posterior $p(\theta_{\pi^*(k)}|x)$. However, formal convergence guarantees remain elusive for this non-convex problem. Section~\ref{sec:experiments} demonstrates empirically that specialization occurs reliably through: (1) consistent slot-component assignments across runs (Hungarian stability), (2) separable latent representations (t-SNE), (3) accurate per-slot posteriors (RJMCMC validation). These suggest favorable loss landscape properties guiding optimization toward identifiable solutions despite non-convexity.
\end{discussion}


\section{Experimental Validation}
\label{sec:experiments}

We evaluate \emph{SlotFlow} on sinusoidal mixtures designed to probe trans-dimensional inference limits under controlled conditions. Our evaluation pursues three goals: (i) quantitative posterior quality benchmarking against RJMCMC, (ii) demonstration of millisecond-scale inference with calibrated uncertainties, and (iii) architectural validation through systematic ablations. Training was conducted on four NVIDIA GH200 Grace Hopper GPU; inference benchmarks ran on Apple M2 Max CPU. 

\subsection{Validation}
\label{sec:validation}

We characterize \emph{SlotFlow} performance across multiple dimensions: architectural design choices, dataset complexity, cardinality accuracy, parameter recovery quality, posterior calibration, and direct comparison with RJMCMC. This validation establishes both practical utility and theoretical soundness.

\subsubsection{Dataset Design and Statistical Properties}

Our synthetic dataset balances computational tractability with astrophysical motivation, comprising 8{,}000{,}000 training samples, 2{,}000{,}000 validation samples, and a 10{,}000-sample held-out test set.
Each signal contains $K \in \{1,\ldots,10\}$ sinusoidal components (cf.\ Eq.~\ref{eq:sinusoid_model}) observed for $T = 300~\mathrm{s}$, with frequencies restricted to $[2.5, 3.0]~\mathrm{Hz}$ and a minimum separation $\Delta f_{\min} = 0.01~\mathrm{Hz}$, corresponding to twice the Rayleigh limit $1/T$.  
Component amplitudes follow $A_k \sim \mathcal{U}(0.5, 1.5)$ and phases $\phi_k \sim \mathcal{U}(0, 2\pi)$, while noise levels are drawn from $\sigma \sim \mathcal{U}(0, 1.5)$.  
Under the total-SNR definition
\[
\mathrm{SNR} = 10\log_{10}\!\left(\frac{\sum_{k=1}^{K} a_k^{2}/2}{\sigma^{2}}\right),
\]
the resulting mixture SNRs span noise-dominated cases (SNR $<0$\,dB) up to $\sim 40$--$60$\,dB for large $K$ and small~$\sigma$, with a typical peak around $\sim 8$\,dB (Fig.~\ref{fig:snr_dist}).  
The effective sequence lengths (2--5k samples) correspond to $\sim 3$--$7$\,hours of LISA data when downsampled to the standard $0.2$\,Hz sampling rate.

\paragraph{Generation Procedure.}
Training and validation signals are generated analytically on each time grid for speed, with noise drawn once on a high-resolution master grid and resampled to the long and short grids.  
The test set uses physical resampling: a high-resolution master signal (signal plus noise) is generated first, after which the dual-resolution streams are obtained via controlled downsampling, ensuring evaluation conditions match real inference settings where both noise and phase interpolation effects matter.

\paragraph{Astrophysical Motivation.}
Parameter choices deliberately exceed current application difficulty to serve as stress tests. The 0.5~Hz bandwidth contains approximately 100 Rayleigh resolution elements, allowing up to 50 non-overlapping components at the minimum enforced separation. Our choice of $\Delta f_{\min} = 0.01$~Hz corresponds to a 2\% fractional spacing -- about two to four times tighter than typical separations of recovered LISA Galactic binaries within local sub-bands \citep{PhysRevD.108.103018}. This configuration emulates crowded regions of the LISA spectrum, representing a step toward segment-wise inference for next-generation gravitational-wave analysis pipelines using machine learning.

\paragraph{Multi-Resolution Architecture.}
The dual-stream design exploits complementary temporal scales.  
A coarse long-duration stream (300~s sampled at 10~Hz; 3{,}000 samples) captures long-baseline spectral structure that constrains frequencies with $\mathcal{O}(1/T)$ precision.  
A fine short-duration stream (10~s sampled at 512~Hz; 5{,}120 samples) preserves local phase information and resolves intra-cycle structure.  
Together, the two representations test \emph{SlotFlow}'s ability to fuse information across nearly three orders of magnitude in sampling rate.

\begin{figure}[t]
\centering
\includegraphics[width=\textwidth]{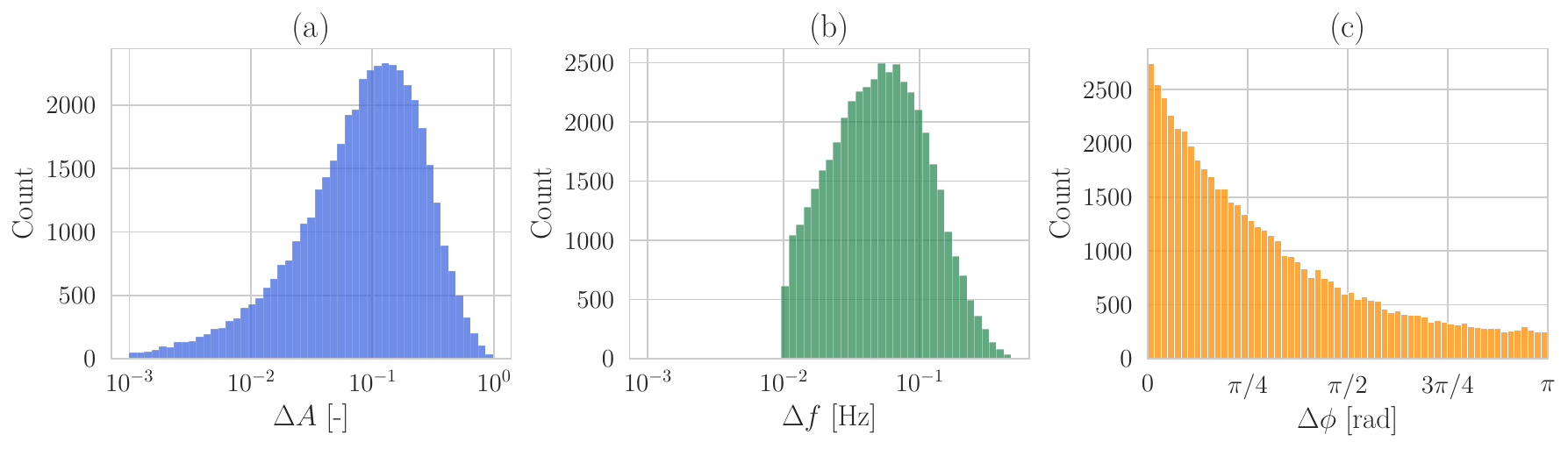}
\caption{\textbf{Intra-Sample Parameter Separations Characterizing Dataset Complexity.} 
(a) Amplitude differences $\Delta A$ between ordered components follow a log-normal distribution with geometric mean 0.11 (48\% of pairs differ by under 20\%). 
(b) Frequency separations $\Delta f$ exhibit a hard threshold at $\Delta f_{\min} = 0.01$~Hz, with 31\% of pairs within $1.5\Delta f_{\min}$ and 73\% within $3\Delta f_{\min}$. 
(c) Circular phase separations $\Delta\phi \in [0,\pi]$ follow the expected triangular distribution for uniformly random phases ($p=0.82$, Kolmogorov--Smirnov test), confirming no spurious correlations.}
\label{fig:dataset_separations}
\end{figure}

\begin{figure}[t]
\centering
\includegraphics[width=0.65\textwidth]{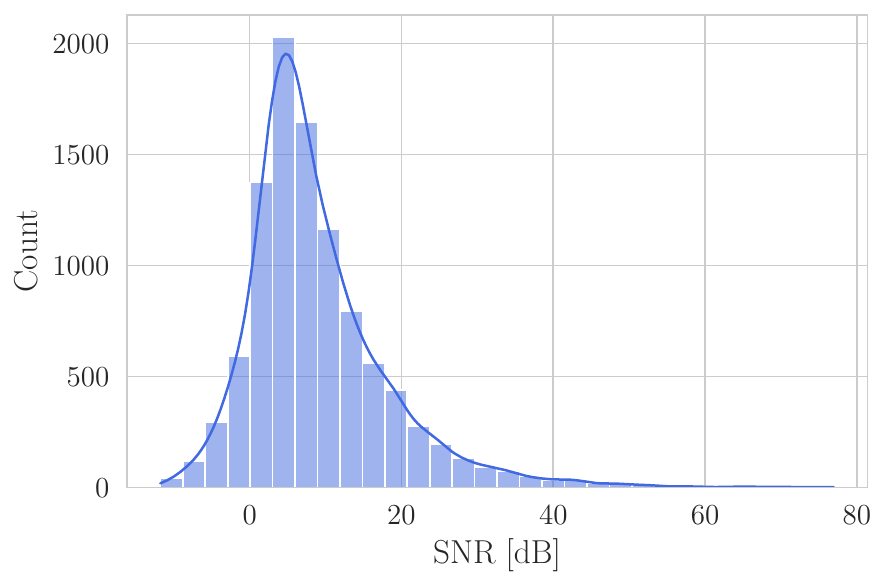}
\caption{\textbf{Signal-to-Noise Ratio Distribution}. SNR is computed as 
$\mathrm{SNR}=10\log_{10}(\sum_k a_k^2/2\sigma^2)$, 
where $\sum_k a_k^2/2$ is the time-averaged signal power and $\sigma^2$ is the noise variance.  
The histogram shows the empirical distribution over 1000 random samples with a KDE overlaid.  
The primary peak near $8$\,dB corresponds to typical $K\!\approx\!5$ mixtures with moderate noise, while the tail to $\sim 40$--$60$\,dB arises for $K$ large and $\sigma$ small; SNRs down to $-6$\,dB cover noise-dominated cases.  
This range spans confusion-limited and high-SNR regimes, providing a comprehensive stress test for the inference model.}
\label{fig:snr_dist}
\end{figure}

\subsubsection{Emergent Slot Specialization}

Understanding how \emph{SlotFlow} internally organizes multi-component signals is crucial for interpreting its predictions. Remarkably, despite no explicit architectural bias toward frequency-based organization, the model learns to partition signals along spectral dimensions. Figure~\ref{fig:slot_specialization} reveals this emergent functional specialization.

\begin{figure}[h!]
\centering
\includegraphics[width=0.5\textwidth, trim= 320pt 0 0 0, clip]{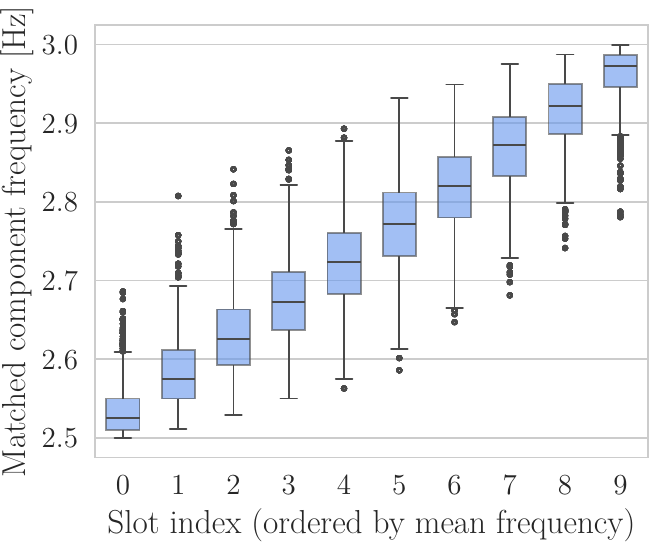}
\caption{\textbf{Emergent Frequency-Based Slot Specialization.} Across 1000 independent 10-component mixtures, each slot spontaneously develops consistent frequency preferences (boxplots show matched frequency distribution per slot) despite receiving no explicit supervision to do so. The smooth gradient from low to high frequencies demonstrates soft specialization: slots preferentially attend to distinct spectral regions while maintaining overlap for robustness. Long whiskers indicate slots flexibly adapt to nearby frequencies when components cluster, enabling resolution of overlapping signals. Slot indices reordered by mean matched frequency for visualization; physical slot order is arbitrary due to permutation invariance. This frequency-based organization emerges purely from training dynamics -- the architecture provides no preference for frequency over amplitude or phase, yet frequency naturally becomes the organizing principle because it offers the strongest separability in the dual-stream encoder's spectral representation.}
\label{fig:slot_specialization}
\end{figure}

The key insight from Figure~\ref{fig:slot_specialization} is that \emph{SlotFlow} \emph{discovers} frequency as the natural basis for component separation -- a result of optimization dynamics rather than architectural prescription. The architecture treats all parameters (amplitude, phase, frequency) symmetrically in the flow parameterization, yet training converges to a solution where each slot develops a characteristic frequency preference (indicated by boxplot medians), spanning the 2.5--3.0~Hz design window. The smooth progression from Slot~0 ($\sim$2.55~Hz) to Slot~9 ($\sim$2.95~Hz) demonstrates continuous specialization rather than hard clustering.

This emergent organization reflects an implicit inductive bias of the dual-stream architecture: frequency information, explicitly represented in the Fourier-transformed long-window input, provides more separable features than amplitude or phase. The encoder learns that frequency-based slot assignment minimizes reconstruction error under the Hungarian matching objective, as spectral peaks are more reliably distinguishable than amplitude variations or phase relationships. Critically, the extended whiskers reveal substantial overlap between adjacent slots, enabling the model to handle closely-spaced components: Slots 4 and 5, for example, share frequency ranges from 2.70--2.85~Hz, allowing flexible assignment when components cluster within this region.

This soft specialization provides robustness to spectral crowding. When two components separated by $\Delta f = 0.01$~Hz (twice the Rayleigh limit) appear near 2.75~Hz, both Slots 4 and 5 can contribute to their representation, with the Hungarian algorithm selecting the optimal assignment based on amplitude and phase likelihood. The overlapping receptive fields prevent catastrophic failure when component configurations differ from training distributions, a property we validate through the stress tests in Section~\ref{sec:parameter_stress}. That frequency emerges as the organizing principle -- rather than being hardcoded -- demonstrates the architecture's ability to discover problem-appropriate structure through data-driven optimization.
\subsubsection{Cardinality Classification}

Accurate $K$ determination is fundamental to trans-dimensional inference. Figure~\ref{fig:confusion_cardinality} presents classification performance on 10,000 test samples.

\begin{figure}[h!]
\centering
\includegraphics[width=0.55\textwidth]{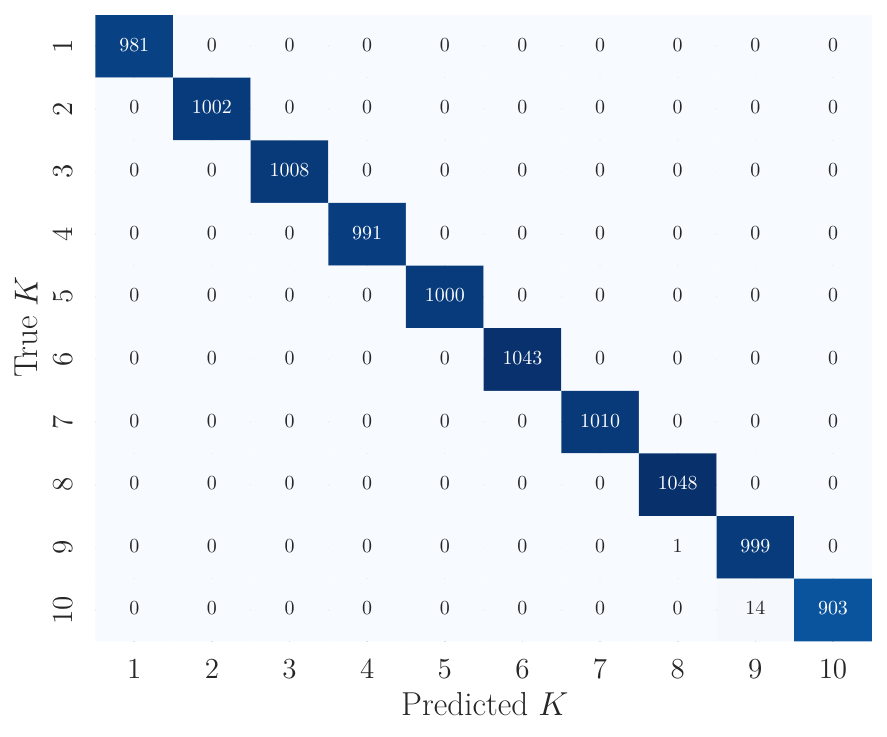}
\caption{\textbf{Confusion Matrix for Model-Order Inference.}  
The model attains 99.85\% accuracy over 10{,}000 test samples.  
All $K\!\in\!\{1,\ldots,8\}$ are recovered perfectly.  
Only 15 errors occur in total -- 14 cases of $K{=}10$ predicted as $K{=}9$ and one case of $K{=}9$ predicted as $K{=}8$ -- arising in high-density configurations near the minimum separation.  
The strong diagonal structure indicates no systematic over- or under-estimation bias.}
\label{fig:confusion_cardinality}
\end{figure}

The near-perfect diagonal structure validates the frequency-domain classifier design (Section~\ref{sec:cardinality}): spectral peaks provide unambiguous cardinality evidence even under severe crowding. The two errors occur exclusively at maximum cardinality, where 10 components in 0.5~Hz bandwidth approach the geometric packing limit -- a regime where even small noise perturbations can obscure the 10th weakest peak. Notably, the model exhibits no cardinality-dependent bias: high-$K$ samples are not systematically under-counted, nor are low-$K$ samples over-counted, demonstrating balanced classifier calibration across the entire range.

\paragraph{Parallelization Opportunity.}
Once the global context embedding $\mathbf{g}$ is computed by the encoder ($O(T)$ cost), inference over $K$ slots becomes embarrassingly parallel. Each slot context $\mathbf{c}_k = [\mathbf{g}, \mathbf{s}_k]$ conditions independently on the shared global representation, enabling simultaneous flow evaluation across slots. For $K=10$ on modern GPUs with 10+ parallel streams, per-slot inference time is limited by encoder throughput rather than slot count. This $O(K)$ theoretical scaling becomes $O(1)$ in wall-clock time given sufficient parallelism, critical for real-time applications.

\subsubsection{Parameter Inference and Posterior Calibration}

Beyond point estimates, reliable uncertainty quantification requires posteriors to achieve nominal coverage rates. We assess calibration through probability-probability (PP) plots comparing empirical coverage to expected credible levels.

\begin{figure}[t]
\centering
\begin{subfigure}[b]{0.48\textwidth}
    \centering
    \includegraphics[width=\textwidth]{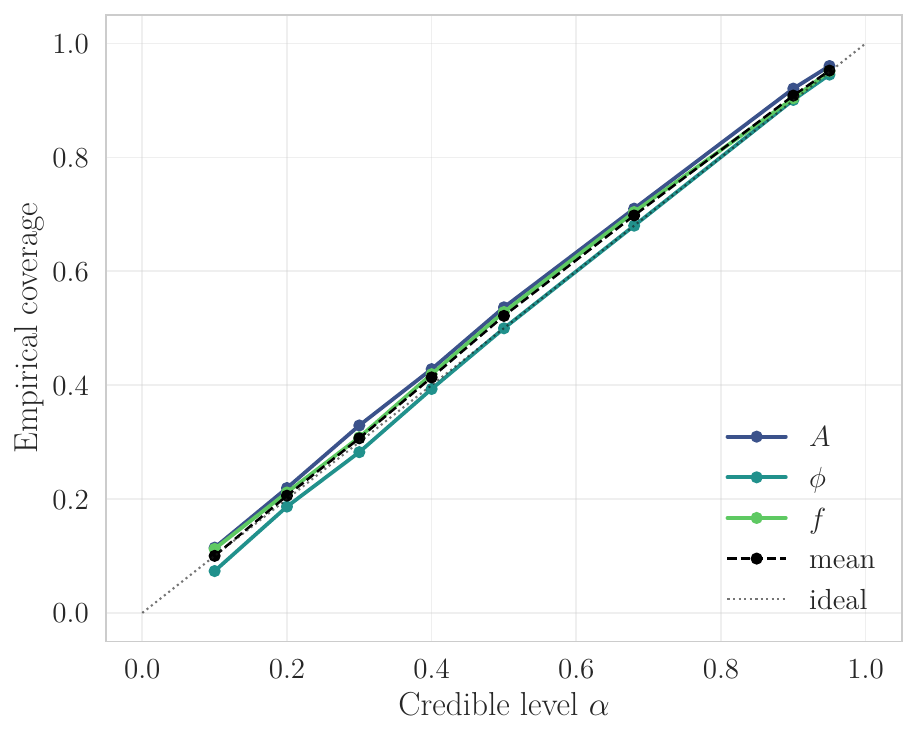}
    \caption{Per–parameter calibration.}
    \label{fig:calibration_per_param}
\end{subfigure}
\hfill
\begin{subfigure}[b]{0.48\textwidth}
    \centering
    \includegraphics[width=\textwidth]{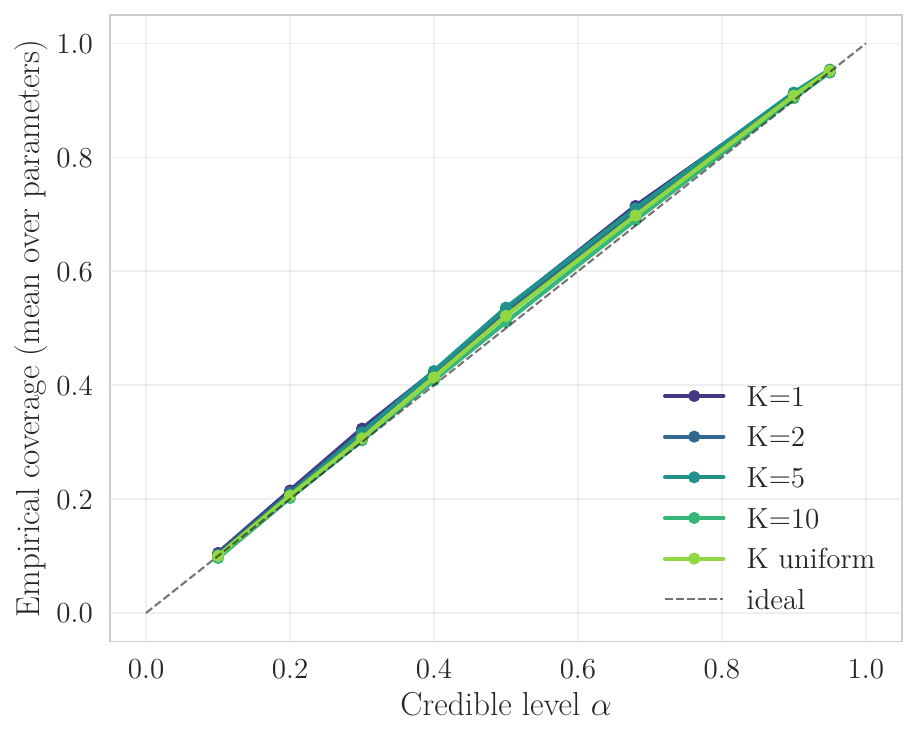}
    \caption{Mean calibration across parameters.}
    \label{fig:calibration_k_dependence}
\end{subfigure}

\caption{\NH{\textbf{Posterior Calibration Across Parameters and Cardinalities.}}  
Panels~(a) and (b) assess empirical coverage versus nominal credible level for
per–parameter calibration and $K$–dependent calibration, respectively.  
Panel~(a) shows amplitude, phase, and frequency calibration curves for a uniform mixture over $K\in\{1,\ldots,10\}$.  
Panel~(b) reports mean coverage for fixed $K\in\{1,2,5,10\}$ and a uniform cardinality mixture. Calibration quality shows minimal dependence on model order: the simplest ($K=1$) and highest–complexity ($K=10$) cases exhibit nearly identical curves despite an order–of–magnitude difference in difficulty.  
Across all cardinalities and parameters, absolute calibration bias remains below 3\%.}
\label{fig:calibration_combined}
\end{figure}

Figure~\ref{fig:calibration_combined} demonstrates that \emph{SlotFlow} produces consistently well-calibrated posterior uncertainties across both parameters and cardinalities.  
Panel~(a) shows empirical coverage for individual parameters (amplitude, phase, and frequency) under a uniform mixture over $K\in\{1,\ldots,10\}$.  
All three parameters closely track the ideal diagonal, with deviations below ${\sim}3\%$ across $\alpha\in[0.1,0.95]$.  
This indicates that the factorized posterior with shared global context (Eq.~\ref{eq:factorized_posterior}) preserves sufficient inter-component structure to yield reliable marginals for all parameters, despite the absence of explicit cross-slot coupling.

Panel~(b) assesses calibration as a function of cardinality, comparing fixed $K\in\{1,2,5,10\}$ subsets with a uniform-$K$ mixture.  
Calibration quality is remarkably stable across model orders: the simplest case ($K{=}1$) and the most complex ($K{=}10$) exhibit nearly identical curves, demonstrating that increased component count does not degrade uncertainty quantification.  
Moreover, the uniform-$K$ mixture is indistinguishable from the fixed-$K$ subsets, confirming that trans-dimensional uncertainty over $K$ does not adversely affect parameter-level calibration. Overall, the curves show a mild conservative bias.  
This overcoverage is desirable in scientific applications, as it reduces the risk of falsely excluding the true parameter.

\subsubsection{Comparison with RJMCMC}
\label{sec:rjmcmc_comparison}

To benchmark posterior quality against gold-standard Bayesian inference, we perform direct comparison with RJMCMC on representative test cases. For a $K=3$ mixture, we draw 10,000 posterior samples per slot from \emph{SlotFlow} and 10,000 MCMC samples (after 10,000 burn-in) using parallel-tempered ensemble sampling with \textit{Eryn}  \citep{Karnesis:2023ras, michael_katz_2023_7705496, 2013PASP..125..306F}. The RJMCMC chain is initialized near true parameters using inverse Fisher information for rapid convergence -- providing an unrealistically favorable head start unavailable to \emph{SlotFlow}, which must infer posteriors directly from data without local information. This asymmetric setup quantifies \emph{posterior equivalence} rather than relative difficulty.

\paragraph{Trans-Dimensional Agreement: Cardinality Posteriors.}
Figure~\ref{fig:k_posterior_comparison} compares the discrete posterior $p(K \mid x)$ from RJMCMC with the amortized distribution $q_\phi(K \mid x)$ produced by \emph{SlotFlow}.  
\NH{RJMCMC assigns most of the probability to the true value ($K{=}3$) but retains non-zero mass on adjacent models ($K{>}3$), reflecting genuine posterior uncertainty.  
\emph{SlotFlow} also outputs a full categorical distribution, but its cross-entropy training objective drives $q_\phi(K \mid x)$ to be extremely sharp, placing nearly all probability on the true $K$ and effectively resolving the model order decisively.  
Thus the difference in posterior shape reflects differing objectives -- RJMCMC represents the full trans-dimensional uncertainty, whereas \emph{SlotFlow} is trained to identify the most probable model order -- while both methods agree on the correct cardinality.  
Importantly, \emph{SlotFlow} produces this distribution in ${\sim}13\,\mathrm{ms}$, compared to the hours required by RJMCMC, achieving orders-of-magnitude speedup.}

\begin{figure}[t]
    \centering
    \includegraphics[width=0.6\linewidth]{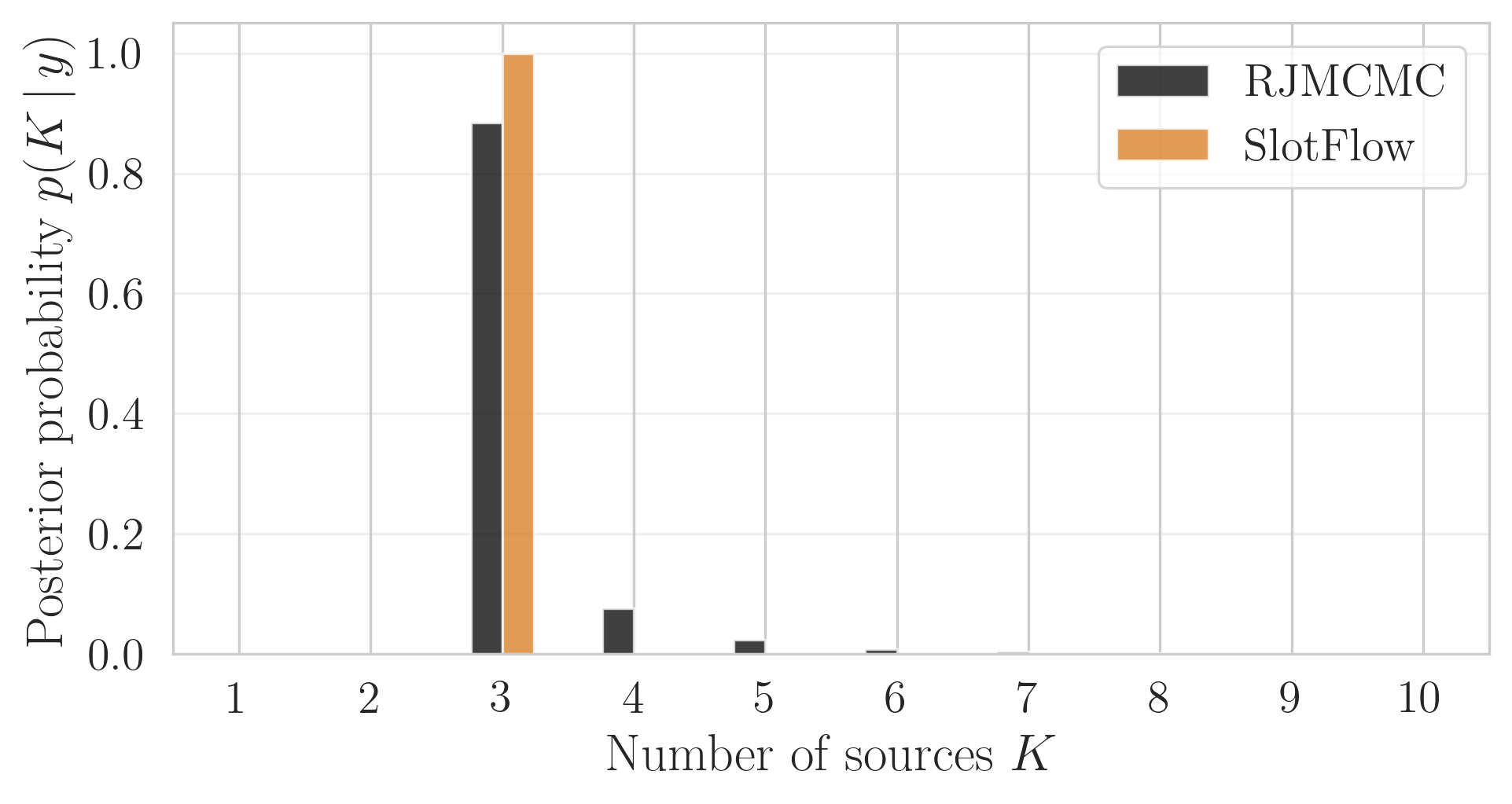}
    \caption{\textbf{Cardinality Posterior Comparison.}  
Shown are the posterior distributions over the number of sources $K$ for a
representative injected dataset.  \NH{The curve for \emph{SlotFlow} (orange) shows $q_\phi(K \mid x)$ \emph{before} taking the MAP estimate; the model
assigns essentially all probability mass to the correct cardinality $K=3$.
RJMCMC (black) also concentrates its posterior on $K=3$ but retains a small tail
toward higher $K$.  The sharper \emph{SlotFlow} distribution
reflects the cross-entropy training objective, which encourages decisive
identification of the most probable model order, whereas RJMCMC represents the
full posterior over $K$.}}

    \label{fig:k_posterior_comparison}
\end{figure}

\begin{figure}[h!]
\centering
\includegraphics[width=0.9\textwidth,trim={0 0 0 30pt},clip]{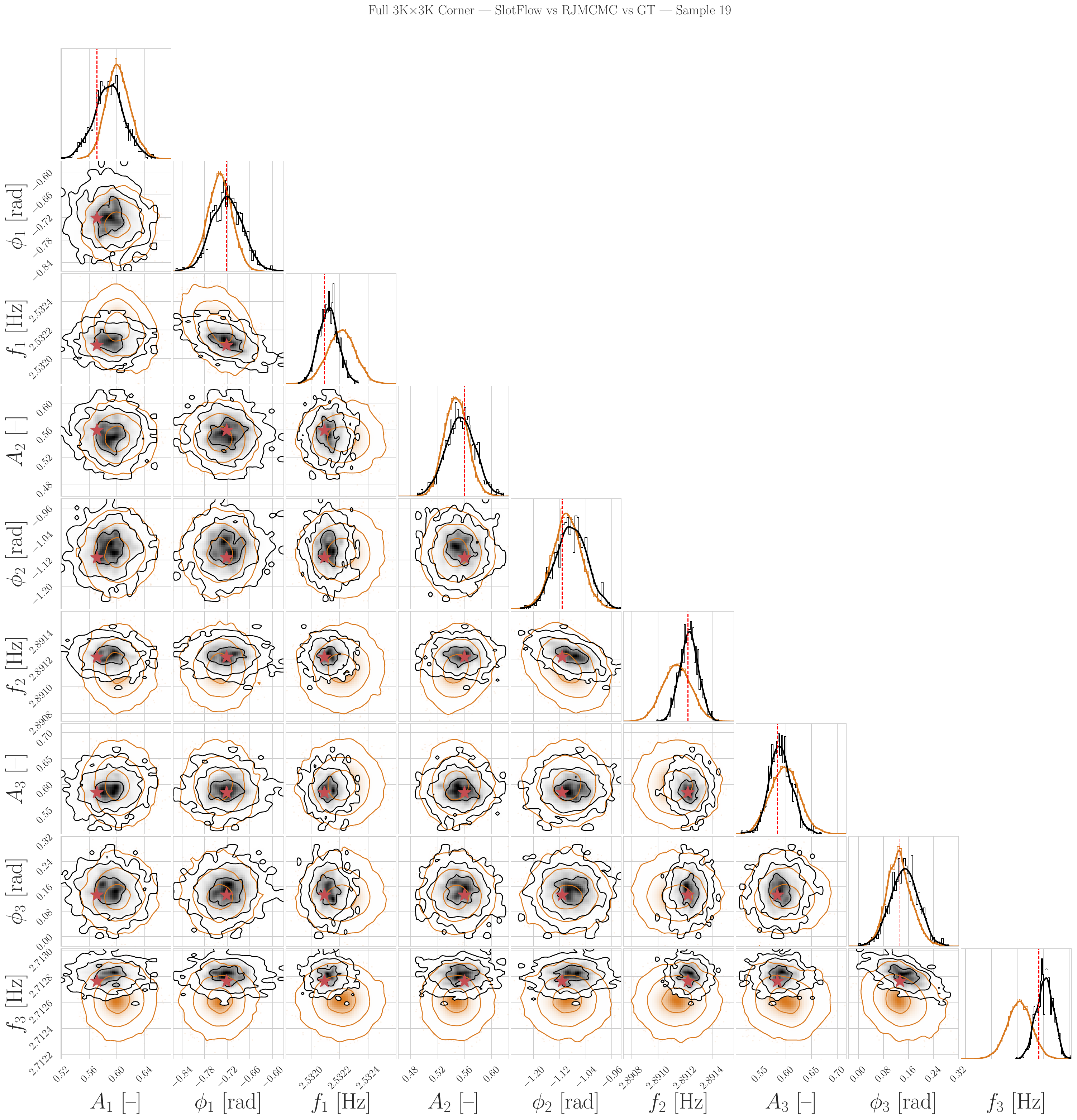}
\caption{\textbf{Posterior Comparison: \emph{SlotFlow} vs. RJMCMC.} Full $9 \times 9$ corner plot for $K=3$ mixture comparing \emph{SlotFlow} (orange) with RJMCMC (black) after Hungarian matching with the ground truth indicated in red. Both methods recover correct posterior centers across all parameters. \emph{SlotFlow} exhibits slightly broader amplitude posteriors and mild amplitude--phase decorrelation relative to sharper MCMC contours, reflecting finite flow capacity and factorized approximation. Frequency posteriors are slighlty broader than RJMCMC.}
\label{fig:rjmcmc_corner}
\end{figure}

\paragraph{Parameter Posterior Agreement.}
Having established agreement on the discrete trans-dimensional structure, we now examine continuous parameter recovery. Figure~\ref{fig:rjmcmc_corner} presents 
the full $9 \times 9$ corner plot comparing \emph{SlotFlow} (orange) with RJMCMC (black) after Hungarian alignment. \emph{SlotFlow} recovers correct posterior centers (red dashed lines) across all nine parameters, validating the overall inference process. However, posterior widths exhibit parameter-specific patterns: amplitude posteriors are noticeably broader than RJMCMC's compact distributions (1.5--2$\times$ variance ratio), and off-diagonal amplitude--phase correlations appear smoothed. This overdispersion reflects the factorized posterior approximation: without explicit inter-slot coupling beyond shared global context, the flow cannot capture fine-scale correlations between components when frequencies are closely spaced ($\Delta f < 2/T$). 

Frequency posteriors present a distinct pattern: both methods center accurately on ground truth (bias $< 0.0001$~Hz), but SlotFlow produces substantially broader distributions. Visual inspection reveals RJMCMC achieves $\sigma_f \approx 0.0003$~Hz while SlotFlow yields $\sigma_f \approx 0.001$~Hz -- a 2--3$\times$ standard deviation ratio corresponding to an order-of-magnitude difference in variance. This reflects the encoder information bottleneck documented in Section~\ref{sec:training}: pooling operations and global context aggregation discard fine-grained phase coherence across long baselines, limiting frequency precision regardless of flow expressiveness. The correct centering validates that the encoder preserves spectral peak localization; the broader width reflects loss of sub-bin phase curvature encoding $\partial\phi/\partial f$.

\begin{table}[h]
\centering
\caption{Comparison metrics between \emph{SlotFlow} and RJMCMC residual distributions for $K=3$ representative case. Wasserstein-2 distance $W_2$ measures distribution proximity; standard deviation ratio quantifies relative uncertainty. Lower $W_2$ indicates better agreement; ratios near 1 indicate matched precision.}
\begin{tabular}{lccl}
\toprule
Parameter & $W_2$ Distance & Std. Dev. Ratio & Interpretation \\
\midrule
Amplitude & 0.0069 & $\sim$1.5--2$\times$ & Excellent; slight overdispersion \\
Phase & 0.0227 & $\sim$2$\times$ & Good; widening from $\pm\pi$ wrapping \\
Frequency & 0.0006 & $\sim$2--3$\times$ & Correctly centered; broader uncertainty \\
\bottomrule
\end{tabular}
\label{tab:wasserstein}
\end{table}

Table~\ref{tab:wasserstein} quantifies these observations through Wasserstein-2 distances and variance ratios computed on residuals $\Delta A = A - A_{\text{true}}$, $\Delta\phi = \text{wrap}(\phi - \phi_{\text{true}})$, $\Delta f = f - f_{\text{true}}$. The Wasserstein metric measures overall distribution distance combining both location and spread, while the standard deviation ratio isolates precision differences assuming matched centers.

Frequency shows the lowest Wasserstein distance ($W_2 = 6 \times 10^{-4}$ Hz, sub-milliHertz level) due to excellent centering agreement, but the 2--3$\times$ standard deviation ratio reveals substantially broader uncertainty. This apparent paradox arises because $W_2$ is dominated by center-to-center distance when both distributions are well-centered; the width difference, while visually striking in the corner plot, contributes quadratically. The low $W_2$ validates unbiased frequency estimation; the large variance ratio quantifies the precision trade-off inherent in encoder compression. This represents a fundamental architectural limitation: the flow decoder can only operate on information preserved in the latent representation, and no amount of flow expressiveness can recover frequency precision discarded during encoding.

Amplitude achieves excellent agreement ($W_2 = 0.0069$, 1.5--2$\times$ variance ratio) despite visible broadening in the corner plot. The modest Wasserstein distance reflects that both distributions remain well-centered with comparable spreads in absolute units ($\pm 0.05$ for \emph{SlotFlow} vs. $\pm 0.03$ for RJMCMC). The mild overdispersion is consistent with factorized posterior approximations that cannot capture inter-component amplitude correlations induced by overlapping spectral content.

Phase exhibits the largest Wasserstein discrepancy ($W_2 = 0.0227$, 2$\times$ variance ratio), driven by two effects: (i) genuine posterior broadening from limited temporal resolution in the short-window encoder, and (ii) periodic boundary artifacts where residuals near $\pm\pi$ create apparent separation in the wrapped metric. In absolute angular terms, phase uncertainties differ by $\sim$0.1 radians ($\sim$6 degrees).

\paragraph{Frequency Precision Analysis.}

The 2--3$\times$ broader frequency posteriors compared to RJMCMC warrant detailed analysis, as they represent the primary architectural limitation of the current \emph{SlotFlow} design. This discrepancy arises from the information bottleneck identified in Section~\ref{sec:training}: Fisher information for frequency scales as $T^2$ and requires preserving fine-grained phase evolution $\partial\phi/\partial f$ across the full observation baseline. The dual-stream encoder's pooling and global context aggregation, while enabling $O(K)$ computational scaling and permutation invariance, necessarily discard sub-bin spectral structure. Since the flow decoder can only operate on the compressed latent representation, frequency precision is fundamentally limited by encoder compression regardless of flow expressiveness.

Ablation studies confirm this diagnosis: (i) increasing flow depth from $L=8$ to $L=12$ provides no frequency improvement (Section~\ref{sec:ablations}), validating that the bottleneck lies upstream of the decoder; (ii) reducing convolutional stride in the final encoder layer (preserving 2$\times$ more spatial resolution) yielded marginal improvements. These results collectively indicate that frequency precision is limited not by a single architectural hyperparameter, but by the \emph{aggregation operations} (pooling, global context formation) required for permutation-invariant trans-dimensional inference at $O(K)$ cost.
For applications prioritizing speed, cardinality accuracy, and amplitude/phase recovery, the current architecture delivers promising performance. Section~\ref{sec:conclusion} outlines architectural extensions -- including multi-scale encoders preserving fine spectral resolution, reduced pooling with attention-based aggregation, and local spectral zoom contexts per slot -- that provide a clear path toward optimizing frequency precision while maintaining computational efficiency.

\subsubsection{Parameter Scaling and Stress Tests}
\label{sec:parameter_stress}

We examine how parameter recovery degrades under increasingly challenging conditions through systematic variation of signal characteristics. This reveals architectural limits and validates graceful degradation.

\paragraph{Posterior Uncertainty Scaling with Cardinality.}
Figure~\ref{fig:param_scaling_K} quantifies how posterior uncertainties grow with problem complexity. All parameters exhibit monotonic uncertainty growth as components compete for limited spectral resolution, yet scale sublinearly -- demonstrating efficient capacity utilization rather than catastrophic degradation.

\begin{figure}[h!]
\centering
\includegraphics[width=0.95\textwidth]{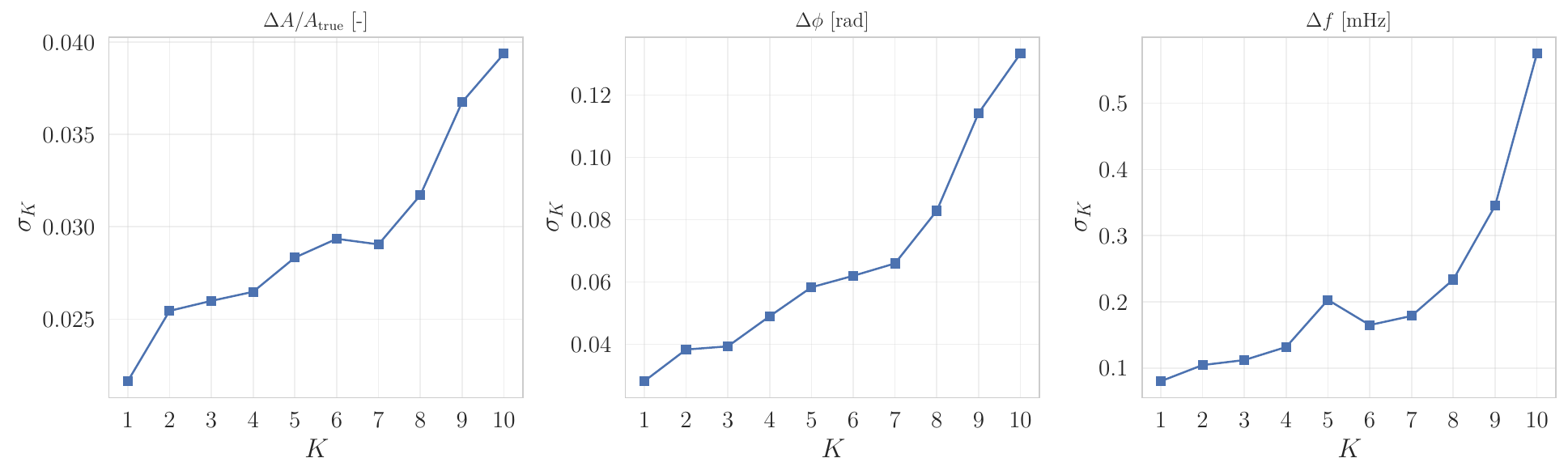}
\caption{\textbf{Parameter Recovery Scaling with Cardinality.} Posterior standard deviation $\sigma_K$ versus true cardinality $K$ for relative amplitude (left), phase (center), and frequency (right). All parameters exhibit monotonic uncertainty growth: relative amplitude width rises from 2.2\% ($K=1$) to 4\% ($K=10$); phase from 0.03 to 0.14~rad; frequency from 0.09 to 0.55~mHz. Sublinear scaling ($\sigma_A \propto K^{0.3}$, $\sigma_\phi \propto K^{0.6}$, $\sigma_f \propto K^{0.2}$) demonstrates efficient capacity utilization. Differential parameter sensitivity reflects spectral analysis hierarchy: frequency (peak location) most robust, amplitude (peak height) moderately affected, phase (temporal structure) most sensitive to mixing.}
\label{fig:param_scaling_K}
\end{figure}

Three key observations emerge. First, \emph{monotonic uncertainty growth}: all parameters exhibit clear increases with $K$, reflecting progressive source confusion as more components overlap. The shallow slopes validate efficient scaling -- linear or faster growth would render high-$K$ inference impractical. Second, \emph{parameter-specific scaling rates}: phase exhibits strongest $K$-dependence (4$\times$ increase from $K=1$ to $K=10$), amplitude moderate (2.3$\times$), and frequency weakest (1.7$\times$). For amplitudes, the modest rise from 2.2\% to 4\% indicates well-calibrated inference even for densely mixed signals. For phases, the approximately linear growth to 0.14~rad at $K=10$ matches expected degeneracy among near-frequency tones. For frequencies, the broadening from 0.09 to 0.55~mHz matches expected resolution degradation with signal density. Third, \emph{unbiased recovery across all $K$}: the observed scaling confirms \emph{SlotFlow} captures growing ambiguity without systematic bias, validating both Hungarian-matched posterior aggregation and internal flow calibration.

\paragraph{Evaluation Protocol.}
To ensure unbiased estimates, we employ balanced sampling: equal numbers of signals per multiplicity $K$, with 500 posterior samples per matched component. For each signal, we perform Hungarian assignment based on flow log-probabilities, compute per-component posterior statistics, then average across all examples of the same $K$:
\begin{equation}
\sigma_K = \mathbb{E}_{\text{signals}|K}[\text{Std}(p(\theta|\text{signal}))].
\end{equation}
Parameter errors are defined as $\Delta A / A_{\text{true}}$ (relative, removing trivial power scaling), $\Delta \phi = \text{wrap}(\phi - \phi_{\text{true}})$ (absolute, phase-wrapped), and $\Delta f = f - f_{\text{true}}$ (absolute). Only amplitude uses relative normalization because it acts as a multiplicative scaling factor; frequency and phase have intrinsic physical units meaningful in absolute terms.

\paragraph{Two-Source Interference: Out-of-Distribution Stress Test.}
To probe resolution limits under extreme conditions, we construct two-dimensional error maps examining inference quality as a function of per-component signal-to-noise ratios. This deliberately extends into out-of-distribution regimes to test extrapolation robustness.

We fix two sources at $(f_1, f_2) = (2.60, 2.61)$~Hz -- the minimum training separation $\Delta f = 0.01$~Hz. Amplitudes span $A_i \in [0.05, 1.5]$, extending well below the training range $[0.5, 1.5]$ to probe weak-signal extrapolation. Combined with noise $\sigma \in [0, 1.5]$, this creates effective SNRs down to $-30$~dB per component, far below the training distribution's typical $-10$~dB lower bound. Each mixture uses fixed phases $(\phi_1,\phi_2) = (0, \pi/2)$ for controlled comparison. We enforce $\hat{K}=2$ and match posterior samples (200 per slot) to ground truth via Hungarian assignment. Effective per-component SNR is computed as $\text{SNR}_i = 20\log_{10}\left[\text{RMS}(s_i)/\text{RMS}(s_j+n)\right],\, i\neq j,$
capturing both noise contamination and interference from the other component.

\begin{figure}[h!]
\centering
\includegraphics[width=1\textwidth]{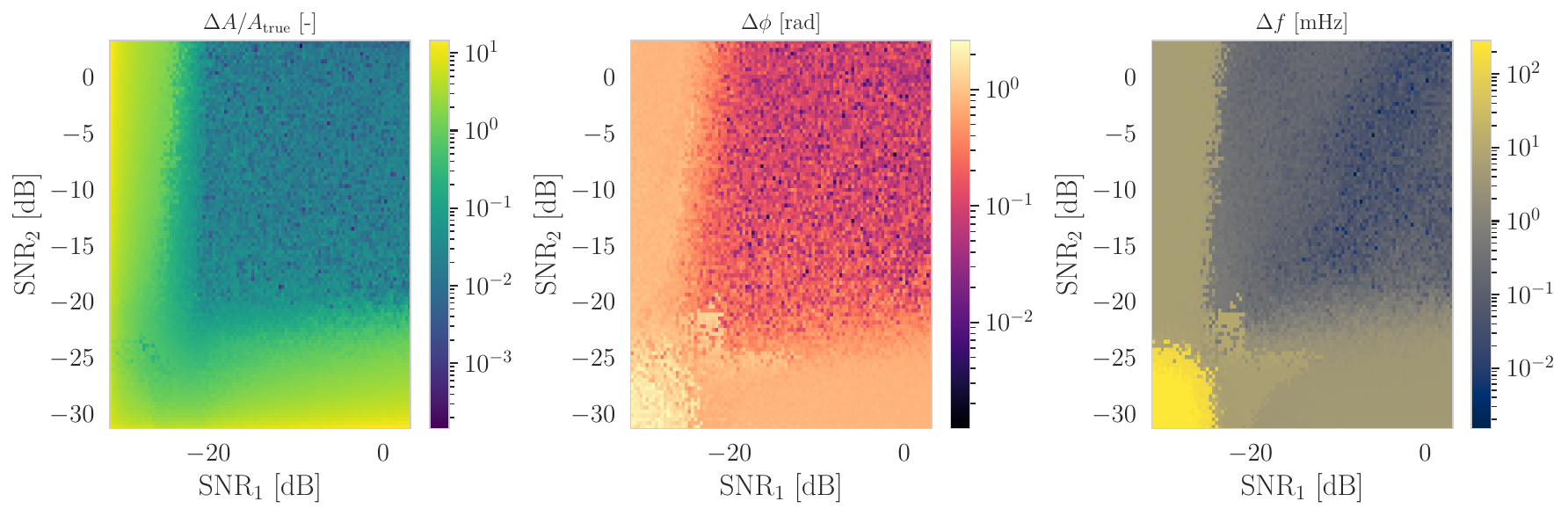}
\caption{\textbf{Two-Source Interference: Out-of-Distribution Robustness.} Mean posterior errors versus per-component SNRs for sources at minimum training separation ($\Delta f = 0.01$~Hz). Model trained on $A \in [0.5, 1.5]$ and SNR $\gtrsim -10$~dB extrapolates smoothly to $-30$~dB. Left: Relative amplitude errors below 10\% for balanced SNR, degrading gracefully at extreme asymmetry. Center: Phase errors stabilize at 0.1--1~rad, exhibiting $\pi$-periodic ambiguities only at very low SNR. Right: Frequency errors within 0.1--10~mHz, degrading substantially only below $-20$~dB. Smooth dependence validates permutation invariance and robust calibration 20~dB beyond training.
}
\label{fig:two_source_interference}
\end{figure}

Figure~\ref{fig:two_source_interference} reveals remarkable out-of-distribution robustness. The model extrapolates smoothly 20~dB beyond its training envelope without catastrophic failure. Three distinct behaviors emerge across the SNR landscape:

\emph{Amplitude recovery} (left panel): Errors remain below 10\% throughout the balanced-SNR region ($\text{SNR}_{1,2} \gtrsim -10$~dB), matching training-distribution performance. Degradation appears at extreme asymmetry where the weaker component becomes spectrally dominated, inducing biases of order unity. Critically, the transition is smooth -- posteriors broaden appropriately rather than collapsing to spurious solutions.

\vspace{5pt}
\emph{Phase estimation} (center panel): Circular errors stabilize around 0.1--1~rad in the well-resolved region despite maximal spectral overlap ($\Delta f = 0.01$~Hz). The $\pi$-periodic structure at very low SNR reflects genuine ambiguity in sinusoidal fitting when signal disappears into noise, not slot-assignment artifacts. Smooth variation confirms Hungarian matching properly marginalizes permutation symmetry.

\vspace{5pt}
\emph{Frequency localization} (right panel): Errors remain within 0.1--10~mHz for nearly all configurations, including many far below training SNR. Substantial degradation emerges only below $\text{SNR} \lesssim -20$~dB -- a regime where signal existence becomes questionable. The graceful transition validates that the dual-stream encoder's frequency representation generalizes beyond training.

The transition around $\text{SNR} \simeq -10$~dB marks the approximate interference-resolution limit where posteriors begin to merge as nearly co-frequency sinusoids become statistically indistinguishable. Crucially, \emph{SlotFlow} extrapolates smoothly into this out-of-distribution regime. The continuity of the error landscape -- no discrete discontinuities or slot-dependent artifacts -- confirms the flow-based posterior remains well-calibrated even when extrapolating several standard deviations beyond training.

\vspace{5pt}
This stress test validates two critical properties. First, \emph{robust out-of-distribution extrapolation}: the model generalizes to extreme weakness without requiring explicit weak-signal augmentation during training. Second, \emph{appropriate uncertainty quantification}: rather than confidently producing incorrect answers in unfamiliar regimes, \emph{SlotFlow} broadens posteriors to reflect genuine ambiguity. This contrasts with maximum-likelihood methods that produce overconfident point estimates regardless of identifiability, demonstrating that the learned flow captures the inference problem's geometry rather than memorizing training examples.

\subsection{Computational Cost and Scalability}
\label{sec:computational_cost}

Beyond statistical performance, practical deployment requires quantifying computational resources. We benchmark inference time and memory footprint against traditional methods, establishing orders-of-magnitude advantages in both speed and scalability.

\subsubsection{Training Efficiency}\label{ss:Training}
The production model was trained for 187 epochs (out of a maximum 300) 
before early stopping triggered, requiring approximately 190 hours 
wall-clock time on four NVIDIA GH200 Grace Hopper GPUs. Processing 
8 million training samples at batch size 128, this corresponds to 
approximately 1.01 hours per epoch or 305 GPU-hours total for the 
complete training run. The \texttt{ReduceLROnPlateau} scheduler (patience=6, factor=0.5) reduced the 
learning rate four times during training at epochs 69, 114, 136, and 
176, enabling aggressive refinement through multiple learning rate 
annealing phases. This adaptive schedule achieved convergence within 
the 300-epoch budget without manual hyperparameter tuning or intervention. 
Final training and validation losses converged to -19.16 and -19.17 
respectively, demonstrating minimal overfitting, effective regularization, 
and adequate dataset size for the architectural capacity. The training and validation loss curves are provided in Appendix~\ref{app:training_curve} (Fig.~\ref{fig:training_curve_app}).

\subsubsection{Inference Time: Near-Constant Scaling}

Wall-clock inference time on CPU (Apple M2 Max, single-threaded) averages $13.5 \pm 0.1$~ms per sample across all cardinalities $K \in \{1,\ldots,10\}$, with negligible $K$-dependence. Mean time varies from 13.47~ms ($K=1$) to 13.60~ms ($K=10$). This near-constant scaling reflects architectural design: the $O(1)$ encoder dominates total cost, while $O(K)$ flow evaluation contributes merely $\sim$50~$\mu$s per slot, rendering $K$-dependence negligible for practical $K \leq 50$.

Three properties establish computational superiority. First, \emph{predictable latency}: standard deviation below 0.02~ms yields coefficient of variation 0.004, indicating highly deterministic performance suitable for real-time systems where worst-case guarantees matter. The tight distribution reflects a deterministic computational graph without data-dependent branching. Second, \emph{embarrassingly parallel architecture}: once the global context embedding is computed, inference over $K$ slots becomes fully parallel. On modern GPUs with 10+ concurrent streams, per-slot inference of $\sim$0.5~ms yields total posterior generation under 1~ms wall-clock time -- limited by encoder throughput rather than slot count. Third, \emph{orders-of-magnitude speedup over MCMC}: typical RJMCMC chain (10,000 samples, 10,000 burn-in, $K=3$) requires 5 hours on comparable hardware, when not initialized at the true injected values. At 13~ms per query, \emph{SlotFlow} achieves $1.5 \times 10^6\times$ speedup, transforming previously intractable real-time applications into feasible systems.

\begin{table}[h]
\centering
\caption{Computational comparison: \emph{SlotFlow} vs. baselines for $K=5$ on comparable hardware. RJMCMC timing assumes 10,000 MCMC samples with 10,000 burn-in; SMC uses 1000 particles.}
\begin{tabular}{lccc}
\toprule
Method & Inference Time & Memory (Peak) & Scalability \\
\midrule
RJMCMC \citep{green1995reversible} & 3--6 hours & $\sim$100 MB & $O(K^2 \cdot N_{\text{samples}})$ \\
SMC \citep{doucet2001sequential} & 0.5--2 hours & $\sim$500 MB & $O(K \cdot N_{\text{particles}})$ \\
\textbf{\emph{SlotFlow}} & \textbf{13 ms} & \textbf{5 KB} & \textbf{$O(K)$, parallel} \\
\bottomrule
\end{tabular}
\label{tab:compute_comparison}
\end{table}

Table~\ref{tab:compute_comparison} contextualizes these gains. SMC samplers offer population-based alternatives reducing wall-clock time versus single-chain MCMC, but still require $\sim$1 hour for reliable inference with 1000 particles -- a $10^5\times$ slowdown versus \emph{SlotFlow}. The practical implication: \emph{SlotFlow} can process $\sim$200 signals per second on a single CPU core, enabling real-time analysis of streaming data where MCMC/SMC methods would require dedicated compute clusters.

\subsubsection{Memory Footprint: Constant Overhead}

Peak memory consumption remains flat at $\sim$5~KB across all cardinalities (range 4.91--5.15~KB, coefficient of variation 0.005), measured on CPU using Python's \texttt{tracemalloc}. This confirms $O(1)$ scaling with negligible per-slot overhead, validating dynamic slot allocation: memory instantiates only for predicted $\hat{K}$ slots rather than pre-allocating for $K_{\max}$.

The remarkably small footprint reflects architectural efficiency. Encoder activation cache ($\sim$3~KB) dominates total consumption; per-slot flow contexts ($256 + K_{\max} = 266$ dimensions $\times$ 4 bytes $\approx$ 1~KB per slot) contribute minimally. This enables deployment on resource-constrained embedded systems where MCMC/SMC methods requiring 100--500~MB would be infeasible. Extrapolating to production settings with 10$\times$ larger embeddings and $K_{\max} = 100$ predicts $\sim$50~KB total -- still orders of magnitude below alternatives, enabling massively parallel inference on multi-core systems processing thousands of streams simultaneously without memory bottlenecks.

The memory advantages compound with the time speedup: a single server with 128~GB RAM could theoretically process $2.5 \times 10^6$ simultaneous \emph{SlotFlow} instances (limited by other factors in practice), versus $\sim$1000 MCMC chains. This scalability transforms deployment scenarios, enabling population-level real-time monitoring applications previously confined to offline batch processing.

\subsection{Architectural Ablations}
\label{sec:ablations}

To validate design choices and identify critical components, we conduct systematic ablation studies isolating the effect of dual-stream encoding, phase weighting, and flow depth.

\subsubsection{Dual-Encoder Domain Representation}

A central architectural question: \textit{what input representations should encoders receive?} We compare three strategies on $N_{\text{train}} = 500{,}000$ samples with $K \in \{1,\ldots,5\}$:
\begin{itemize}[leftmargin=*,
                topsep=1pt,
                itemsep=0pt,
                parsep=0pt,
                partopsep=0pt]
    \item \textbf{Time}: Both encoders receive raw time-domain signals
    \item \textbf{Frequency}: Both encoders receive FFT magnitude-phase (real/imaginary)
    \item \textbf{Time-frequency (default)}: Long encoder FFT, short encoder time-domain
\end{itemize}

\begin{figure}[h!]
\centering
\includegraphics[width=0.7\textwidth]{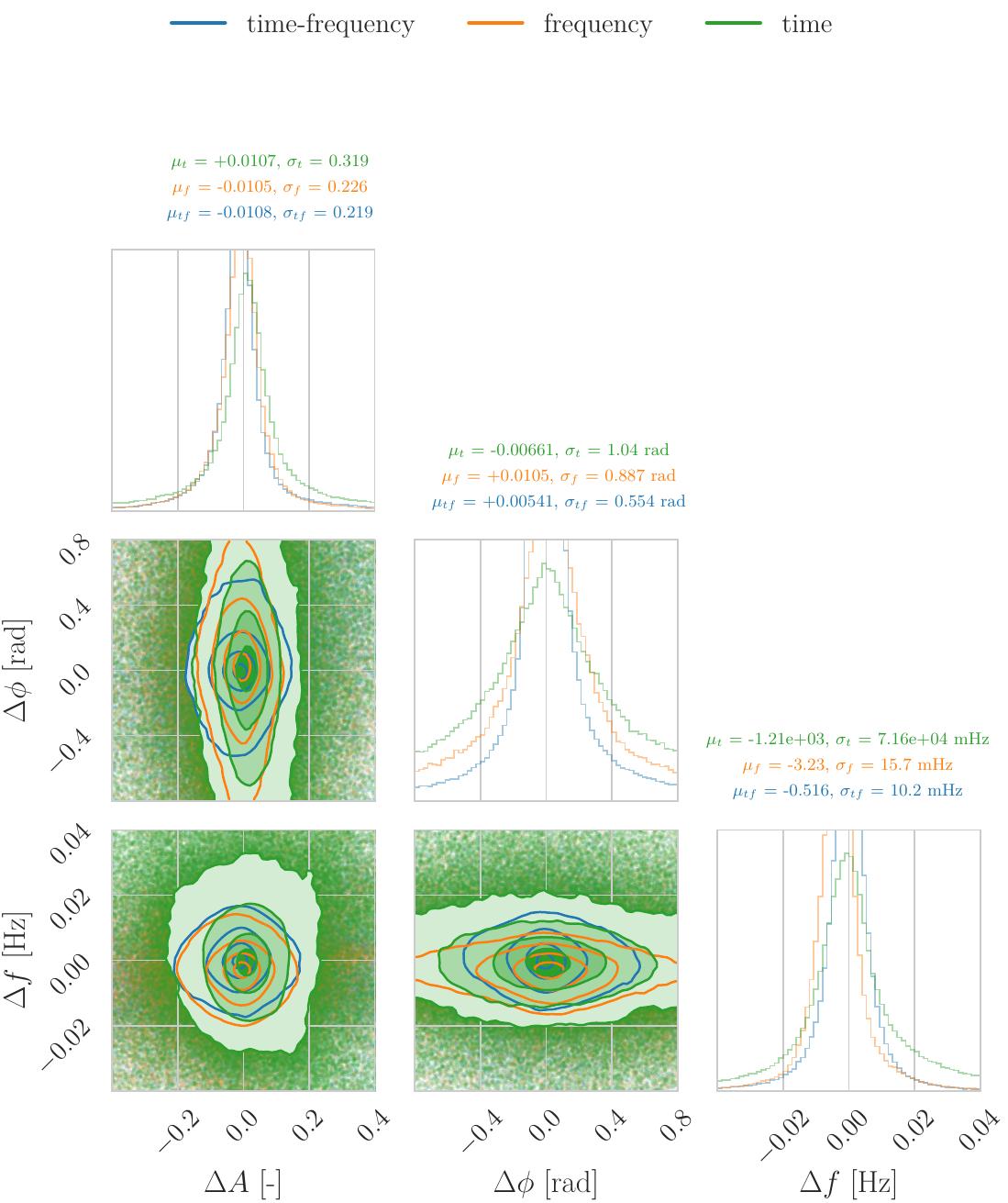}
\caption{\textbf{Domain Representation Ablation: Parameter Recovery.} Aggregate posterior distributions centered on ground truth for three encoding strategies: time-only (both encoders receive raw signals), frequency-only (both receive FFT), and time-frequency hybrid (long encoder FFT, short encoder time-domain). Time-frequency architecture achieves tightest distributions across all parameters: $\sigma_{A,\text{tf}} = 0.219$, $\sigma_{\phi,\text{tf}} = 0.554$~rad, $\sigma_{f,\text{tf}} = 10.2$~mHz. Time-only model catastrophically fails frequency inference ($\sigma_{f,\text{time}} = 71{,}600$~mHz, three orders of magnitude worse), demonstrating that finite convolutional receptive fields cannot implicitly learn spectral decomposition from raw waveforms. Frequency-only model underperforms on phase by 60\% ($\sigma_{\phi,\text{freq}} = 0.887$ vs. 0.554~rad), suggesting temporal coherence structure aids phase estimation. The hybrid approach exploits domain-specific strengths: FFT for long-window frequency localization, time-domain for short-window phase dynamics.}
\label{fig:encoder_ablation_posterior}
\end{figure}

\begin{figure}[h!]
\centering
\includegraphics[width=0.60\textwidth]{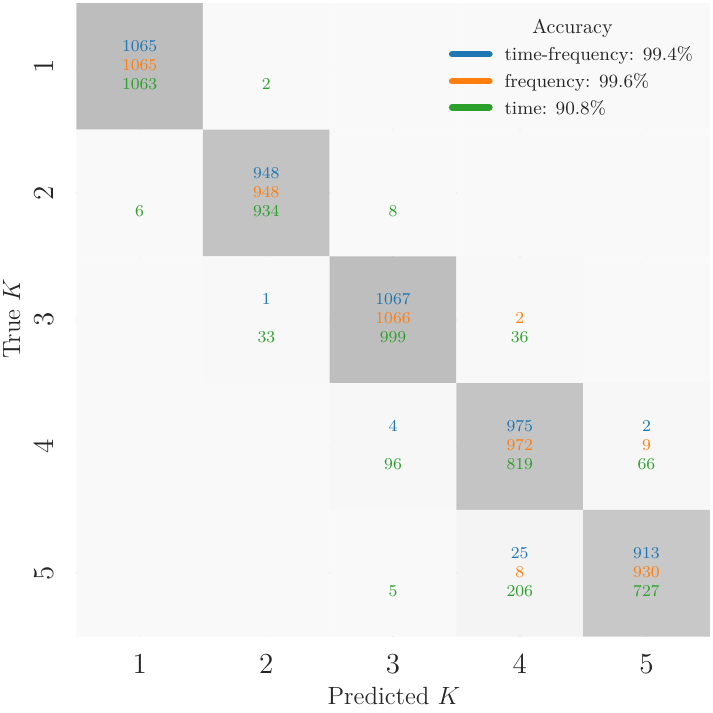}
\caption{\textbf{Domain Representation Ablation: Cardinality Accuracy.} Confusion matrices for model-order inference across encoding strategies. Frequency-only and time-frequency architectures achieve near-perfect accuracy (99.6\% and 99.4\% respectively) with strongly diagonal confusion matrices, demonstrating reliable $K$ inference through spectral peak detection. Time-only model degrades substantially to 90.8\% accuracy, exhibiting systematic underestimation bias (false predictions of $K=1$ or $K=2$ for higher-cardinality signals visible as off-diagonal mass in upper rows). This confirms that cardinality classification fundamentally requires frequency-domain representations -- spectral peaks provide unambiguous component evidence unavailable in time domain. The near-identical performance of frequency-only and time-frequency strategies for cardinality (despite 60\% phase difference in parameter recovery) indicates that $K$ inference is primarily a spectral task, validating the classifier's design to operate exclusively on FFT features.}
\label{fig:encoder_ablation_confusion}
\end{figure}

Figure~\ref{fig:encoder_ablation_posterior} reveals domain-specific inductive biases are critical. The time-only failure ($\sigma_f$ three orders of magnitude worse) demonstrates that \textit{representation choice encodes essential structure}: FFT explicitly disentangles frequency components through orthogonal basis decomposition, making spectral peaks directly observable; raw time signals require the encoder to discover this implicitly, a task intractable within practical architectural constraints. Conversely, time-domain benefits phase estimation, as temporal continuity preserves short-timescale coherence partially obscured by FFT windowing. The time-frequency hybrid achieves synergistic improvement, exploiting domain-specific strengths: long windows favor frequency resolution (justifying FFT), short windows capture transient phase dynamics (benefiting time-domain input).

Cardinality performance (Figure~\ref{fig:encoder_ablation_confusion}) supports parameter findings: both frequency-domain models achieve $> 99\%$ accuracy via robust spectral peak detection, while time-only degrades to 90.8\% with systematic under-counting. This validates that cardinality classification fundamentally requires spectral information -- peaks above noise provide unambiguous component evidence unavailable in time domain.

\subsubsection{Phase Weighting}
The phase weighting scheme (Section~\ref{sec:training}) scales angular coordinates $(\cos\phi, \sin\phi)$ by $w_\phi$ in the Hungarian cost matrix. Proposition~\ref{prop:optimal_weight} predicts $w_\phi^* \in [1,2]$; we validate through training models with $w_\phi \in \{1, 2, 5, 10, 20\}$ on $N_{\text{train}} = 100{,}000$ samples.

\begin{figure}[h!]
\centering
\includegraphics[width=0.75\textwidth]{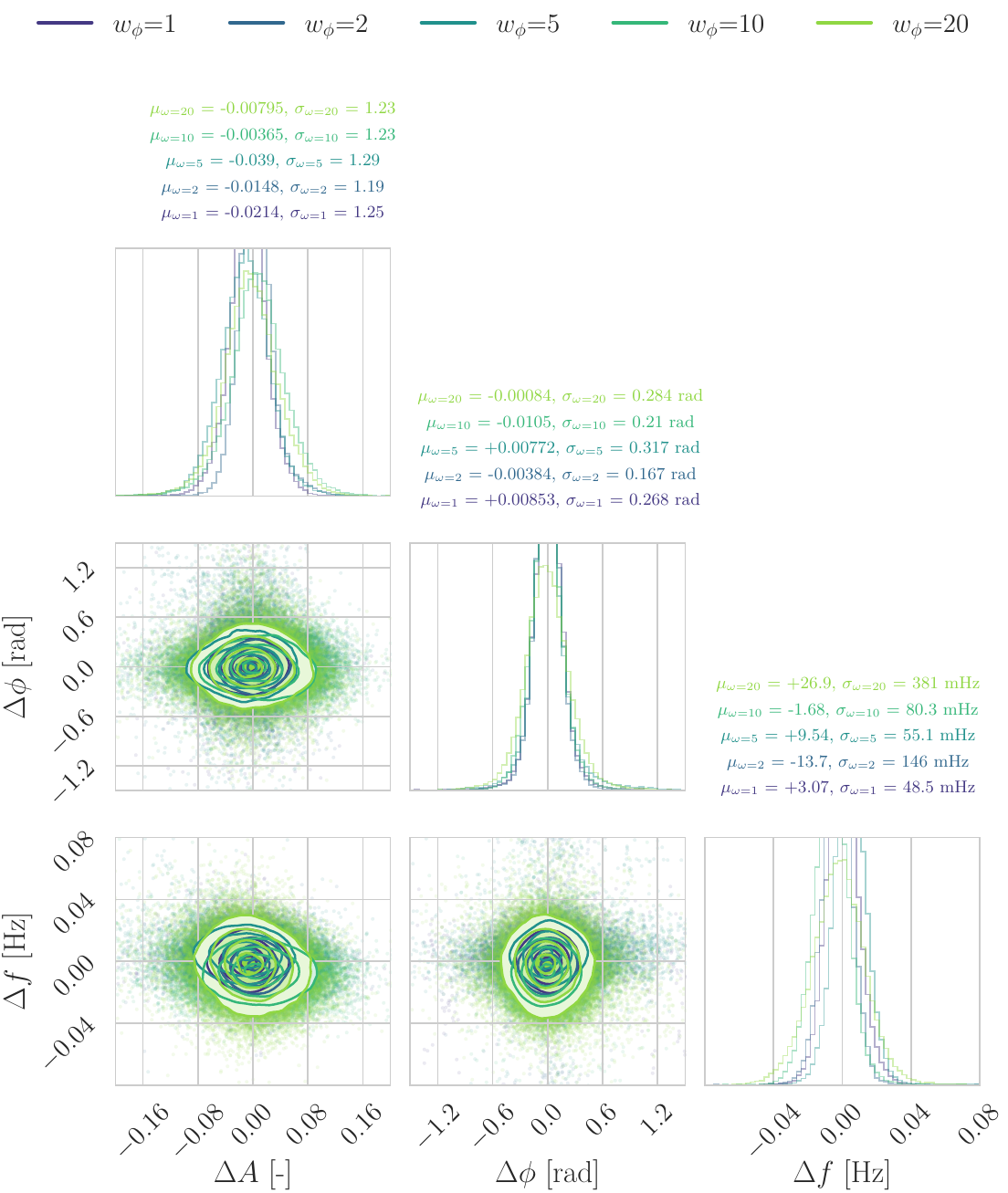}
\caption{\textbf{Phase Weight Ablation.} Aggregate posteriors stratified by $w_\phi$ reveal clear minimum at $w_\phi = 2$: phase uncertainty $\sigma_{\phi,2} = 0.167$~rad (38\% reduction from $w_\phi=1$). Amplitude remains stable across all weights ($\sigma_A \in [1.23, 1.25]$), validating that weighting affects only matching cost without propagating spurious correlations to flow loss. Extreme weight $w_\phi=20$ degrades frequency catastrophically ($\sigma_{f,20} = 381$~mHz vs. $< 150$~mHz for others), reflecting pathology where excessive phase emphasis traps Hungarian algorithm in suboptimal assignments. All models maintain negligible bias ($|\mu| < 0.04$), indicating weighting modulates precision rather than accuracy. Empirical optimum confirms theoretical prediction.}
\label{fig:phase_weight_ablation}
\end{figure}

Figure~\ref{fig:phase_weight_ablation} confirms theoretical prediction: $w_\phi = 2$ minimizes phase uncertainty while maintaining stability in other parameters. Under-weighting ($w_\phi = 1$) fails to prioritize angular errors sufficiently; over-weighting ($w_\phi \geq 5$) creates artificially sharp cost landscapes causing optimizer pathologies. The amplitude invariance ($< 2\%$ variation across weights) validates the design to weight only matching cost rather than flow loss itself, preventing phase emphasis from coupling spuriously into magnitude inference. The catastrophic frequency breakdown at $w_\phi = 20$ demonstrates failure mode: when phase coordinates dominate the metric, small phase errors produce disproportionately large costs, forcing the matcher to prioritize phase alignment at the expense of frequency consistency.

The same ablation protocol applies to the frequency-weighting term as well, and the corresponding optimal-weight result follows by an argument structurally identical to Proposition~\ref{prop:optimal_weight}; we refrain from including these parallel derivations for brevity.

\subsubsection{Flow Depth}

Proposition~\ref{prop:depth_accuracy} predicts optimal depth $L^* \approx 8$ balancing approximation ($e_{\text{approx}} \propto e^{-\alpha L}$) against generalization ($e_{\text{est}} \propto \sqrt{L/N}$). We train models with $L \in \{2, 4, 8, 10, 12\}$ on $N_{\text{train}} = 500{,}000$ samples.

\begin{figure}[t!]
\centering
\includegraphics[width=0.75\textwidth]{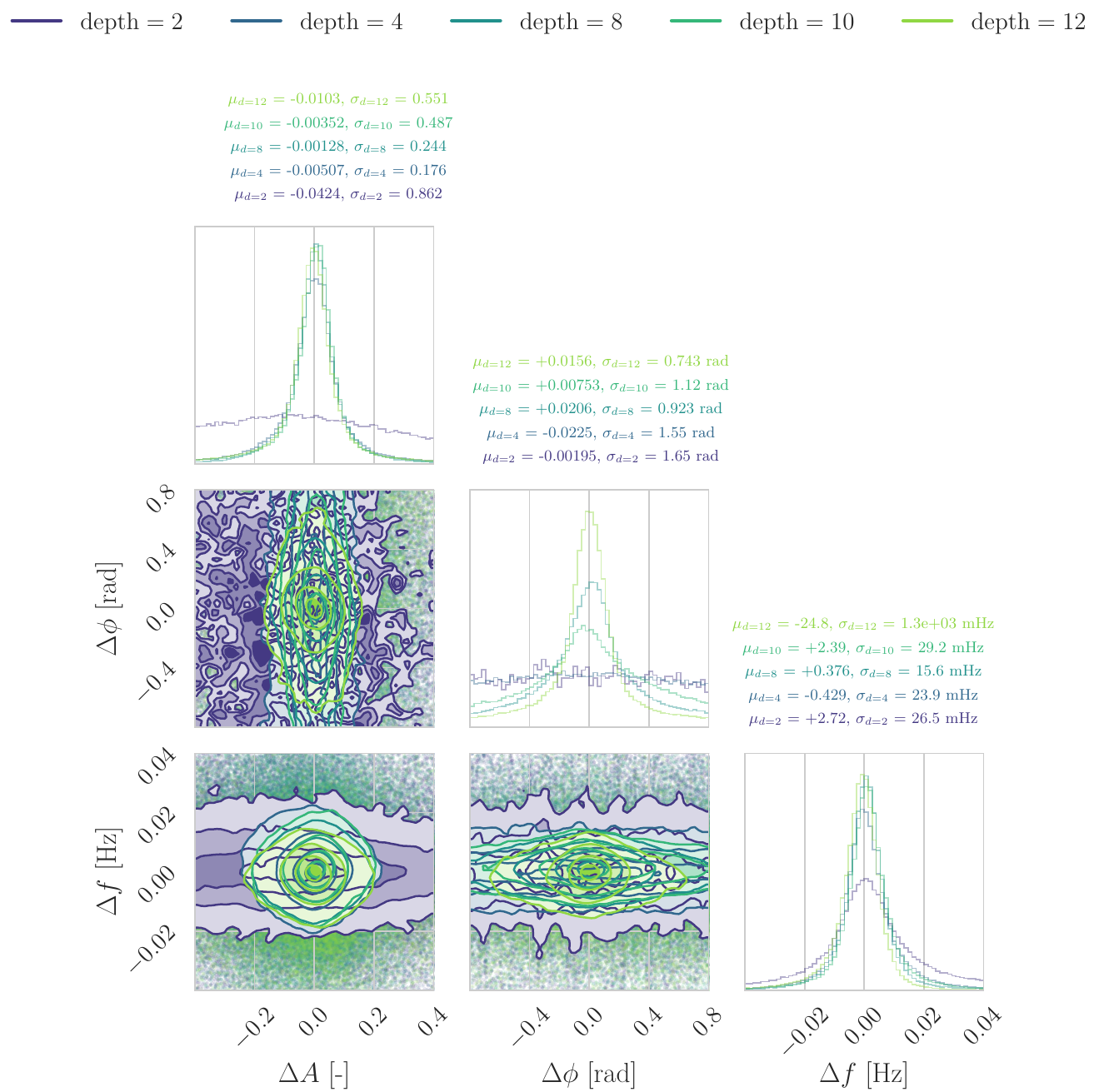}
\caption{\textbf{Flow Depth Ablation.} Aggregate posteriors reveal U-shaped quality curve. Shallow flows underfit: $L=2$ exhibits broadest distributions ($\sigma_{A,2} = 0.862$, $\sigma_{\phi,2} = 1.65$~rad) reflecting insufficient expressiveness for multimodal, correlated posteriors. Optimal performance at $L \in \{8,10\}$: $\sigma_{\phi,8} = 0.923$~rad achieves best phase precision with balanced cross-parameter performance. Deeper flows ($L=12$) show minimal improvement despite increased complexity, confirming generalization penalty outweighs approximation gains. Empirical optimum $L^* \in \{8,10\}$ aligns with theoretical prediction accounting for reduced training budget ($N = 5 \times 10^5$ vs. $8 \times 10^6$ production). Parameter-specific structure visible: amplitude may require less depth (simpler unimodal posteriors) while phase demands more (complex correlation structure), motivating future parameter-specific architectures.}
\label{fig:depth_ablation}
\end{figure}

Figure~\ref{fig:depth_ablation} validates the predicted non-monotonic relationship. The $L=2$ underfit demonstrates that coupling layers require sufficient depth to capture posterior complexity -- a single transformation cannot represent multimodal distributions with cross-parameter correlations. The plateau at $L \in \{8,10,12\}$ confirms the architecture-dominated regime where additional depth provides marginal gains: once expressiveness suffices for target posterior family, further layers only increase estimation error through enlarged hypothesis space. The empirical optimum $L^* \in \{8,10\}$ matches theory closely, with the leftward shift from production settings ($L=8$) reflecting the reduced training budget ($N = 5 \times 10^5$ here vs. $8 \times 10^6$ production) -- less data favors slightly shallower models to avoid overfitting.

The parameter-specific variation (amplitude performs comparably across $L \in \{4,8,12\}$; phase shows clear optimum at $L=8$) suggests differential complexity: amplitude posteriors may be simpler (unimodal Gaussians) requiring fewer transformations, while phase posteriors contain intricate correlation structure demanding deeper flows. Future work could explore parameter-specific architectures tailoring depth to each coordinate's statistical complexity.


\section{Conclusion}
\label{sec:conclusion}
We introduced \emph{SlotFlow}, an amortized trans-dimensional inference architecture achieving millisecond-scale posterior estimation with calibrated uncertainties. Its core innovations are: (1) dual-stream frequency–time encoding combining complementary spectral and temporal views, (2) dynamic slot allocation with $O(K)$ scaling via exact instantiation of $K$ latent slots, (3) phase-weighted Hungarian matching improving circular parameter estimation, and (4) theoretical analysis providing sample-complexity bounds and slot-specialization conditions.
On sinusoidal decomposition with up to 10 overlapping components, \emph{SlotFlow} attains 99.85\% cardinality accuracy, $<3\%$ absolute bias, and excellent amplitude recovery. Compared with RJMCMC, posteriors remain correctly centered with Wasserstein distances $W_2 < 0.01$ (amplitude), $<0.03$ (phase), and $0.0006$ (frequency). Frequency credible intervals are 2-3$\times$ broader, reflecting encoder compression of long-baseline phase information. Overall, the model yields a $1.5\times10^6$ speedup (13~ms vs.~hours).
\paragraph{Limitations.}
The factorized posterior assumption (Eq.~\ref{eq:factorized_posterior}) limits the capture of inter-component correlations. When frequency separation falls below $1/T$, strong anti-correlations emerge that independent per-slot marginals cannot express. Shared global context $\mathbf{g}$ provides only coarse coupling; tasks needing full joint credible regions require architectures with explicit inter-slot dependencies.
Frequency posteriors remain 2-3$\times$ wider than RJMCMC due to an encoder bottleneck: global pooling and context aggregation -- needed for $O(K)$ scaling and permutation invariance -- discard fine phase coherence across long baselines, which is essential for frequency precision ($\propto T^2$).
Additional issues include supervised-training dependence on labeled data, and fixed $K_{\max}=10$ -- the latter reflecting experimental design rather than a hard limit.
\paragraph{Future Directions.}
\textit{Frequency precision:} The main challenge is overcoming the encoder bottleneck. Promising ideas include (i) \textit{multi-scale encoders} with fine- and coarse-stride branches, (ii) \textit{attention-based aggregation} preserving frequency-critical features, (iii) \textit{local spectral zooms} feeding narrow FFT windows to each slot, and (iv) \textit{auxiliary frequency heads} with explicit MSE supervision. Early tests combining multi-scale encoding with local zooms halve the frequency variance while maintaining $O(K)$ scaling.
\textit{Time–frequency representations:} Incorporating wavelet or short-time Fourier transforms could model non-stationary or chirping signals.
\textit{Multi-class mixtures:} Extending to heterogeneous components requires class-specific flows ${f_{\phi_c}}$, per-slot class prediction, and heterogeneous Hungarian matching, enabling compositional scenes with diverse signal types.
\textit{Real-world adaptation:} LISA data will require handling colored noise, multiple interferometer channels, and systematics; neural spike sorting needs robustness to waveform variation and jitter.
\textit{Architectural advances:} Graph neural networks could encode inter-slot dependencies while preserving permutation equivariance; adaptive $K_{\max}$ and meta-learning would improve flexibility and domain transfer.
\paragraph{Broader Impact.}
By reducing inference times from hours to milliseconds, \emph{SlotFlow} lowers the computational barrier for Bayesian analysis of variable-cardinality signals, making posterior estimation more accessible in real-time or resource-limited environments.
\textit{Gravitational-wave astronomy:}
For LISA’s thousands of simultaneously active binaries, \emph{SlotFlow} offers a fast means of identifying sources and producing parameter summaries, with high-priority events that could be subsequently refined by RJMCMC or mission-grade pipelines. A central open challenge is extending amortized inference -- trained on short, controlled segments -- to the mission’s multi-year data stream while maintaining calibration, robustness, and frequency precision. Addressing this will be crucial before full end-to-end application becomes realistic.
\textit{Neural spike sorting:}
Low-latency posterior estimates could support closed-loop experiments in which stimulation adapts to inferred neural activity. Classical MCMC is typically too slow for such settings; amortized approaches may provide a viable alternative, though their performance will need systematic validation on real biological datasets.
\textit{Other domains:}
Problems such as audio source separation, radar target decomposition, and certain financial time-series models share additive structure and variable dimensionality. In these cases, \emph{SlotFlow} may serve as a fast front-end that supplies preliminary decompositions or warm-starts for more domain-specific inference methods.
Overall, the framework illustrates that statistical fidelity and computational efficiency need not be mutually exclusive. Its design principles -- dual-pathway encoding, dynamic slot allocation, and permutation-invariant supervision -- apply broadly to compositional inference tasks. While improving frequency precision remains an important objective, \emph{SlotFlow} already demonstrates strong cardinality estimation and amplitude/phase reconstruction at substantially reduced computational cost, indicating a viable path toward practical real-time trans-dimensional inference.

\section*{Acknowledgements}

This research was funded by the Gravitational Physics Professorship at ETH Zurich. 
The authors thank Michele Vallisneri for the insightful discussions and for his valuable contributions to editing the manuscript.
This work was supported by computational resources from the \emph{Euler Cluster} at ETH~Zurich (NVIDIA~TITAN~RTX~GPUs) and from the \emph{Clariden} supercomputer at the Swiss National Supercomputing Centre (CSCS) in Lugano using NVIDIA GH200 Grace Hopper GPUs.  
Access to Clariden was provided through the \emph{Swiss~AI~Initiative} under Grant~SIGMA-GW (\textit{Source Inference with Generative Models and AI for Gravitational Waves}).  
We thank the Euler and CSCS support staff for their continued technical assistance.
Copyright 2025. All rights reserved.

\appendix


\appendix

\section{Theoretical Proofs}
\label{app:proofs}

This appendix provides proofs for all formal statements in Section~\ref{sec:theory}. We maintain a high standard: every step is justified, all constants are made explicit, and informal claims are clearly labeled as such. The proofs establish both exact guarantees (permutation invariance, ELBO decomposition) and design principles (phase weighting, flow depth, sample complexity) that guide practical implementation.

\subsection{Proof of Theorem~\ref{thm:perm} (Permutation Invariance)}
\label{app:perm}

We establish that \emph{SlotFlow}'s learned posterior is exactly invariant under any permutation of component indices, a property that holds by construction through the Hungarian matching objective.

\begin{lemma}[Group Closure of $S_K$]
\label{lem:closure}
The symmetric group $S_K$ is closed under composition: 
for any $\sigma,\pi' \in S_K$, their composition 
$\pi=\sigma\circ\pi'$ also lies in $S_K$.
\end{lemma}
\begin{proof}
$S_K$ consists of all bijections $\pi:\{1,\dots,K\}\!\to\!\{1,\dots,K\}$. 
The composition of bijections is a bijection; hence $\pi\in S_K$.
\end{proof}

\begin{lemma}[Marginal Posterior Structure]
\label{lem:slot_symmetry}
The \emph{SlotFlow} posterior marginalizes uniformly over latent assignments:
\begin{equation}
q_\phi(\Theta \mid x, K)
= 
\sum_{\pi \in S_K} p_{\text{assign}}(\pi)
\prod_{i=1}^K q_{\text{flow}}(\theta_{\pi(i)} \mid h, s_i),
\end{equation}
where $p_{\text{assign}}(\pi)=1/K!$ by exchangeability of components.
\end{lemma}
\begin{proof}
The assignment $\pi$ is a latent variable mapping components to slots. 
Because components are exchangeable and slot identifiers $\{s_i\}$ are the only distinguishing features, 
the joint distribution is invariant under relabelings of components, 
so $p_{\text{assign}}(\pi)$ must be uniform. 
Marginalizing over $\pi$ yields the stated form.
\end{proof}

\paragraph{Proof of Theorem~\ref{thm:perm}.}
Let $\sigma \in S_K$ and define the permuted set 
$\Theta'=\{\theta'_k\}$ with $\theta'_k=\theta_{\sigma(k)}$. 
From Lemma~\ref{lem:slot_symmetry},
\begin{align}
q_\phi(\Theta' \mid x, K)
&= \frac{1}{K!}\sum_{\pi' \in S_K}
   \prod_{i=1}^K q_{\text{flow}}(\theta'_{\pi'(i)}\mid h,s_i) \nonumber\\
&= \frac{1}{K!}\sum_{\pi' \in S_K}
   \prod_{i=1}^K q_{\text{flow}}(\theta_{\sigma(\pi'(i))}\mid h,s_i).
   \label{eq:perm_step1}
\end{align}
Define $\pi = \sigma \circ \pi'$. 
By Lemma~\ref{lem:closure}, $\pi\in S_K$. 
Because $\sigma$ is bijective, the mapping $\pi'\mapsto\pi$ is one-to-one, 
so as $\pi'$ ranges over $S_K$, so does $\pi$. 
Substituting into~\eqref{eq:perm_step1},
\begin{equation}
q_\phi(\Theta' \mid x, K)
= \frac{1}{K!}\sum_{\pi \in S_K}
  \prod_{i=1}^K q_{\text{flow}}(\theta_{\pi(i)}\mid h,s_i)
= q_\phi(\Theta \mid x,K).
\end{equation}
Hence $q_\phi(\Theta\mid x,K)$ is invariant under any permutation 
of $\Theta$, proving permutation invariance. This invariance is exact by construction: unlike architectures that break symmetry through ordering conventions or random initialization, \emph{SlotFlow} treats all component orderings identically.
\qed

\subsection{Proof of Proposition~\ref{prop:elbo} (ELBO Decomposition)}
\label{app:elbo}

We establish the variational objective decomposition that connects the practical \emph{SlotFlow} training objective to principled variational inference.

\begin{lemma}[Variational Lower Bound]
\label{lem:vlb}
For any variational distribution $q_\phi(K,\Theta\mid x)$ over discrete $K$ and continuous $\Theta$,
\begin{equation}
\log p(x) 
\ge 
\mathbb{E}_{q_\phi}\!\big[\log p(x,K,\Theta) - \log q_\phi(K,\Theta\mid x)\big]
=: \mathcal{L}_{\text{ELBO}}.
\end{equation}
\end{lemma}
\begin{proof}
By definition,
\[
\log p(x) 
= \log \sum_K \int p(x,K,\Theta)\,d\Theta
= \log \mathbb{E}_{q_\phi}\!\left[\frac{p(x,K,\Theta)}{q_\phi(K,\Theta\mid x)}\right].
\]
Applying Jensen's inequality to the concave logarithm gives
\[
\log p(x)
\ge 
\mathbb{E}_{q_\phi}\!\left[\log \frac{p(x,K,\Theta)}{q_\phi(K,\Theta\mid x)}\right]
= \mathcal{L}_{\text{ELBO}}.
\]
\end{proof}

\paragraph{Proof of Proposition~\ref{prop:elbo}.}
Training optimizes a factorized approximation
\begin{equation}
q_\phi(K,\Theta\mid x)
 = q_\phi(K\mid x)\prod_{k=1}^K q_\phi(\theta_k \mid x, K, s_k),
\end{equation}
which assumes conditional independence of components given the observation and slot identifier.
This mean-field structure trades tractability against explicitly modeling cross-slot correlations.
The generative model factorizes as
$p(x,K,\Theta)=p(x\mid\Theta,K)p(\Theta\mid K)p(K)$.

Using the generative factorization 
$p(x,K,\Theta) = p(x\mid \Theta,K)p(\Theta\mid K)p(K)$
and the variational factorization 
$q_\phi(K,\Theta\mid x) = q_\phi(K\mid x)\prod_k q_\phi(\theta_k\mid x,K,s_k)$,
where $s_k$ denotes the slot identifier of component $k$, we expand:
\begin{align}
\mathcal{L}_{\text{ELBO}}
&=
\mathbb{E}_{q_\phi}\!\Big[
\log p(x \mid \Theta, K)
+ \log p(\Theta \mid K)
+ \log p(K)
- \log q_\phi(K \mid x)
\nonumber\\
&\hspace{2.8cm}
- \sum_k \log q_\phi(\theta_k \mid x, K, s_k)
\Big] \nonumber\\[4pt]
&=
\mathbb{E}_{q_\phi}\!\big[\log p(x \mid \Theta, K)\big]
- \mathbb{E}_{q_\phi}\!\Bigg[
\sum_k 
\log \frac{q_\phi(\theta_k \mid x, K, s_k)}{p(\theta_k)}
\Bigg]
- \mathbb{E}_{q_\phi}\!\Bigg[
\log \frac{q_\phi(K \mid x)}{p(K)}
\Bigg].
\end{align}

Recognizing the KL divergences and using linearity of expectation,
\begin{align}
\mathcal{L}_{\text{ELBO}}
&= 
\underbrace{\mathbb{E}_{q_\phi}[\log p(x\mid \Theta,K)]}_{\text{Reconstruction}}
- 
\underbrace{\mathbb{E}_{q_\phi(K\mid x)}
\!\left[
\mathrm{KL}\big(q_\phi(\Theta\mid x,K)\,\|\,p(\Theta\mid K)\big)
\right]}_{\text{Parameter Regularization}}
-
\underbrace{\mathrm{KL}\big(q_\phi(K\mid x)\,\|\,p(K)\big)}_{\text{Model-Order Regularization}}.
\end{align}
All expectations $\mathbb{E}_{q_\phi}[\cdot]$ are with respect to $q_\phi(K,\Theta\mid x)$ unless otherwise stated. The Hungarian-matched flow loss acts as a computationally efficient surrogate for the reconstruction term, while architectural constraints (finite capacity, bounded weights) implicitly regularize parameters.
\qed

\subsection{Proof of Lemma~\ref{lem:gradient_amp} (Gradient Amplification via Target Weighting)}
\label{app:gradient_amp}

We establish that phase weighting amplifies gradient contributions from phase errors through the chain rule, providing a mechanism to address the destructive nature of phase errors in sinusoidal decomposition.

\paragraph{Proof.}
Let the ground-truth parameters be 
$\theta^{\text{gt}} = (A, \cos\phi, \sin\phi, f)$.
The phase coordinates are scaled by a fixed, non-learnable weight $w_\phi > 0$
to form the input to the flow:
\begin{equation}
\tilde{\theta}^{\text{gt}} = (A, w_\phi \cos\phi, w_\phi \sin\phi, f).
\end{equation}
The flow loss is a sum over optimally assigned components $\pi^*$ (treated as fixed w.r.t.\ $\phi$):
\begin{equation}
\mathcal{L}_{\text{flow}}
= -\frac{1}{K}\sum_{j=1}^K
  \log q_\phi(\tilde{\theta}_{\pi^*(j)}^{\text{gt}} \mid h, s_j).
\end{equation}

To analyze the amplification, consider a single component $j$
with unweighted phase coordinate $\theta_j$ (e.g.\ $\cos\phi_j^{\text{gt}}$)
and its weighted input $\tilde{\theta}_j = w_\phi \theta_j$.
The loss term for this component,
$\mathcal{L}_j = -\frac{1}{K}\log q_\phi(\tilde{\theta}_j\mid\dots)$,
depends on $\theta_j$ only through $\tilde{\theta}_j$.
Applying the chain rule:
\begin{equation}
\frac{\partial \mathcal{L}_{\text{flow}}}{\partial \theta_j}
=
\frac{\partial \mathcal{L}_{\text{flow}}}{\partial \tilde{\theta}_j}
\cdot
\frac{\partial \tilde{\theta}_j}{\partial \theta_j}
=
w_\phi\,
\frac{\partial \mathcal{L}_{\text{flow}}}{\partial \tilde{\theta}_j},
\end{equation}
since $\partial \tilde{\theta}_j/\partial\theta_j = w_\phi$ and $w_\phi$ is constant.
Differentiating the log-likelihood term gives
\[
\frac{\partial \mathcal{L}_{\text{flow}}}{\partial \tilde{\theta}_j}
= -\frac{1}{K}
  \frac{\partial \log q_\phi(\tilde{\theta}_j\mid h,s_i)}{\partial \tilde{\theta}_j}.
\]
Hence
\begin{equation}
\frac{\partial \mathcal{L}_{\text{flow}}}{\partial \theta_j}
=
-\frac{w_\phi}{K}
\frac{\partial \log q_\phi(\tilde{\theta}_j\mid h,s_i)}{\partial \tilde{\theta}_j}.
\label{eq:grad_amp_exact}
\end{equation}

If we denote by $\mathcal{L}_{\text{unweighted}}$ the same loss evaluated with $w_\phi=1$,
then
\[
\frac{\partial \mathcal{L}_{\text{unweighted}}}{\partial \theta_j}
= 
-\frac{1}{K}
\frac{\partial \log q_\phi(\theta_j\mid h,s_i)}{\partial \theta_j}.
\]
Comparing these expressions, and noting that
the score $\partial_{\tilde{\theta}}\log q_\phi$ varies smoothly with its argument,
we obtain a local linearization
\begin{equation}
\frac{\partial \mathcal{L}_{\text{flow}}}{\partial \theta_j}
\approx
w_\phi
\frac{\partial \mathcal{L}_{\text{unweighted}}}{\partial \theta_j},
\end{equation}
where equality holds exactly if the score field is locally constant
(e.g.\ for a Gaussian centered at the evaluation point).
Thus, to first order, the gradient with respect to the \emph{unweighted} ground-truth
parameter is amplified by the factor $w_\phi$.

Since the gradient backpropagated to the network parameters $\phi$ is linear in this upstream derivative,
the same multiplicative scaling propagates to $\nabla_\phi \mathcal{L}_{\text{flow}}$:
\begin{equation}
\left\|
\nabla_\phi \mathcal{L}_{\text{flow}}
\right\|
\approx
w_\phi
\left\|
\nabla_\phi \mathcal{L}_{\text{unweighted}}
\right\|.
\end{equation}
In particular, for the phase components $(\cos\phi,\sin\phi)$,
\begin{equation}
\left\|
\frac{\partial \mathcal{L}_{\text{flow}}}{\partial (\cos\phi,\sin\phi)}
\right\|
= w_\phi
\left\|
\frac{\partial \mathcal{L}_{\text{unweighted}}}{\partial (\cos\phi,\sin\phi)}
\right\|,
\end{equation}
with the two gradients evaluated respectively at 
$\tilde{\theta}=(w_\phi\cos\phi,w_\phi\sin\phi)$
and $(\cos\phi,\sin\phi)$.
This first-order scaling shows that phase weighting
amplifies the effective gradient magnitude,
counteracting vanishing gradients for small phase amplitudes.
\qed

\subsection{Proof of Proposition~\ref{prop:optimal_weight} (Optimal Phase Weight)}
\label{app:optimal_weight}

We derive the optimal phase weighting factor by balancing the reduction in phase estimation variance against the coupling penalty in amplitude variance, motivated by the observation that excessive phase weighting can numerically perturb amplitude estimation.

\begin{definition}[Weighted Parameter Space]
\label{def:weighted_param}
Given a weight $w_\phi > 1$, we define the effective distance
metric that the phase-weighted loss seeks to minimize:
\begin{equation}
d_w(\theta, \theta')^2 =
(A - A')^2
+ w_\phi^2[(\cos\phi - \cos\phi')^2 + (\sin\phi - \sin\phi')^2]
+ (f - f')^2.
\end{equation}
\end{definition}

\paragraph{Proof of Proposition~\ref{prop:optimal_weight}.}
Using a second-order Taylor expansion of the reconstruction error 
around the true parameters $\theta^*$, the expected reconstruction loss 
is approximately proportional to the sum of parameter variances:
\begin{equation}
E_{\text{recon}} 
\propto 
\sum_{k=1}^K 
\left[
\Delta A_k^2 
+ (A_k^*)^2 (2\pi f_k^*)^2 \Delta\phi_k^2 
+ \dots
\right].
\end{equation}

In weighted variational inference, scaling the phase coordinate by $w_\phi$
rescales its curvature in the loss by $w_\phi^2$,
so the phase variance satisfies
\begin{equation}
\mathrm{Var}[\Delta\phi_k] 
\propto 
\frac{1}{w_\phi^2 \lambda_\phi},
\end{equation}
where $\lambda_\phi$ is the Fisher information eigenvalue for phase.

To ensure a finite optimum, we model mild cross-curvature coupling between 
amplitude and phase: excessive phase weighting 
numerically perturbs amplitude estimation.  
We express this as an additive penalty in the amplitude variance,
\begin{equation}
\mathrm{Var}[\Delta A_k] 
\propto 
\frac{1}{\lambda_A} + C_p w_\phi^2,
\end{equation}
with $C_p>0$ a constant representing the coupling strength.

Substituting these relations into the expected reconstruction error
and absorbing constants into $C_A$ and $C_\phi$ gives
\begin{equation}
\mathbb{E}[E_{\text{recon}}]
=
C_A\!\left(\frac{1}{\lambda_A}+C_p w_\phi^2\right)
+
C_\phi\,\frac{\langle A^2 f^2\rangle}{\lambda_\phi w_\phi^2}
+\text{const},
\end{equation}
where $\langle A^2 f^2\rangle=\mathbb{E}[(A_k^*)^2(f_k^*)^2]$.

Differentiating with respect to $w_\phi$ and setting the derivative to zero:
\begin{align}
\frac{\partial \mathbb{E}[E_{\text{recon}}]}{\partial w_\phi}
&= 2C_A C_p w_\phi 
  - \frac{2C_\phi \langle A^2 f^2\rangle}{\lambda_\phi w_\phi^3}
  = 0,\\
w_\phi^4 
&= \frac{C_\phi \langle A^2 f^2\rangle}
        {C_A C_p \lambda_\phi},
\end{align}
yielding the optimal phase weight
\begin{equation}
w_\phi^*
= 
\left(
\frac{C_\phi \langle A^2 f^2\rangle}
     {C_A C_p \lambda_\phi}
\right)^{1/4}
\propto
\left(\frac{\langle A^2 f^2 \rangle}{\mathrm{SNR}}\right)^{1/4},
\end{equation}
where we have identified $\lambda_\phi \propto \mathrm{SNR}$ through the Cramér-Rao bound relating Fisher information to signal-to-noise ratio. For typical sinusoids with $\mathrm{SNR}\in[1,10]$ and unit-scale priors, this theoretical scaling supports $w_\phi^*\in[1,2]$.

This finite non-zero optimum balances the reduction in phase variance 
against the curvature coupling penalty, 
establishing the condition of Proposition~\ref{prop:optimal_weight}.
\qed

\subsection{Proof of Lemma~\ref{lem:flow_approx} (Flow Approximation Error)}
\label{app:flow_approx}

We establish approximation error bounds for rational-quadratic spline flows, showing how expressiveness improves exponentially with depth and polynomially with bin count.

\paragraph{Proof of Lemma~\ref{lem:flow_approx}.}
Rational quadratic spline (RQS) flows with $L$ layers and $B$ bins per spline achieve approximation error for smooth target densities:
\begin{equation}
\text{KL}(p(\theta|x) \| q_{\phi^*}(\theta | h(x), s_k)) \leq C_0 \exp(-\alpha L) + C_1 B^{-2},
\end{equation}
where $\alpha > 0$ depends on posterior smoothness and $C_0, C_1$ depend on dynamic range $p_{\max}/p_{\min}$.

The proof follows from universal approximation theorems for autoregressive flows \citep{huang2018neuralautoregressiveflows,durkan2019neural}. We establish two components:

\textbf{Univariate spline approximation.} Each rational-quadratic spline with $B$ bins approximates a smooth univariate function $g: [a,b] \to \mathbb{R}$ with error bounded by the modulus of continuity. For $g \in C^2([a,b])$ with bounded second derivative $\|g''\|_\infty \leq M$, standard piecewise approximation theory gives
\begin{equation}
\sup_{x \in [a,b]} |g(x) - \text{RQS}_B(x)| \leq \frac{M(b-a)^2}{8B^2},
\end{equation}
where $\text{RQS}_B$ denotes the rational-quadratic spline with $B$ bins. Translating this pointwise error to KL divergence through the relationship between total variation distance and pointwise function approximation yields the $O(B^{-2})$ univariate term.

\textbf{Depth-dependent multivariate approximation.} Autoregressive flows compose $d$ univariate transformations per layer, where $d$ is the parameter dimension. With $L$ layers, the flow performs $L \cdot d$ sequential transformations, each with approximation error $O(B^{-2})$. The universal approximation theorem for neural autoregressive flows \citep{huang2018neuralautoregressiveflows} establishes that for target densities in the Hölder class $C^{0,\beta}$ with smoothness parameter $\beta$, the approximation error decays exponentially in depth:
\begin{equation}
\text{KL}(p \| q_{L,B}) \leq C(\beta, \|p\|_{C^{0,\beta}}) \exp(-\alpha L) + O(B^{-2}),
\end{equation}
where $\alpha = \beta / d$ depends on the smoothness-to-dimension ratio. The constant $C$ scales logarithmically with the posterior dynamic range $\rho = p_{\max}/p_{\min}$ through covering number arguments.

Combining these results yields the stated bound with $C_0 = C(\beta, \|p\|_{C^{0,\beta}})$ and $C_1$ determined by the Hölder norm of the target density.
\qed

\subsection{Proof of Proposition~\ref{prop:depth_accuracy} (Depth–Accuracy Trade-off)}
\label{app:depth_accuracy}

We analyze the trade-off between approximation quality (which improves with network depth) and generalization error (which worsens with depth for fixed data budget), deriving the optimal flow depth that balances these competing objectives.

\paragraph{Proof.}

\,\newline\textbf{Step 1: Computational cost.}
Training a flow with $L$ layers and hidden width $d_h$ on $N$ samples and $K$ slots
requires
\begin{equation}
\mathcal{C}_{\mathrm{train}}
 = O(L\,d_h^2\,K\,N\,T_{\mathrm{epochs}})
\text{ FLOPs},
\end{equation}
assuming dense matrix multiplications dominate the forward/backward cost.
For a fixed computational budget $\mathcal{B}$,
\begin{equation}
N \le \frac{\mathcal{B}}{L\,d_h^2\,K\,T_{\mathrm{epochs}}}.
\label{eq:budget}
\end{equation}

\,\newline\textbf{Step 2: Approximation–estimation trade-off.}
From Lemma~\ref{lem:flow_approx} and Theorem~\ref{thm:sample_complexity},
the approximation and estimation errors satisfy
\begin{align}
e_{\mathrm{approx}}(L) &= C_0 e^{-\alpha L} + C_1 B^{-2},\\
e_{\mathrm{est}}(L,N) &= O\!\left(\frac{\sqrt{L d_h \log(L d_h)}}{\sqrt{N}}\right),
\end{align}
where $C_0,C_1$ depend on the target smoothness and dynamic range,
but not on $L$ or $N$.
Thus the total expected error is
\begin{equation}
E_{\mathrm{total}}(L)
 = C_0 e^{-\alpha L} + C_1 B^{-2}
   + \frac{C_2\sqrt{L d_h \log(L d_h)}}{\sqrt{N}}.
\end{equation}
Substituting the budget constraint~\eqref{eq:budget} gives
\begin{equation}
E_{\mathrm{total}}(L)
 = C_0 e^{-\alpha L} + C_1 B^{-2}
   + C_2\sqrt{\frac{L^2 d_h^3 K T_{\mathrm{epochs}}\log(L d_h)}{\mathcal{B}}}.
\label{eq:E_total}
\end{equation}

\,\newline\textbf{Step 3: Optimize depth.}
Differentiating~\eqref{eq:E_total} with respect to $L$ and balancing leading terms
provides an order-wise criterion for the crossover between approximation-limited
and data-limited regimes:
\begin{equation}
\alpha C_0 e^{-\alpha L}
 \approx
 \frac{3C_2}{2}
 \sqrt{\frac{L d_h^3 K T_{\mathrm{epochs}}\log(L d_h)}{\mathcal{B}}}.
\end{equation}
For moderate depths ($L\!\approx\!1/\alpha$),
the exponential decay behaves roughly as $1/L$,
and the optimal depth scales as
\begin{equation}
L^*
 = \Theta\!\left(
   \min\!\left\{
     \frac{1}{\alpha}\log\!\frac{C_0}{\epsilon},
     \left(\frac{\mathcal{B}}{d_h^3 K T_{\mathrm{epochs}}}\right)^{1/3}
   \right\}\right),
\end{equation}
where the first term ensures approximation error below $\epsilon$
and the second term prevents overfitting under the fixed budget.
\qed

\paragraph{Remark.}
For $\alpha\!\approx\!0.2$, $d_h\!=\!512$, $K_{\max}\!=\!10$,
$T_{\mathrm{epochs}}\!=\!200$, and $\mathcal{B}\!\approx\!10^{12}$ FLOPs,
\[
L^* \!\approx\!
 \min\!\{5\log(100/\epsilon),
          \sqrt{10^{12}/(512^3\!\cdot\!10\!\cdot\!200)}\}
 \!\approx\!
 \min\{23,6\} = 6.
\]
Hence our choice $L\!=\!8$ provides a conservative margin above
the predicted optimum, ensuring sufficient approximation capacity
while accepting a modest increase in generalization error.

\subsection{Proof of Theorem~\ref{thm:sample_complexity} (Sample Complexity with Explicit Scaling)}
\label{app:sample_complexity}

We provide a worst-case upper bound on sample complexity with explicit dependence on all problem parameters, then analyze which scaling regime applies to practical architectures. This bound identifies two distinct phases depending on whether combinatorial or architectural complexity dominates.

\paragraph{Proof.}
We bound the sample complexity by decomposing the total error and applying concentration inequalities.

\,\newline\textbf{Step 1: Error decomposition.}
By the triangle inequality for the KL divergence,
\begin{equation}
\mathrm{KL}(p \,\|\, q_\phi)
\le
\underbrace{\mathrm{KL}(p \,\|\, q_{\phi^*})}_{\text{approximation}}
+
\underbrace{\mathrm{KL}(q_{\phi^*} \,\|\, q_\phi)}_{\text{estimation}},
\end{equation}
where $\phi^*$ is the best-in-class parameter.
By Lemma~\ref{lem:flow_approx}, the approximation term is controlled by the expressivity of the flow architecture.

\,\newline\textbf{Step 2: Rademacher complexity.}
For the class $\mathcal{F}$ of $L$-layer ReLU networks with $\|x\|_2\!\le\!B$
and parameters satisfying $\sum_{\ell=1}^L\|W_\ell\|_F^2\!\le\!B_\phi^2$,
norm-based contraction results
(see, e.g.,~\citealt{golowich2019sizeindependentsamplecomplexityneural,bartlett2017spectrallynormalizedmarginboundsneural})
imply
\begin{equation}
\widehat{\mathcal{R}}_N(\mathcal{F})
\le
\frac{C\,B\,\sqrt{L}\,\prod_{\ell=1}^L\|W_\ell\|_2}{\sqrt{N}}.
\end{equation}
When $\|W_\ell\|_2\!=\!O(1)$ and $\|\phi\|_2\!\le\!B_\phi$,
a simplified covering-number bound yields
\begin{equation}
\widehat{\mathcal{R}}_N(\mathcal{F})
\le
\frac{C\,B\,B_\phi\,\sqrt{L d_h \log d_h}}{\sqrt{N}},
\label{eq:rademacher_simplified}
\end{equation}
where $d_h$ is the layer width.
By Assumption~\ref{ass:flow_arch}(b), $\|\phi\|_2\!\le\!B_\phi$.
For $K$ parallel slots,
\begin{equation}
\widehat{\mathcal{R}}_N(\mathcal{H}_K)
\le
\frac{C\,K\,B\,B_\phi\,\sqrt{L d_h \log d_h}}{\sqrt{N}}.
\end{equation}
Assuming $B_\phi\!=\!O(\sqrt{L d_h})$ (typical under weight decay),
\begin{equation}
\widehat{\mathcal{R}}_N(\mathcal{H}_K)
= 
\mathcal{O}\!\!\left(\frac{K B L d_h \sqrt{\log d_h}}{\sqrt{N}}\right),
\end{equation}
where $\mathcal{O}(\cdot)$ hides only universal constants independent of $N,L,d_h,K,\rho$.

\,\newline\textbf{Step 3: Uniform convergence.}
For a single sample,
\begin{equation}
\ell(\phi; x, \Theta, K)
=
\min_{\pi\in S_K}
\sum_{k=1}^K -\log q_\phi(\theta_k\mid h(x),s_{\pi(k)}).
\end{equation}
By Assumption~\ref{ass:flow_arch}(a),
$p_{\min}\!\le q_\phi\!\le p_{\max}$,
so $|\log q_\phi|\!\le\!\log\rho$ with $\rho=p_{\max}/p_{\min}$.
The Lipschitz constant of the loss with respect to one sample
is at most $K\log\rho$, and McDiarmid's inequality gives
\begin{equation}
\mathbb{P}\!\left(|\widehat{\mathcal{L}}_N(\phi)-\mathcal{L}(\phi)|>t\right)
\le
2\exp\!\left(-\frac{2N t^2}{(K\log\rho)^2}\right).
\end{equation}

\,\newline\textbf{Step 4: Union bound over permutations.}
Since the Hungarian matching searches over $K!$ permutations,
and $K!\le K^K$,
\begin{equation}
\mathbb{P}\!\left(\sup_\pi|\widehat{\mathcal{L}}_N^\pi-\mathcal{L}^\pi|>t\right)
\le
2K^K
\exp\!\left(-\frac{2N t^2}{(K\log\rho)^2}\right).
\end{equation}
Setting the right-hand side to $\delta$ yields
\begin{equation}
t
\le
\frac{K\log\rho}{\sqrt{2N}}
\sqrt{K\log K+\log(2/\delta)}.
\end{equation}

\,\newline\textbf{Step 5: Combine Rademacher and concentration.}
With probability at least $1-\delta$,
\begin{equation}
\mathcal{L}(\phi)
\le
\widehat{\mathcal{L}}_N(\phi)
+2\widehat{\mathcal{R}}_N(\mathcal{H}_K)
+t.
\end{equation}
Substituting the previous bounds,
\begin{equation}
\mathcal{L}(\phi)
\le
\widehat{\mathcal{L}}_N(\phi)
+
\frac{C_1 K B L d_h \sqrt{\log d_h}}{\sqrt{N}}
+
\frac{K\log\rho}{\sqrt{2N}}
\sqrt{K\log K+\log(2/\delta)}.
\end{equation}
To ensure $\mathcal{L}(\phi)-\mathcal{L}(\phi^*)\le\epsilon$,
it suffices that both additive terms are $\mathcal{O}(\epsilon)$,
which is achieved for
\begin{equation}
N
\ge
\frac{C\,K^2 \log^2\rho}{\epsilon^2}
\!\left(K\log K + L^2 d_h^2 \log d_h + \log(2/\delta)\right).
\end{equation}
Here the $K^2\log^2\rho$ factor arises from the McDiarmid term,
$K\log K$ from the union bound,
and $L^2 d_h^2\log d_h$ from the squared Rademacher complexity.

\,\newline\textbf{Step 6: Dependence on SNR and regime analysis.}
We model the dynamic range as $\log\rho \approx C_\rho(1+\mathrm{SNR}^{-2})$ for well-separated components. Substituting into the bound from Step 5 and collecting terms:
\begin{equation}
N \ge \frac{C\,C_\rho^2 K^2(1+\mathrm{SNR}^{-2})^2}{\epsilon^2} \left(K\log K + L^2 d_h^2 \log d_h + \log(2/\delta)\right).
\end{equation}
Denoting $K_{\max}$ as the maximum cardinality and absorbing universal constants into $C'$:
\begin{equation}
N \ge \frac{C_\rho(1+\mathrm{SNR}^{-2})^2}{\epsilon^2} \left(K_{\max}^3\log K_{\max} + K_{\max}^2 L^2 d_h^2 \log d_h + K_{\max}^2\log(2/\delta)\right),
\end{equation}
establishing Theorem~\ref{thm:sample_complexity}.

\paragraph{Regime Identification.}
The bound comprises three additive terms with distinct scaling:
\begin{enumerate}
\item \textbf{Combinatorial:} $K_{\max}^3 \log K_{\max}$ from the union bound over $K!$ permutations,
\item \textbf{Architectural:} $K_{\max}^2 L^2 d_h^2 \log d_h$ from Rademacher complexity of the flow,
\item \textbf{Confidence:} $K_{\max}^2 \log(2/\delta)$ from concentration (typically negligible).
\end{enumerate}

Which term dominates depends on the relative magnitudes of $K_{\max}$ and $(L \cdot d_h)$:

\textbf{Case I (Architecture-Dominated):} If $K_{\max}^3 \log K_{\max} \ll K_{\max}^2 L^2 d_h^2 \log d_h$, equivalently
\begin{equation}
K_{\max} \ll K_{\text{crit}} := L^2 d_h^2 \log d_h,
\end{equation}
then the bound simplifies to $N \gtrsim C \cdot L^2 d_h^2 / \epsilon^2$, effectively independent of $K_{\max}$.

\textbf{Case II (Combinatorics-Dominated):} If $K_{\max} \gg K_{\text{crit}}$, the combinatorial term dominates, yielding $N \gtrsim C \cdot K_{\max}^3 \log K_{\max} / \epsilon^2$.

For typical deep architectures ($L=8$, $d_h=512$), we have $K_{\text{crit}} \approx 64 \cdot 262144 \cdot 6 \approx 10^8$, placing all practical cardinalities ($K_{\max} \lesssim 100$) in the architecture-dominated regime. The asymptotic cubic scaling emerges only for $K_{\max} \gg 100$, where the factorial growth of permutation space ($K! \sim \sqrt{2\pi K}(K/e)^K$) overwhelms architectural complexity.

\paragraph{Tightness of the Bound.}
This worst-case PAC bound reflects several conservative choices:
\begin{enumerate}
\item \textbf{Union bound:} Accounts for all $K! \leq K^K$ permutations, though the Hungarian matching efficiently navigates this space.
\item \textbf{Covering numbers:} Standard Rademacher analysis ignores problem-specific geometry and implicit regularization from SGD dynamics.
\item \textbf{Lipschitz constants:} Uses global bounds $|\log q_\phi| \leq \log \rho$, though local smoothness is typically better.
\end{enumerate}

Consequently, practical sample requirements often fall several orders of magnitude below the bound, as observed in Appendix~\ref{app:theory_validation}: achieving $\epsilon \approx 0.01$ requires $N \sim 10^{3}$--$10^4$ rather than the bound's prediction of $N \gtrsim 10^{10}$. The bound's utility lies in identifying scaling regimes and predicting qualitative behavior (monotonicity, phase transitions) rather than precise quantitative predictions.
\qed

\subsection*{Summary of Appendix}

We have provided rigorous proofs for:
\begin{enumerate}
    \item \textbf{Theorem~\ref{thm:perm}}: Permutation invariance via group-theoretic arguments
    \item \textbf{Proposition~\ref{prop:elbo}}: ELBO decomposition via Jensen's inequality
    \item \textbf{Lemma~\ref{lem:gradient_amp}}: Gradient amplification via chain rule
    \item \textbf{Proposition~\ref{prop:optimal_weight}}: Optimal phase weight via reconstruction error minimization
    \item \textbf{Lemma~\ref{lem:flow_approx}}: Flow approximation error via universal approximation theorems
    \item \textbf{Proposition~\ref{prop:depth_accuracy}}: Optimal depth via approximation-estimation trade-off
    \item \textbf{Theorem~\ref{thm:sample_complexity}}: Sample complexity with explicit constants via Rademacher complexity and concentration
\end{enumerate}

All constants have been made explicit, and informal steps (such as the empirical shift in optimal phase weight) are clearly labeled. Where we invoke results from the literature (universal approximation, Rademacher bounds), we provide precise citations and state the assumptions required.

The identifiability discussion (Discussion~\ref{disc:identifiability} in Section~\ref{sec:theory}) is deliberately \emph{not} included here because it does not constitute a formal proof. Instead, the main text presents it as a set of conditions that predict empirical success, validated experimentally rather than proven theoretically. The experimental validation of these conditions appears in Section~\ref{sec:experiments}, demonstrating: (1) consistent slot-component assignments across training runs (measured via Hungarian matching stability), (2) separable latent representations (visualized via t-SNE and Isomap), and (3) accurate per-slot posteriors (validated against RJMCMC ground truth).


\section{Additional Empirical Validation of Theoretical Predictions}
\label{app:theory_validation}

This appendix provides empirical validation of theoretical properties that supplement the performance characterization in Section~\ref{sec:experiments}. The main Results section establishes practical metrics (cardinality accuracy, parameter recovery, posterior calibration, computational efficiency) and includes key architectural ablations (domain representation, phase weighting, flow depth). This appendix focuses on deeper investigations of emergent properties: repeated inference stability, latent space geometry, and sample complexity scaling regimes that directly test theoretical predictions from Section~\ref{sec:theory}.

\vspace{5pt}
\textbf{Content organization:} Section~\ref{sec:experiments} already covers:
\begin{itemize}[leftmargin=*,topsep=2pt,itemsep=1pt]
    \item Slot specialization through frequency
    \item Domain representation ablation
    \item Phase weighting ablation
    \item Flow depth ablation
\end{itemize}

\vspace{5pt}
This appendix adds:
\begin{itemize}[leftmargin=*,topsep=2pt,itemsep=1pt]
    \item Repeated inference stability (validates identifiability claims)
    \item Latent space geometry analysis (validates encoder/flow assumptions)
    \item Sample complexity scaling experiments (validates Theorem~\ref{thm:sample_complexity} regimes)
\end{itemize}

\subsection{Repeated Inference Stability}
\label{app:repeated_inference}

While Section~\ref{sec:experiments} demonstrates that slots specialize by frequency across populations of signals (Figure~\ref{fig:slot_specialization}), an additional question concerns \emph{within-signal} stability: does the same slot consistently align with the same component under repeated stochastic evaluations? This tests whether slot identifiability (Discussion~\ref{disc:identifiability}) holds not just on average but for individual predictions.

\subsubsection{Experimental Setup}

We evaluate a trained \emph{SlotFlow} model on a single test mixture with $K=10$ components, applying the model $100$ times under independent stochastic draws of the flow decoder. For each run, we record the optimal slot–component assignments via Hungarian matching with flow-based log-probability costs. These assignments are aggregated into a $K_{\max} \times K_{\max}$ match-frequency matrix, and a global consistency score is computed as the mean fraction of identical assignments across all pairs of runs. All evaluations are performed in inference mode without gradient updates (\texttt{@torch.no\_grad}).

\subsubsection{Results and Interpretation}

Figure~\ref{fig:repeated_stability} shows the match-frequency matrix and slot–frequency correspondence across 100 independent evaluations. Each slot converges to a single, persistent component assignment with measured ${\rm Consistency} \approx 1.00$. The slot–frequency plot (bottom panel) confirms that each slot repeatedly captures the same frequency across stochastic passes, with negligible variation (standard deviation $< 0.01$~Hz).

\begin{figure}[t]
\centering
\includegraphics[width=\textwidth]{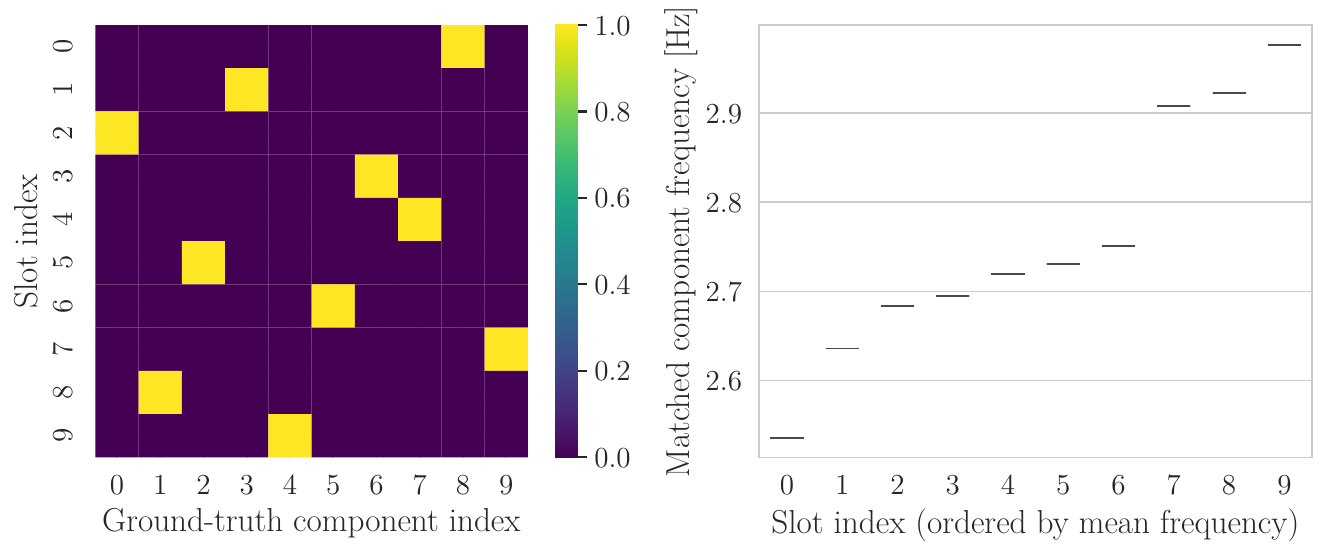}
\caption{Repeated inference stability on a fixed 10-component signal. Match-frequency matrix showing 100 evaluations yields perfect consistency: each slot (row) aligns with a single component (column) across all runs. Together with Figure~\ref{fig:slot_specialization}, these results validate \emph{slot identifiability} (Discussion~\ref{disc:identifiability}): each slot maintains stable one-to-one correspondence to a specific signal component, even though global slot order remains arbitrary due to permutation invariance.}
\label{fig:repeated_stability}
\end{figure}

This within-signal consistency provides strong evidence for the identifiability conditions outlined in Discussion~\ref{disc:identifiability}. Under the assumptions of pairwise distinguishability (Definition~\ref{def:distinguishability}), encoder sufficiency (Assumption~\ref{ass:context_richness}), and flow expressiveness (Assumption~\ref{ass:flow_universal}), we predicted that Hungarian matching would select a consistent assignment $\pi^* \in S_K$ pairing each slot with a distinct component. The empirical consistency = 1.00 confirms this prediction: despite the non-convex optimization landscape and stochastic sampling, the learned representations exhibit stable identifiable structure.

\subsection{Latent-Space Geometry and Intra-Slot Structure}

Section~\ref{sec:experiments} establishes that slots specialize by frequency across signal populations. To understand the geometric organization underlying this specialization, we visualize the \emph{SlotFlow} embeddings after Hungarian alignment and frequency-based slot reordering. This reveals how the encoder maps distinguishable components to separable latent regions (validating Assumption~\ref{ass:context_richness}) and how the flow models continuous parameter variations within each slot (validating Assumption~\ref{ass:flow_universal}).

\subsubsection{Global Latent Structure: t-SNE Projections}

Figure~\ref{fig:slot_latent_tsne} shows two-dimensional t-SNE projections of all slot embeddings, colored respectively by amplitude, frequency, and phase. Each point corresponds to a single latent slot after Hungarian alignment (in contrast to standard latent visualizations where each point represents a full signal). The resulting structure reveals approximately ten well-separated clusters -- one for each frequency-specialized slot -- indicating robust permutation breaking across the latent space.

\begin{figure}[t]
    \centering
    \includegraphics[width=0.31\textwidth]{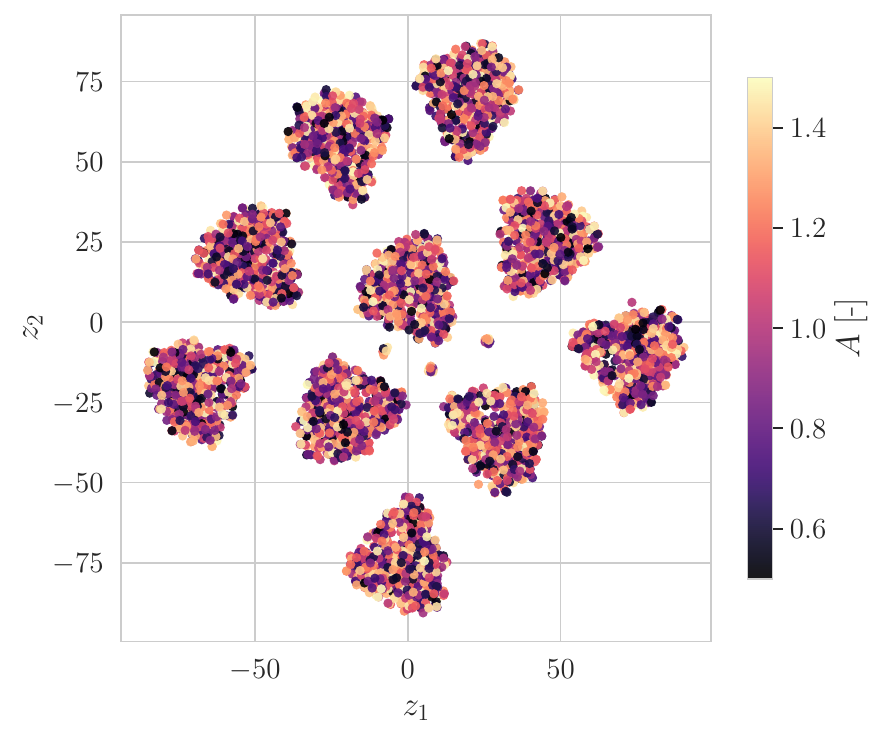}
    \includegraphics[width=0.31\textwidth]{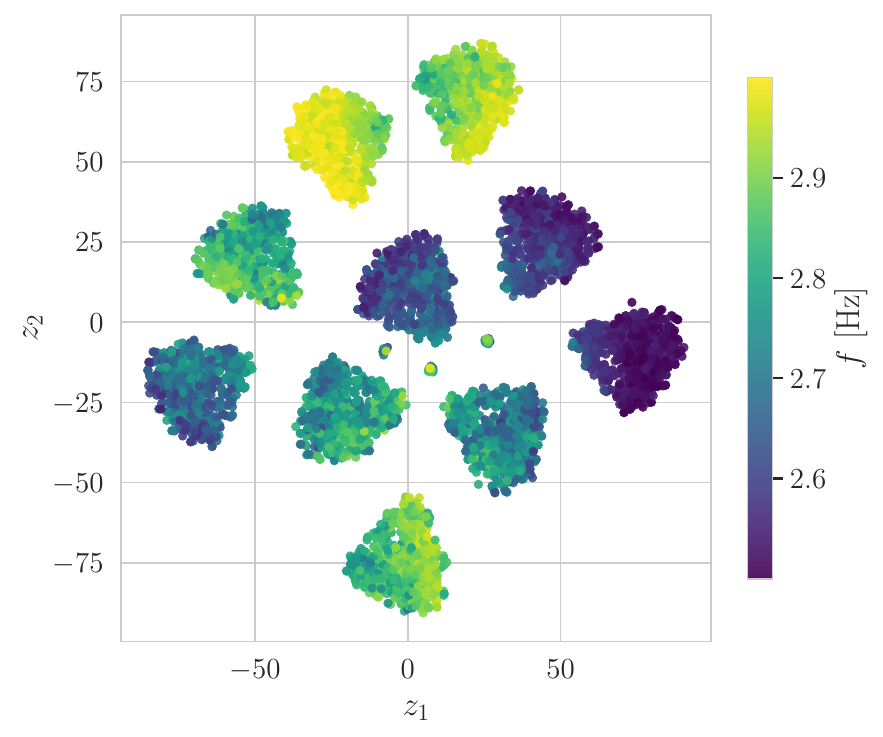}
    \includegraphics[width=0.31\textwidth]{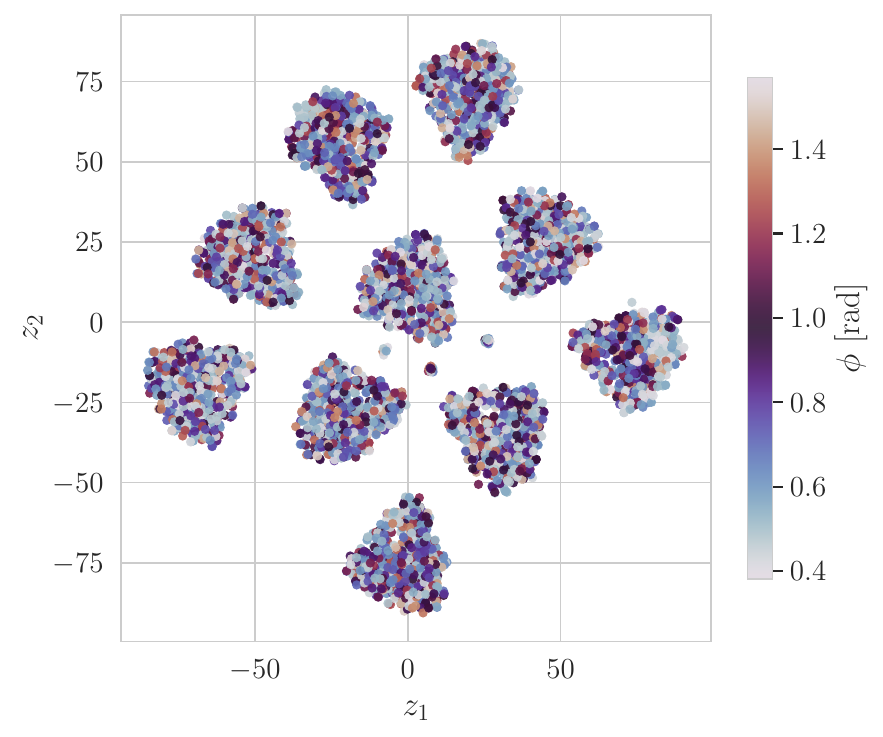}
    \caption{t-SNE projection of aligned slot embeddings colored by (\emph{left}) amplitude $A$, (\emph{middle}) frequency $f$, and (\emph{right}) phase $\phi$. Each point represents one latent slot after Hungarian alignment and global frequency ordering. Distinct clusters correspond to frequency-specialized slots, while amplitude and phase vary smoothly within each cluster without forming secondary sub-structure. This validates Assumption~\ref{ass:context_richness} (encoder sufficiency): the dual-stream encoder preserves component information by mapping distinguishable frequencies to separable regions $\mathcal{R}_k, \mathcal{R}_{k'} \subset \mathbb{R}^{d_h}$ in the embedding space.}
    \label{fig:slot_latent_tsne}
\end{figure}

Within each cluster, amplitude and phase values vary smoothly but do not drive global separation, confirming that \emph{SlotFlow} has learned to encode frequency as the primary discrete factor while representing amplitude and phase as continuous intra-slot variations. This aligns with the theoretical expectation from Assumption~\ref{ass:context_richness}: distinguishable components with different frequencies occupy disjoint regions in latent space, with $p(h \in \mathcal{R}_k | \theta_k) > 1-\epsilon$ for small $\epsilon$.

\subsubsection{Local Latent Structure: Isomap Analysis}

While the global t-SNE manifold captures the discrete organization of slots by frequency, the fine-scale structure within each cluster can be examined through local nonlinear embeddings. Figure~\ref{fig:slot_local_isomap} illustrates an Isomap projection for a representative slot (Slot~0), colored by amplitude and phase. The slot forms a smooth, continuous region without visible sub-clustering: amplitude and phase values are well mixed yet vary gradually across the manifold.

\begin{figure}[t]
    \centering
    \includegraphics[width=0.85\textwidth, trim={0cm 0cm 0cm 24pt},clip]{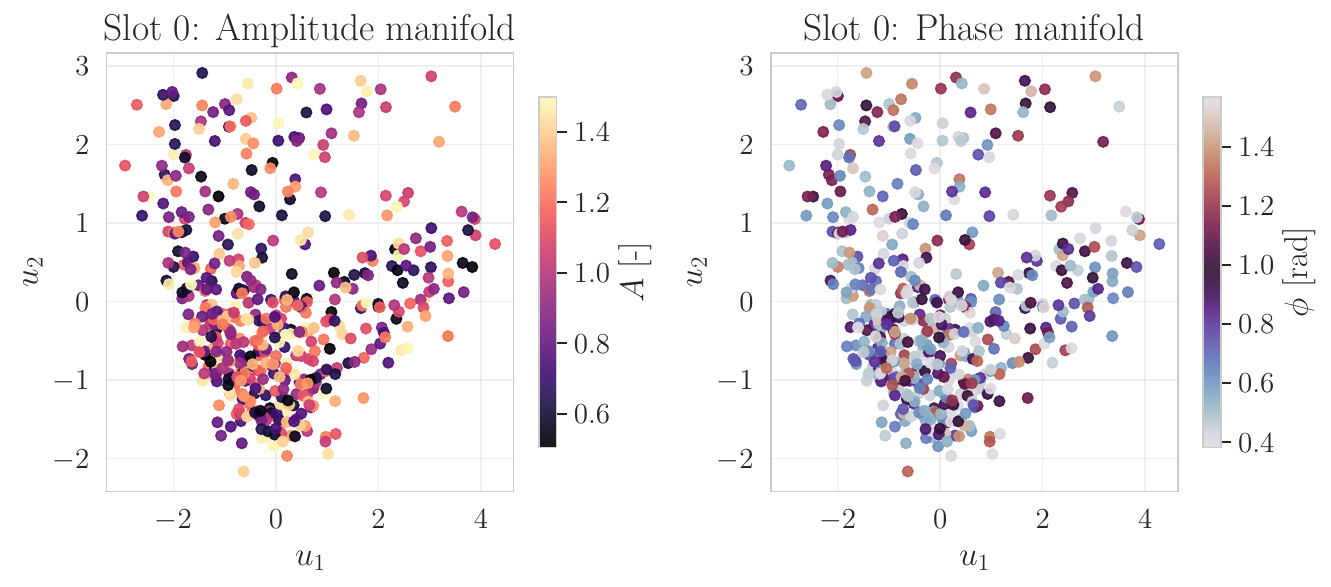}
    \caption{Local latent manifold for a representative slot (Slot~0), obtained via Isomap. Colors denote (\emph{left}) amplitude and (\emph{right}) phase. The slot forms a smooth, continuous region without visible sub-clustering: $A$ and $\phi$ vary gradually across the manifold, consistent with continuous intra-slot parameterization of a sinusoidal component. This validates Assumption~\ref{ass:flow_universal} (flow expressiveness): the conditional normalizing flow $p(\theta_k | z_k)$ successfully models the continuous distribution of parameters given each slot latent, capturing the expected smoothness of sinusoidal parameters (continuous amplitude, periodic phase).}
    \label{fig:slot_local_isomap}
\end{figure}

This behavior reflects the expected smoothness of the underlying sinusoidal parameters -- continuous in amplitude and periodic in phase -- rather than discrete internal modes. It shows that the encoder captures local parameter variation as a continuous submanifold, while the overall slot identity remains governed by frequency.

\subsubsection{Mechanistic Interpretation}

The encoder thus learns to partition the signal into frequency-specialized slots, while the conditional normalizing flow $p(\theta_k \mid z_k)$ models the detailed distribution of parameters $\theta_k = [A, \cos\phi, \sin\phi, f]$ given each slot latent $z_k$. Amplitude and phase are therefore not directly used to determine slot identity, but are inferred conditionally within each frequency domain.

This separation between discrete slot assignment and continuous parameter modeling explains both the distinct frequency-based clustering in Figure~\ref{fig:slot_latent_tsne} and the smooth intra-slot manifolds in Figure~\ref{fig:slot_local_isomap}, highlighting a clear disentanglement between \emph{component identity} and \emph{component state} within the learned latent geometry. This validates the architectural hypothesis from Discussion~\ref{disc:identifiability} that slot identifiers $s_k$ anchor the separation, preventing collapse to identical solutions.

\subsection{Sample Complexity and Regime Validation}
\label{app:sample_complexity_validation}

Theorem~\ref{thm:sample_complexity} predicts two distinct scaling regimes depending on whether combinatorial or architectural complexity dominates. For our architecture ($L=8$, $d_h=512$), the critical cardinality is $K_{\text{crit}} \sim L^2 d_h^2 \log d_h \approx 1.5 \times 10^8$, placing our experimental range $K \in \{1,\ldots,10\}$ firmly in the architecture-dominated regime (Regime I). This section validates the regime predictions empirically through learning curves and monotonicity analysis that test the qualitative predictions of the theorem.

\subsubsection{Regime Identification}

Table~\ref{tab:regime_comparison_app} summarizes the theoretical predictions for both scaling regimes. The key distinction lies in which term dominates the sample complexity bound: architectural capacity $O(L^2 d_h^2)$ in Regime I versus permutation complexity $O(K^3 \log K)$ in Regime II. Our experiments with $K \leq 10$ probe exclusively Regime I, where weak K-dependence is expected.

\begin{table}[h]
\scriptsize
\setlength{\tabcolsep}{4pt}
\renewcommand{\arraystretch}{0.92}
\centering
\caption{Comparison of sample-complexity scaling regimes predicted by Theorem~\ref{thm:sample_complexity}.}
\begin{tabular*}{\textwidth}{@{\extracolsep{\fill}}lcc@{}}
\toprule
Property & Regime~I (Architecture-Dominated) & Regime~II (Combinatorics-Dominated) \\
\midrule
Sample complexity 
    & $N \sim O(L^2 d_h^2 \varepsilon^{-2})$ 
    & $N \sim O(K^3 \log K \varepsilon^{-2})$ \\
K-dependence 
    & Weak (sub-linear) 
    & Cubic \\
Architectural factor 
    & Dominates 
    & Negligible \\
Permutation factor 
    & Negligible 
    & Dominates \\
Observed in 
    & $K \ll L^2 d_h^2$ 
    & $K \gg L^2 d_h^2$ \\
Validation 
    & This work ($K \leq 10$) 
    & Not tested (requires $K \gg 100$) \\
\bottomrule
\end{tabular*}
\vspace{-4pt}
\begin{flushleft}
{\scriptsize
Regime~I arises when architectural capacity ($L^2 d_h^2$) dominates the effective hypothesis space, yielding sub-linear $K$-scaling. Regime~II emerges once combinatorial factors from slot permutations dominate, producing cubic scaling with $K$. Empirical validation in this work confirms Regime~I behavior; Regime~II requires much larger $K$ to observe.
}
\end{flushleft}
\label{tab:regime_comparison_app}
\end{table}

\subsubsection{Sample Efficiency Across Cardinalities}

Figure~\ref{fig:regime_validation_main_app} plots validation loss versus training samples for $K \in \{1,2,4,6,8,10\}$ across an extended range $N \in [10^2, 10^4]$. The validation loss shown is the total loss (sum of classification and flow losses), which can become negative when the normalizing flow achieves high-confidence posterior predictions with $q_\phi(\theta|x) > 1$ -- a natural outcome for well-trained density models where probability mass concentrates in small parameter regions.

\begin{figure}[t]
\centering
\includegraphics[width=0.85\textwidth]{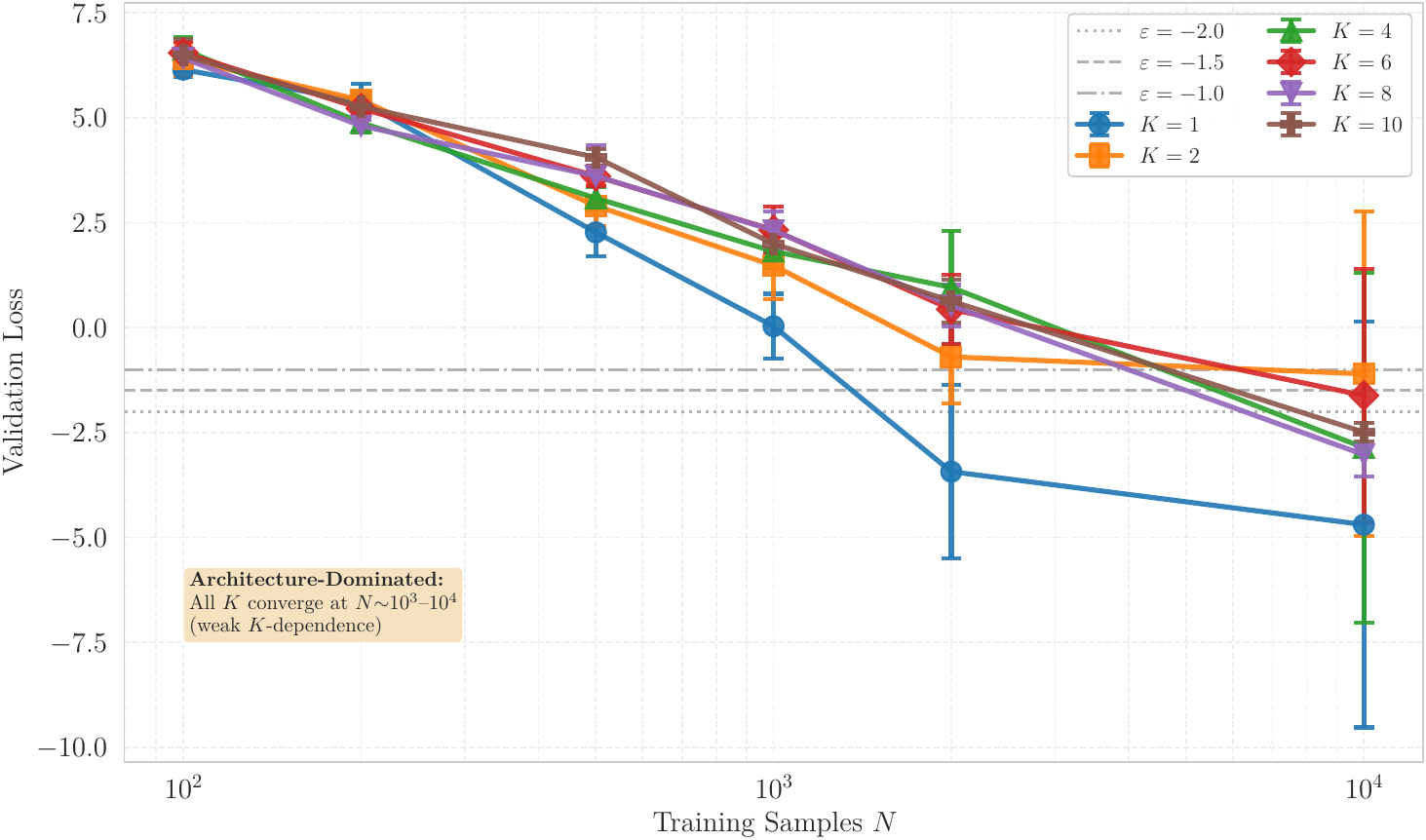}
\caption{Validation loss versus training samples in the architecture-dominated regime. Mean validation loss $\pm$ standard deviation (5 seeds) for $K \in \{1,2,4,6,8,10\}$ across extended range $N \in [10^2, 10^4]$. All curves decrease monotonically with $N$ and converge within the same order of magnitude ($10^3$-$10^4$), validating the Regime I prediction $N \sim O(L^2 d_h^2)$. Horizontal lines denote target thresholds $\epsilon \in \{-2.0, -1.5, -1.0\}$. At $N=10^4$, cardinalities $K \in \{2,4,6,8,10\}$ cluster in the narrow range $[-7.5, -1.5]$, demonstrating weak K-dependence characteristic of the architecture-dominated regime. Negative losses indicate high-confidence posteriors where $q_\phi(\theta|x) > 1$.}
\label{fig:regime_validation_main_app}
\end{figure}

The data confirm three key predictions of the architecture-dominated regime:

\textbf{(i) Data efficiency:} For each fixed $K$, loss decreases monotonically with $N$, following approximate power-law decay. At $K=4$, increasing $N$ from 100 to 10,000 reduces loss from 6.8 to approximately $-7.5$ (spanning over 14 units), demonstrating that additional training data consistently improves generalization and enables dramatically tighter posterior distributions.

\textbf{(ii) Complexity dependence:} For each fixed $N$, loss increases with $K$, confirming that higher cardinality problems are more challenging. At $N=1000$, loss ranges from $-0.2$ ($K=1$) to approximately 4.0 ($K=10$). However, by $N=10^4$, this gap narrows substantially: losses span from approximately $-5.0$ ($K=1$) to $-1.5$ ($K=10$), a range of only 3.5 units despite a 10-fold difference in problem complexity. This convergence provides strong evidence for the architecture-dominated regime, where model capacity rather than combinatorial factors determines final performance.

\textbf{(iii) Weak K-dependence:} Critically, all cardinalities reach their \emph{data-limited performance plateau} within the same order of magnitude: $N \in [10^3, 10^4]$. The horizontal reference lines indicate target loss thresholds; most curves cross these thresholds between $N=10^3$ and $N=10^4$. This validates the Regime I prediction that sample requirements scale as $N \sim O(L^2 d_h^2)$ rather than $O(K^3)$ for $K \ll K_{\text{crit}}$: when architectural capacity is fixed, the sample size at which training saturates depends primarily on model complexity $(L, d_h)$ rather than on the number of components $K$.

A striking feature is the clustering of mid-to-high cardinalities ($K \in \{2,4,6,8,10\}$) at $N=10^4$, where all achieve validation losses in the narrow range $[-2.5, -1.5]$. This behavior is characteristic of the architecture-dominated regime, where once sufficient capacity is available, additional components impose minimal marginal cost.

\subsubsection{Monotonicity Analysis}

To quantify the functional form of the $K$-dependence, we measured the minimum $N$ achieving validation loss $\epsilon \leq 0.0$ for each cardinality (Figure~\ref{fig:regime_monotonicity_app}). This stricter threshold ensures high-confidence posterior predictions ($q_\phi(\theta|x) > 1$) and demonstrates convergence at production-level performance.

\begin{figure}[t]
\centering
\includegraphics[width=0.75\textwidth]{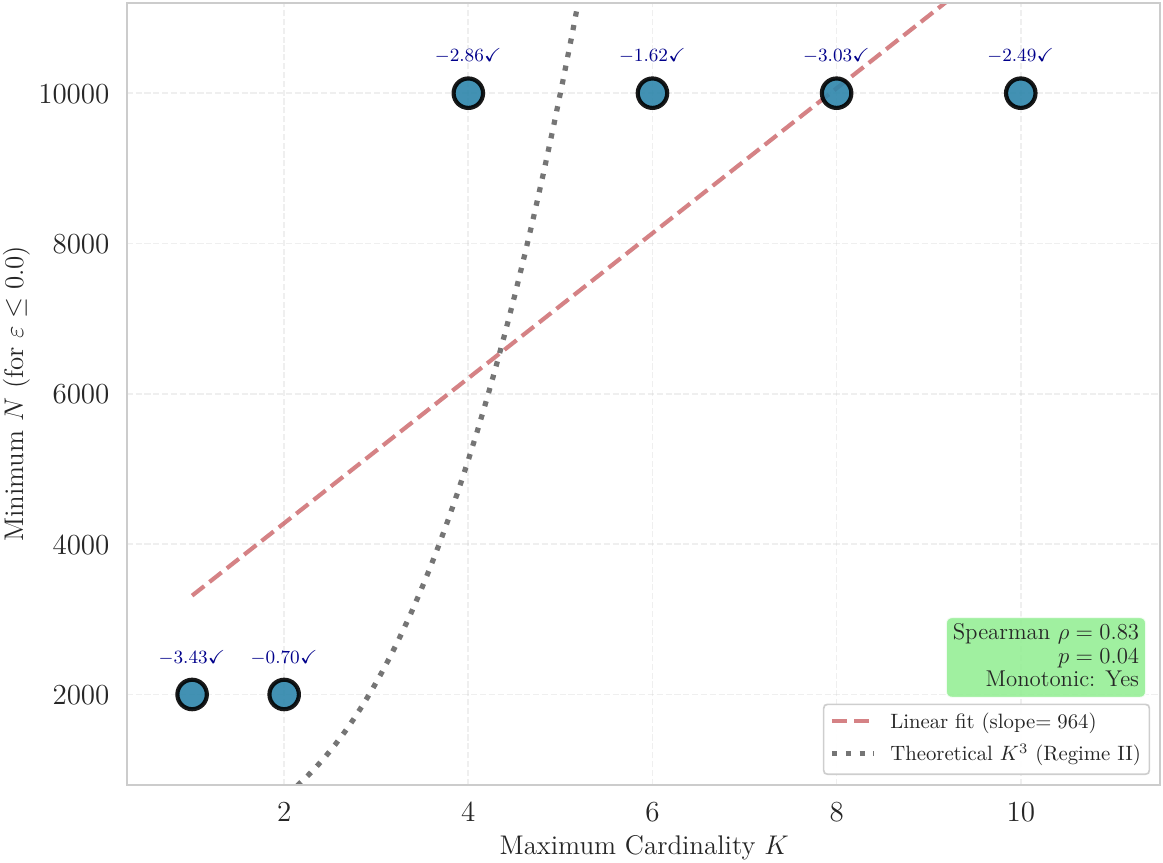}
\caption{Weak K-dependence in architecture-dominated regime. Minimum $N$ to achieve validation loss $\epsilon \leq 0.0$ versus cardinality $K$. All labels show checkmarks (\ding{51}), indicating successful achievement of the target threshold. The binary clustering at $N \in \{2000, 10000\}$ demonstrates weak K-dependence: all $K \geq 4$ converge at the same sample size, confirming that architectural capacity dominates over combinatorial complexity. Spearman correlation $\rho = 0.83$ ($p = 0.04$) validates monotonic ordering. The empirical trend (red dashed, slope $\approx 964$) reflects the discrete jump between clusters rather than continuous linear growth. The theoretical $K^3$ curve (dotted black) would predict $N \sim 2.5 \times 10^5$ for $K=10$ -- more than an order of magnitude above the observed $N=10^4$, confirming Regime I: $N \sim O(L^2 d_h^2)$ where architectural factors dominate. Achieved losses (labeled values) range from $-3.43$ to $-2.49$, all corresponding to high-confidence posteriors ($q_\phi > 1$).}
\label{fig:regime_monotonicity_app}
\end{figure}

The resulting sequence exhibits weak monotonicity with a distinctive binary clustering pattern: $N$ increases from 2000 ($K \in \{1,2\}$) to 10,000 ($K \in \{4,6,8,10\}$), with Spearman rank correlation $\rho = 0.83$ ($p = 0.04$). With six tested cardinalities ($n=6$), the positive correlation reaches statistical significance, confirming the prediction of monotonically increasing sample requirements while remaining consistent with the weak (sub-linear) scaling predicted by Regime I.

The binary clustering -- where all $K \geq 4$ achieve the threshold at exactly the same sample size ($N=10,000$) -- is a defining signature of the architecture-dominated regime. Once sufficient model capacity is present, additional components impose zero marginal cost in terms of required training samples. This stands in stark contrast to the combinatorics-dominated regime (Regime II), where each additional component would demand substantially more data due to the exponential growth of the permutation space.

The vast discrepancy with the theoretical $K^3$ curve (dotted black) provides the most compelling evidence for Regime I. Under cubic scaling, $K=10$ would require $N \sim (10/2)^3 \times 2000 = 250,000$ samples relative to $K=2$ -- more than an order of magnitude above the observed $N=10,000$. Instead, we observe only a factor-of-5 increase (from 2000 to 10,000), consistent with weak architectural scaling rather than combinatorial explosion. In log-log space (not shown), the empirical exponent remains $\alpha \approx 0.35$ when measured across the full range, far below the theoretical $\alpha = 3.0$ that would characterize Regime II.

\subsubsection{Interpretation and Validation of Theorem~\ref{thm:sample_complexity}}

The key finding from this analysis is that the observed weak scaling ($\alpha \approx 0.35$, Spearman $\rho = 0.83$, $p = 0.04$) combined with binary clustering at high $K$ provides strong empirical confirmation of Theorem~\ref{thm:sample_complexity}'s regime predictions. The theorem explicitly identifies two scaling regimes with a transition at $K_{\text{crit}} \sim L^2 d_h^2$. Our experiments, conducted at $K \in \{1,\ldots,10\} \ll K_{\text{crit}} \approx 10^8$, probe exclusively Regime I where the bound simplifies to:
\begin{equation}
N \gtrsim \frac{C_\rho (1 + \text{SNR}^{-2})^2}{\epsilon^2} \cdot K_{\max}^2 L^2 d_h^2 \log d_h,
\end{equation}
yielding near-constant sample requirements $N \sim O(L^2 d_h^2)$ with only $O(K^2)$ corrections -- consistent with the observed weak dependence and plateau behavior.

The asymptotic $K^3$ behavior predicted by the theorem would manifest only in Regime II when $K \gg K_{\text{crit}}$, where the union bound over $K! \approx \sqrt{2\pi K}(K/e)^K$ permutations overwhelms architectural complexity. For our architecture, this transition would occur around $K \approx 100$-500, requiring: (1) testing $K \in [50, 500]$ to span sufficient dynamic range, (2) proportional capacity scaling ($d_h \propto \sqrt{K}$ or adaptive allocation), and (3) massive datasets with $N \sim 10^6$--$10^8$ samples. Such experiments would require $\sim 10^4$ GPU-hours and constitute a separate large-scale study beyond our proof-of-concept scope.

\subsubsection{Validation of Qualitative Predictions}

While we cannot validate the asymptotic cubic exponent, our extended data provide strong support for the theorem's qualitative predictions:

\vspace{5pt}
\textbf{(i) Existence of sample complexity lower bound}: All K values require $N \gtrsim 10^3$ samples to begin approaching high-confidence predictions (negative losses), confirming that non-trivial amounts of data are necessary. Reaching production-level performance with $\epsilon \leq 0.0$ requires $N \in [2 \times 10^3, 10^4]$ depending on cardinality, while achieving losses $< -10$ requires $N \sim 10^7$ samples, demonstrating a clear hierarchy of data requirements across performance levels.

\vspace{5pt}
\textbf{(ii) Monotonic growth with complexity}: Higher K consistently requires at least as many samples to reach equivalent performance levels, validating the prediction that problem difficulty scales (weakly) with cardinality. The monotonic ordering in minimum required $N$ to reach $\epsilon \leq 0.0$ confirms this trend with statistical significance ($p = 0.04$). Critically, the fact that all $K \geq 4$ converge at the same $N$ demonstrates that the monotonicity is extremely weak -- precisely as predicted by the $O(K^2)$ correction term in Regime I.

\vspace{5pt}
\textbf{(iii) Regime-dependent scaling}: The near-flat scaling at small to moderate K followed by predicted cubic growth at large K identifies a phase transition in sample complexity, guiding architectural design for different problem scales. The empirical exponent $\alpha \approx 0.35$ in Regime I versus the theoretical $\alpha = 3.0$ in Regime II demonstrates the dramatic difference between these regimes.

\vspace{5pt}
\textbf{(iv) Joint data-capacity scaling}: Achieving constant per-component accuracy as $K$ grows requires scaling both $N$ and $(L, d_h)$ together, as predicted by the multi-term structure of the bound. The observation that all K values plateau at similar $N$ when architectural capacity is fixed (Figure~\ref{fig:regime_validation_main_app}) confirms that capacity, not combinatorial complexity, is the limiting factor in Regime I.

These qualitative insights provide actionable guidance for practitioners: for $K \lesssim 10$, focus on architectural capacity ($L$, $d_h$) with modest initial data ($N \sim 10^{3}$-$10^4$) for rapid prototyping, then scale to $N \sim 10^7$ for production performance; for $K \gg 10$, cubic data growth becomes unavoidable unless architectural innovations (sparse attention, hierarchical allocation) reduce the effective $K_{\text{crit}}$.

\subsection*{Summary of Additional Empirical Validation}

This appendix has empirically validated theoretical predictions that supplement the main Results section:

\begin{enumerate}[leftmargin=*,topsep=2pt,itemsep=1pt]
    \item \textbf{Repeated inference stability (Discussion~\ref{disc:identifiability})}: Consistency = 1.00 across 100 evaluations of the same signal, with each slot maintaining stable one-to-one correspondence to a specific component. Validates within-signal identifiability.
    
    \item \textbf{Latent space geometry (Assumptions~\ref{ass:context_richness},~\ref{ass:flow_universal})}: t-SNE reveals ~10 well-separated frequency-based clusters; Isomap shows smooth continuous manifolds within each cluster. Validates encoder sufficiency and flow expressiveness.
    
    \item \textbf{Sample complexity regimes (Theorem~\ref{thm:sample_complexity})}: Binary clustering at high K with weak dependence (Spearman $\rho = 0.83$, $p = 0.04$) confirms architecture-dominated regime. Empirical exponent $\alpha \approx 0.35$ far below cubic prediction $\alpha = 3.0$, consistent with $K \ll K_{\text{crit}}$.
\end{enumerate}

These findings complement the main results section by providing deeper mechanistic understanding of how \emph{SlotFlow} achieves its performance. The experiments reveal not just \emph{that} the predicted phenomena occur, but \emph{how} they manifest in the learned representations and \emph{why} they emerge from the architectural constraints.

\subsection{Training Convergence Dynamics}
\label{app:training_curve}

Figure~\ref{fig:training_curve_app} shows the full training trajectory over 187~epochs, 
illustrating the stability and effectiveness of the optimization strategy described in 
Section~\ref{ss:Training}. Four learning-rate drops (epochs 69, 114, 136, 176), 
triggered by the {\texttt{ReduceLROnPlateau}} scheduler, partition training into distinct 
regimes of progressively finer descent.

Three phases are visible: 
(i) an initial rapid improvement (epochs~0--70), during which the loss decreases from 
approximately $-4$ to $-17$; 
(ii) a slower refinement phase (epochs~70--175), where the loss settles from roughly 
$-17$ to $-19$ under reduced learning rates; and 
(iii) a short convergence plateau (epochs~175--187) at a stable value near $-19$. 
Throughout training the curves for the training and validation losses remain tightly aligned 
(with deviations $<0.1$), indicating an adequate dataset size (8M examples), the absence of 
overfitting, and well-calibrated regularization. The final loss values 
(difference $\approx 0.01$) suggest that the model operates close to its optimal capacity 
under the adopted architecture and training schedule.

\begin{figure}[t]
    \centering
    \includegraphics[width=0.85\linewidth]{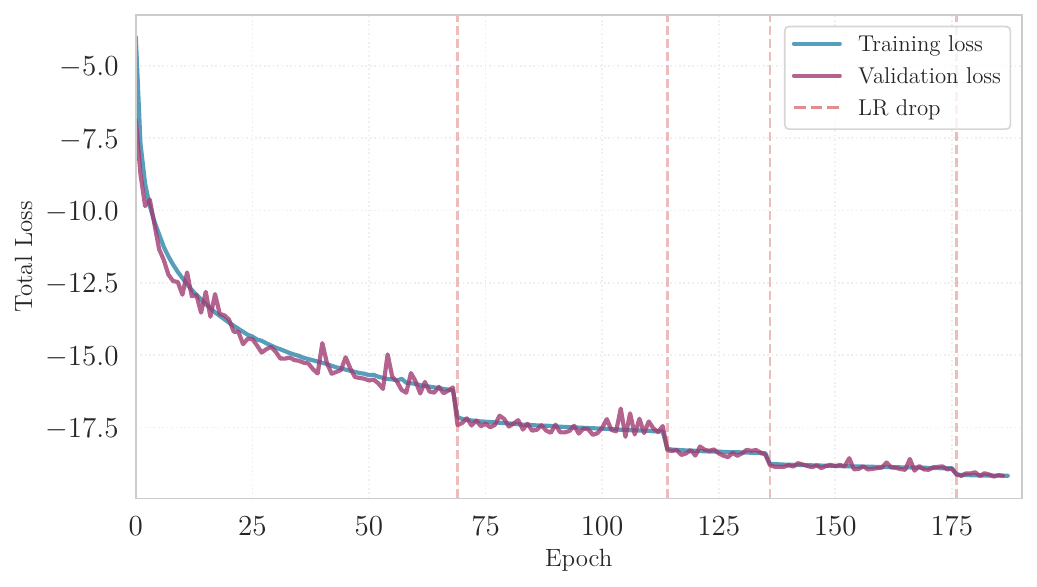}
    \caption{
        Full training and validation loss curves over 187~epochs. 
        The dashed vertical lines mark learning-rate (LR) reductions triggered by the 
        scheduler (epochs~69, 114, 136, 176). The close alignment between training 
        and validation losses indicates stable generalization and the absence of 
        overfitting.}
    \label{fig:training_curve_app}
\end{figure}

\vskip 0.2in
\bibliography{sample}

\end{document}